\documentclass[english]{article}
\usepackage[frozencache,cachedir=.]{minted} 
\usepackage{mathpazo}
\usepackage[T1]{fontenc}
\usepackage[latin9]{inputenc}
\usepackage{geometry}
\usepackage{babel, graphicx}
\usepackage{verbatim}
\usepackage{varioref}
\usepackage{float}
\usepackage{mathtools}
\usepackage{amsmath}
\usepackage{amsthm}
\usepackage{amssymb}
\usepackage{stmaryrd}
\usepackage[unicode=true,pdfusetitle,
 bookmarks=true,bookmarksnumbered=false,bookmarksopen=false,
 breaklinks=false,pdfborder={0 0 0},pdfborderstyle={},backref=false,colorlinks=false]
 {hyperref}

\usepackage{thmtools}
\usepackage{thm-restate}

\usepackage{tikz,circuitikz}
\usetikzlibrary{arrows.meta}
\usetikzlibrary{backgrounds}
\tikzstyle{background rectangle}=[fill=gray!10]
\ctikzset{bipoles/thickness=1} 
\usepackage{pifont}

\newcommand*\circled[1]{\tikz[baseline=(char.base)]{
            \node[shape=circle,draw,inner sep=2pt,fill=gray!25] (char) {#1};}}

\newcommand*\squared[1]{\tikz[baseline=(char.base)]{
            \node[shape=rectangle,draw,inner sep=3.5pt,fill=gray!25] (char) {#1};}}

\newcommand*\squaredr[1]{\tikz[baseline=(char.base)]{
            \node[shape=rectangle,draw,inner sep=3.5pt,fill=red!25] (char) {#1};}}
            
\newcommand*\squaredg[1]{\tikz[baseline=(char.base)]{
            \node[shape=rectangle,draw,inner sep=3.5pt,fill=green!25] (char) {#1};}}

\makeatletter
\theoremstyle{definition}
\newtheorem{defn}{\protect\definitionname}
\theoremstyle{plain}
\newtheorem{thm}{\protect\theoremname}
\theoremstyle{plain}
\newtheorem{cor}{\protect\corollaryname}
\theoremstyle{plain}
\newtheorem{prop}{\protect\propositionname}
\theoremstyle{plain}
\newtheorem{lem}{\protect\lemmaname}
\theoremstyle{definition}
\newtheorem{example}{Example}
\newtheorem{notat}{Notation}
\newtheorem{remark}{Remark}

\addto\captionsenglish{
  \renewcommand{\contentsname}%
    {
	 }%
}

\usepackage{tikz}
\usetikzlibrary{arrows,positioning,fit,automata}
\tikzset{    
>=stealth',
empty/.style={rectangle, rounded corners, text width=3.5em, minimum height=2em, text centered},
world/.style={rectangle, draw=black, rounded corners, text width=70pt, minimum height=19pt, text centered},
real/.style={rectangle, draw=black, ultra thick, rounded corners, text width=40pt, minimum height=19pt, text centered},
pil/.style={->, thick, shorten <=2pt, shorten >=2pt,}, 
every loop/.style={max distance=10mm,looseness=10},
}

\makeatother

\providecommand{\corollaryname}{Corollary}
\providecommand{\definitionname}{Definition}
\providecommand{\lemmaname}{Lemma}
\providecommand{\propositionname}{Proposition}
\providecommand{\theoremname}{Theorem}
\DeclareMathOperator{\dom}{dom}
\def\pre{\mathsf{pre}}
\def\post{\mathsf{post}}

\begin{document}
\title{Dynamic Term-Modal Logics for \\First-Order Epistemic Planning}
\author{
Andr{\'e}s Occhipinti Liberman$^1$, Andreas Achen$^2$ and Rasmus K. Rendsvig$^3$\\ 
$^1$\small{DTU Compute, Technical University of Denmark}\\ 
$^2$\small{London School of Economics}\\ 
$^3$\small{Center for Information and Bubble Studies, University of Copenhagen. \texttt{rasmus@hum.ku.dk}}\\
}
\date{}
\maketitle

\begin{center}
This version is penultimate. The published version appears here:\\
\href{https://doi.org/10.1016/j.artint.2020.103305}{https://doi.org/10.1016/j.artint.2020.103305}
\par\end{center}

\begin{abstract}
\noindent Many classical planning frameworks are built on first-order languages. The first-order expressive power is desirable for compactly representing actions via schemas, and for specifying quantified conditions such as $\neg\exists x\mathsf{blocks\_door}(x)$. In contrast, several recent epistemic planning frameworks are built on propositional epistemic logic. The epistemic language is useful to describe planning problems involving higher-order reasoning or epistemic goals such as $K_{a}\neg\mathsf{problem}$.

This paper develops a first-order version of Dynamic Epistemic Logic (DEL). In this framework, for example, $\exists xK_{x}\exists y\mathsf{blocks\_door}(y)$ is a formula. The formalism combines the strengths of DEL (higher-order reasoning) with those of first-order logic (lifted representation) to model multi-agent epistemic planning. The paper introduces an epistemic language with a possible-worlds semantics, followed by novel dynamics given by first-order action models and their execution via product updates. Taking advantage of the first-order machinery, epistemic action schemas are defined to provide compact, problem-independent domain descriptions, in the spirit of PDDL. 
 
Concerning metatheory, the paper defines axiomatic normal term-modal logics, shows a Canonical Model Theorem-like result which allows establishing completeness through frame characterization formulas, shows decidability for the finite agent case, and shows a general completeness result for the dynamic extension by reduction axioms.

    
\end{abstract}{}

\section{Introduction}

Most classical planning languages are first-order. Standard formalisms like PDDL \cite{mcdermott1998pddl} and ADL \cite{pednault1987formulating}, for example, are first-order. One major reason for using a first-order language over a propositional one is that variables can be used to describe actions compactly. For instance, in the PDDL description of BlocksWorld, the \textit{action schema} $\mathsf{stack}(X,Y)$ uses variables $X$ and $Y$ to represent generic blocks and state the preconditions and effects of all actions of the form: ``put block $X$ on top of block $Y$''. This is possible because the action of stacking block $A$ on block $B$ has the same type of effects as the action of stacking block $C$ on $D$; only the names of the blocks are different. Action schemas use variables to exploit this repeated structure in actions, resulting in action representations whose size is independent of the number of objects in a domain. While $\mathsf{stack}(X,Y)$ describes the preconditions and effects of performing a stack action on any two blocks, regardless of total number of blocks, with a propositional language each stack action has to be represented by a separate model, yielding $n^2-n$ propositional models for a domain with $n$ blocks. Generally, given an action schema with $k$ variables and $n$ constant symbols standing for domain objects, the schema has up to $n^k$ different instantiations, each requiring a separate model in a propositional representation.

\textit{Dynamic Epistemic Logic} (DEL) has proved to be a very expressive framework for \textit{epistemic planning}, i.e., planning explicitly involving e.g. knowledge or belief. DEL uses the language of propositional epistemic logic to describe the knowledge or belief held by a community of agents. This language is built from a set of propositional atoms, standard logical connectives, and modal operators $K_i$ for each agent $i$ in a fixed set of agent indices $I=\{1,\dots,n\}$. An example of a formula is $K_1 p \wedge K_2 K_1 p$, which expresses that agent 1 knows the propositional atom $p$ and that agent 2 knows that agent 1 knows $p$. Actions in DEL are described by so-called \textit{action models} \cite{baltagmoss2004,BaltagBMS_1998} or variants thereof. Action models describe preconditions and effects of events, and provide a rich framework for representing the agents' uncertainty about such events. However, as action models are based on the propositional epistemic language, propositional DEL cannot achieve the generality of action schemas. Variabilized, general descriptions are not possible, so one action model is required for each action. Hence, while propositional DEL adds a great deal of expressivity to planning, this comes at a cost in terms of representational succinctness.

This paper presents a DEL-based epistemic planning framework built on \textit{epistemic term-modal logic}. The underlying language is first-order and includes modalities indexed by first-order terms.  Examples of formulas include $K_c \:  \mathsf{on}(A,B)$ (agent $c$ knows that block $A$ is on block $B$), $K_c \exists x \: \mathsf{on}(x,B)$ ($c$ knows that there is a block on top of $B$) and $\forall y K_y \exists x \: \mathsf{on}(x,B)$ (all agents know that there is a block on top of $B$). Term-modal languages thus extend the expressive power of first-order modal languages by treating modal operators \textit{both} as operators \textit{and} as predicates. 

In addition to higher-order knowledge expressions, the first-order apparatus of epistemic term-modal logic allows for domain descriptions in terms of objects and relations, as well as abstract reasoning via variables and quantification. The term-modal aspect ensures that these first-order aspects also extend to agents and their knowledge. Importantly, the presence of variables enables the definition of \textit{epistemic action schemas}. Epistemic action schemas can be exponentially more succinct than standard DEL event models (see Section \ref{subsec:action_schemas}). Moreover, epistemic action schemas provide an action representation whose size is independent of the number of agents and objects in the domain. We consider the development of this epistemic planning framework our first main contribution.

Our second main contribution is the development of term-modal logic, its dynamic extension and the metatheory for both. Many papers have been dedicated to term-modal logic and its metatheory (see Section \ref{subsec:TML.litt.} for a detailed review), but due to the many complications that may arise in such generalized first-order modal systems, no general completeness results have been shown. In this paper, we define a rich but well-behaved semantics that allow us to define axiomatic normal term-modal logics and show a Canonical Model Theorem-like result that allow completeness results through frame characterization formulas. Adding reduction axioms to the term-modal logics then allow us to show completeness for the dynamic extension. 

The paper progresses as follows. Section \ref{sec:running_example} presents \textit{SelectiveCommunication} used as running example of epistemic planning. Section \ref{sec:static_language} presents term-modal logical languages and Section \ref{sec:state.rep} defines state representations: first-order Kripke models where the agent set is a part of the domain of quantification. Section \ref{sec:action.rep} introduces action representations (action models) and how these may be succinctly represented as action schemas. The action representations are used in Section \ref{sec:planning} to define epistemic planning problems and related notions, and an example describing a term-modal planning domain and problem using a `PDDL-like syntax' is there given. Section \ref{sec:act.lang} details how to extend the term-modal language to allows reasoning about actions, 
Section \ref{sec:metatheory} turns to axiomatic systems and metatheory, while Section \ref{sec:related.work} turns to related work on epistemic planning, dynamic epistemic logic and term-modal logic, respectively. Section \ref{sec:final.remarks} contains final remarks and open questions. All proofs may be found in Appendix \ref{A.proof.appendix}.

\section{A Running Example \label{sec:running_example}} 
Throughout the paper, we illustrate the planning formalism with a simple running example in a variant of the \textit{SelectiveCommunication} (SC) domain, adapted from \cite{kominis2015beliefs}. Here we describe it informally, but it will serve as an example for the various formal notions throughout the paper. In the $\text{SC}(n,m,k,\ell)$ domain, there are $n$ agents. Each agent is initially in one of $m$ rooms arranged in a corridor. There are $k$ boxes distributed in the rooms, each having one of $\ell$ available colors. See Figure \ref{fig:SCrooms} for an example. Agents can perform four types of actions:

\begin{itemize}
    \item $\text{Move}(agent,room1,room2)$: agents can move from a room to a contiguous room, by going right or left. In this adaptation of the domain, we model the move actions as partially observable: if agent $\alpha$ is in room $\rho_i$ and moves to room $\rho_{j}$, only the agents in either of the rooms can see that $\alpha$'s location has changed. 
    \item  $\text{SenseLoc}(agent,agent\_or\_box,room)$: while in a room, agents can sense the location of other agents or boxes in that room. Other agents in the room notice the sensing action.
    \item $\text{SenseCol}(agent,color,box,room)$: agents can sense the color of a box when they are in the same room as the box. Other agents in the room notice the sensing action.
    \item $\text{Announce}(agent,color,box,room)$: agents can make announcements concerning the colors of boxes. If $\alpha$ makes an announcement in a room, all agents in the same room or in a contiguous room will hear what was announced. $\alpha$ can use announcements to ensure that some agents get to know the truth value of some $\varphi$ while the remaining agents do not. 
\end{itemize}

A specific choice for the parameters $n,m,k,\ell$ yields an \textit{instance} of the SC domain. For example, $\text{SC}(3,4,1,2)$ is the instance of \textit{SelectiveCommunication} involving three agents ($\alpha_1$, $\alpha_2$ and $\alpha_3$), four rooms  ($\rho_1$, $\rho_2$, $\rho_3$ and $\rho_4$), one box ($\beta_1$) and two possible colors for the box (e.g., red and green). Figure \ref{fig:SCrooms} depicts a \textit{possible state} of the environment in this domain.

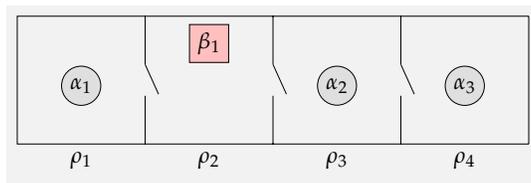
\begin{figure}[h!] 
  \begin{center}
      \scalebox{0.85}{
        \begin{circuitikz}[framed] \draw
        (1,0) node[below] {$\rho_1$}
        (3,0) node[below] {$\rho_2$}
        (5,0) node[below] {$\rho_3$}
        (7,0) node[below] {$\rho_4$}
        (2,0) -- (0,0) -- (0,2) -- (2,2)
        (2,2) to[nos] (2,0)
        (1,1.35) node[below] {\circled{$\alpha_1$}}
        (2,2) -- (4,2) 
        (2,0) -- (4,0)
        (4,2) to[nos] (4,0)
        (3,2) node[below] {\squaredr{$\beta_1$}}
        (4,2) -- (6,2) 
        (4,0) -- (6,0)
        (6,2) to[nos] (6,0)
        (5,1.35) node[below] {\circled{$\alpha_2$}}
        (6,2) -- (8,2) 
        (6,0) -- (8,0)
        (8,2) -- (8,0)
        (7,1.35) node[below] {\circled{$\alpha_3$}}
        ;
        \end{circuitikz}
        }
  \end{center}
  \caption{\label{fig:SCrooms}A depiction of a possible state in $\text{SC}(3,4,1,2)$, where a red box $\beta_1$ is in room $\rho_2$.}
\end{figure}

A possible \textit{goal} $g = g_1 \wedge g_2 \wedge g_3$ in this domain is given by the conjunction of the following subgoals: 
\begin{itemize}
    \item $g_1$: $\alpha_1$ and $\alpha_2$ know the color of $\beta_1$
    \item $g_2$: $\alpha_1$ knows that $\alpha_2$ knows the color of $\beta_1$
    \item $g_3$: $\alpha_1$ knows that $\alpha_3$ does not know the color of $\beta_1$
\end{itemize}
That is, $g$ requires $\alpha_1$ to learn the color of $\beta_1$ and privately communicate this information to $\alpha_2$. I.e., the goal is epistemic; it requires $\alpha_1$ to achieve \textit{first-order knowledge} about the environment ($g_1$) as well as \textit{higher-order knowledge} about what others know ($g_2$ and $g_3$).

The nature of a \textit{plan} for achieving $g$ depends on the \textit{initial state} as well as the assumptions made about the planning problem. For illustrative purposes, we consider a simple problem. Suppose that only $\alpha_1$ can act and that the initial state $s_0$ satisfies the following conditions:
\begin{itemize}
    \item $c_1$: each agent knows the location of all agents and the box $\beta_1$.
    \item $c_2$: no agent knows the color of $\beta_1$ (which is in fact red).
    \item $c_3$: conditions $c_1$ and $c_2$ are common knowledge among the agents.
\end{itemize}
In this case, $\alpha_1$ can easily reach a state satisfying $g$ from $s_0$. The following plan achieves $g$: \textit{$\alpha_1$ moves to $\rho_2$, $\alpha_1$ senses the color of $\beta_1$, $\alpha_1$ announces the color of $\beta_1$}. Of course, more initial uncertainty, or allowing other agents to act (sequentially or in parallel), results in more complex tasks. Such tasks can be defined with the formalism presented in this paper; however, for a first take on the formalism, this toy problem will be considered. 

For additional examples, in \cite{LibermanRendsvig2019} we use the framework to model social networks with epistemic dynamics.

\section{Language} \label{sec:static_language}

As term-modal logical languages include first-order aspects, they
are parameterized by a \textit{signature} specifying the \textit{non-logical
symbols} and their \textit{type}\textemdash i.e., the constants and
relation- and function symbols, and the sort and order of arguments
(agent or object) they apply to. Also variables are here assigned
a type.%

\begin{notat}
For a vector $v=(x_{1},...,x_{n})$, let $len(v)$ denote its length, let $v_{i}$ denote its $i$'th element, i.e., $v_{i}\coloneqq x_{i}$, and let $v_{|k}$ denote its restriction to its prefix of length $k$, i.e., $v_{|k}\coloneqq(x_{1},...,x_{k})$.
\end{notat}

\begin{defn}
A \textbf{\textit{signature}} is a tuple $\Sigma=(\mathtt{V},\mathtt{C},\mathtt{R},\mathtt{F},\mathtt{t})$ with $\mathtt{V}$ a countably infinite set of variables, and $\mathtt{C}$, $\mathtt{R}$ and $\mathtt{F}$ countable sets of respectively constants, relation symbols and function symbols with the one requirement that $\{=\}\subseteq\mathtt{R}$. Finally, $\mathtt{t}$ is a type
assignment map that satisfies 
\begin{enumerate}
    \item For $x\in\mathtt{V}$, $\mathtt{t}(x)\in\{\mathtt{agt},\mathtt{obj}\}$ such that both $\mathtt{V}\cap\mathtt{t}^{-1}(\mathtt{agt})$ and $\mathtt{V}\cap\mathtt{t}^{-1}(\mathtt{obj})$ are countably infinite. 
    \item For $c\in\mathtt{C}$, $\mathtt{t}(c)\in\{\mathtt{agt},\mathtt{obj}\}$.
    \item For $r\in\mathtt{R}$,
    \begin{enumerate}
        \item for some $n\in\mathbb{N}$, $\mathtt{t}(r)\in\{\mathtt{agt},\mathtt{obj},\mathtt{agt\_or\_obj}\}^{n}$, and 
        \item for $=\,\in\mathtt{R}$, $\mathtt{t}(=)=(\mathtt{agt\_or\_obj},\mathtt{agt\_or\_obj})$.
    \end{enumerate}
    \item For $f\in\mathtt{F}$, 
    \begin{enumerate}
        \item for some $n\in\mathbb{N}$, $\mathtt{t}(f)\in\{\mathtt{agt},\mathtt{obj},\mathtt{agt\_or\_obj}\}^{n}\times\{\mathtt{agt},\mathtt{obj}\}$, and 
        \item if $\mathtt{t}(f)_{|n}=(\mathtt{t}(t_{1}),...,\mathtt{t}(t_{n}))$, then $\mathtt{t}(f(t_{1},...,t_{n}))=\mathtt{t}(f)_{n+1}$.
    \end{enumerate}
 
\end{enumerate}
\end{defn}
Identity is treated as a relation symbol; for it, infix notation is
used with $=(a,b)$ written $a=b$. 

\begin{example}[Signature for SelectiveCommunication] \label{example:language}  
The following signature $\Sigma=(\mathtt{V},\mathtt{C},\mathtt{R},\mathtt{F},\mathtt{t})$ can be used to specify the $\text{SC}(n,m,k,\ell)$ domain introduced in Section \ref{sec:running_example}:
\begin{itemize}
    \item Variables $\mathtt{V}= \{x^\star,x,y,z,x_1,x_2,x_3,\dots\}$.
    \item Constants $\mathtt{C}= Agents_{con} \cup Rooms_{con} \cup Boxes_{con} \cup Colors_{con}$, where $Agents_{con}=\{a_1,\dots,a_n\}$, $Rooms_{con}=\{r_1,\dots, r_m\}$ $Boxes_{con}=\{b_1,\dots, b_k\}$ and $Colors_{con}=\{c_1, \dots, c_\ell\}$.
    \item Relation symbols $\mathtt{R} = \{\mathsf{In},\mathsf{Color},\mathsf{Adj},=\}$ where $\mathsf{In}(x,y)$ states that agent or box $x$ is in room $y$, $\mathsf{Color}(x,y)$ states that box $x$ has color $y$ and $\mathsf{Adj}(x,y)$ states that room $x$ is adjacent to room $y$.
    \item Function symbols $\mathtt{F} = \emptyset$.
    \item Type assignment $\mathtt{t}$ with constant types $\mathtt{t}(x) = \mathtt{agt}$ for $x \in Agents_{con}$, $\mathtt{t}(x) = \mathtt{obj}$ for $x \in Rooms_{con} \cup Boxes_{con} \cup Colors_{con}$, relation types $\mathtt{t}(\mathsf{In}) = (\mathtt{agt\_or\_obj},\mathtt{obj})$, and $\mathtt{t}(\mathsf{Color}) = \mathtt{t}(\mathsf{Adj}) = (\mathtt{obj},\mathtt{obj})$.
\end{itemize}

\end{example}

\begin{defn}
The set of \textbf{\textit{terms}}\textit{ }$\mathtt{T}$ of a signature $\Sigma=(\mathtt{V},\mathtt{C},\mathtt{R},\mathtt{F},\mathtt{t})$
is given by the grammar
\[
t\Coloneqq x \mid c \mid f(t_{1},...,t_{n})
\]
for $x \in \mathtt{V}$, $c \in \mathtt{C}$ and $f\in \mathtt{F}$, provided that $t_{1},...,t_{n}\in \mathtt{T}$ and $\mathtt{t}(f)_{|len(\mathtt{t}(f))-1} = (\mathtt{t}(t_{1}),...,\mathtt{t}(t_{n}))$.

A term is \textbf{\textit{ground }}if it does not contain any variables;
it is \textbf{\textit{free}} if all its terms are (i) variables or (ii) function symbols all whose arguments are free terms.
\end{defn}

By the definitions of type assignments and terms, it is the case that for all $t\in\mathtt{T}$,  $\mathtt{t}(t)\in\{\mathtt{agt},\mathtt{obj}\}$. This allows for a uniform definition of
formulas in term-modal languages:

\begin{defn}
\label{Def. wff} Let $\Sigma=(\mathtt{V},\mathtt{C},\mathtt{R},\mathtt{F},\mathtt{t})$ be a signature. Let $t_{1},...,t_{n}\in\mathtt{T}$ and $r\in\mathtt{R}$ with $\mathtt{t}(r)=(\mathtt{t}(t_{1}),...,\mathtt{t}(t_{n}))$, let $\dagger\in\mathtt{T}$ with $\mathtt{t}(\dagger)=\mathtt{agt}$, and let $x\in\mathtt{V}$. The \textbf{\textit{language}} $\mathcal{L}$ is then given by the grammar
\[
\varphi\Coloneqq r(t_{1},...,t_{n})\mid\neg\varphi\mid\varphi\wedge\varphi\mid K_{\dagger}\varphi\mid\forall x\varphi
\]
An \textbf{\textit{atom}} is a formula obtained by the first clause. An atom is \textbf{\textit{ground }}if all its terms are ground; it is free if all its terms are \textbf{\textit{free}}. Denote by $\mathtt{GroundAtoms}(\mathcal{L})$ and $\mathtt{FreeAtoms}(\mathcal{L})$ the set of all ground and free atoms in $\mathcal{L}$, respectively. 
\end{defn}
Throughout, the standard Boolean connectives as well as $\top$, $\bot$ and $\exists$ are used as meta-linguistic abbreviations as usual. We abbreviate inequality expressions of the form $\neg(t_1 = t_2)$ by $(t_1 \neq t_2)$. \textbf{\textit{Free}} and \textbf{\textit{bound}} variables may be defined recursively as usual with the free variables of $K_{t}\varphi$ the free variables of $\varphi$ plus the variables in $t$. A formula is a \textbf{\textit{sentence}} if it has no free variables. With $\varphi\in\mathcal{L},t\in\mathtt{T}$, $x\in\mathtt{V},\mathtt{t}(x)=\mathtt{t}(t)$ and no bound variables of $\varphi$ occurring in $t$, the result of replacing all occurrences of $x$ in $\varphi$ with $t$ is denoted $\varphi(x\mapsto t)$.

\begin{remark} $K_t \varphi$ is read as ``agent $t$ knows that $\varphi$''. Epistemic expressions are only well-defined when $t$ is an agent term. The language $\mathcal{L}$ neither enforces nor requires a fixed-size agent  set, in contrast with standard epistemic languages, where the set of operators is given by reference to some index set. Fixed-size agent sets are discussed throughout.
\end{remark}

\section{State Representation}\label{sec:state.rep}
In planning frameworks based on epistemic logic, states are often represented using  \textit{possible-worlds models}, tracing back to the work of Hintikka \cite{TML_Hintikka1962} and Kripke \cite{Kripke1962}. The standard epistemic interpretation of such models---employed here in all examples---is one of \textit{indistinguishability}, as follows. A model contains a set of worlds, each representing a physical state of affair. For each agent, a model contains a binary relation on the set of worlds. Under the indistinguishability interpretation, this relation is taken to be an equivalence relation. If two worlds are related for agent $\alpha$, then $\alpha$ cannot distinguish them given her current information. I.e., they are \textit{informationally indistinguishable} for $\alpha$. Hence, when $\alpha$ in fact is in some world $w$, she cannot tell which of the worlds related to $w$ she is in fact in. The set of worlds indistinguishable from $w$ for agent $\alpha$ is therefore sometimes referred to as agent $\alpha$'s \textit{range of uncertainty} (at $w$). The term \textit{information cell} is used to cover the same, and a world in $\alpha$'s range of uncertainty is said to be considered \textit{possible} by $\alpha$ (at $w$). An agent's range of uncertainty determines its knowledge: an agent knows $\varphi$ in world $w$ if $\varphi$ is true in all the worlds in the agent's range of uncertainty at $w$. For instance, if the agent has no information about two blocks $A$ and $B$, and therefore cannot tell whether one of them is stacked on the other or not, she will consider at least three worlds possible: one in which $A$ is indeed stacked on $B$, one in which it is not, and one in which $B$ is stacked on $A$. Possible-worlds models represent also all levels of \textit{higher-order knowledge}. E.g. agent $\alpha$ knows that agent $\beta$ knows $\varphi$ if $\alpha$ does not consider it possible that $\beta$ considers it possible that $\varphi$ is false.

A possible-worlds model is formally defined as a structure in general called a \textit{Kripke model}. Kripke models need not enforce any properties on the agents' relations. Our results hold for the general case, with equivalence relations a special case. Under the indistinguishability interpretation, Kripke models are often called \textit{epistemic models} or \textit{epistemic states}. For a thorough explanation of the components of a Kripke model, we refer the reader to \cite{sep-dynamic-epistemic,ditmarsch2007}. When the context makes it clear, such a structure may simply be called a \textit{model}. A model consists of a \textit{frame} and an \textit{interpretation}. Two things
differentiate the frame used here from the standard, propositional
version. First, a frame here contains a constant domain of elements existing
in each world. Working with distinct agents and objects, the domain
is a disjoint union of two sets, the agent domain and the object domain. Second, the accessibility relations over worlds are directly associated with elements in the agent domain.
The agent domain thereby makes reference to an index set\textemdash as
used in non-term-modal logical frames\textemdash redundant. The definition
of a frame is thereby self-contained. 
\begin{defn}
\label{Def. Frame} A \textbf{\textit{frame}} $F$ is a triple $F=(D,W,R)$
where
\begin{enumerate}
\item $D\coloneqq D_{\mathtt{agt\_or\_obj}}\coloneqq D_{\mathtt{agt}}\dot{\cup}D_{\mathtt{obj}}$, called the \textit{domain}, is the disjoint union of the non-empty sets $D_{\mathtt{agt}}$ and $D_{\mathtt{obj}}$, called the \textit{agent domain} and the \textit{object domain}, respectively.
\item $W$ is a non-empty set of \textit{worlds}.
\item $R$ is a map associating to each agent $i\in D_{\mathtt{agt}}$ a binary \textit{accessibility relation }on $W$. I.e.,\\ $R:D_{\mathtt{agt}}\longrightarrow\mathcal{P}(W\times W)$.
\end{enumerate}
Write $R_{i}$ for $R(i)$, write $wR_{i}w'$ for $(w,w')\in R_{i}$ and write $R_{i}(w)$ for $\{w'\in W\colon wR_{i}w'\}$. If $|D_{\mathtt{agt}}|=n$ and $|D_{\mathtt{obj}}|=m$, ($n,m\in\mathbb{N})$, say $F$ is of \textit{size} $(n,m)$. Denote by $\boldsymbol{F}$ the class of all frames and by $\boldsymbol{F}_{(n,m)}$ the class of all frames of size $(n,m)$.
\end{defn}

For propositional modal logic, a frame is augmented by a valuation
assigning an extension of worlds to each propositional symbol. In
the first-order and term-modal cases, each non-logical symbol is
assigned an extension in the domain. Here, this extension is assigned
\textit{world-relatively} for both relation symbols, function symbols and constants.
In particular the last is note-worthy: the constants are thereby \textit{non-rigid}\textemdash they
may refer to different objects (and agents!) in different worlds.
The non-rigidity of constants allows for uncertainty about identity
cf. the example of Section \ref{sec:example.non-rigid} and play an important role concerning the validity of frame-property characterizing axioms, cf. Section \ref{subsec:frame.char}.

Constants that do not vary with worlds\textemdash so-called \textit{rigid
}constants\textemdash often come in handy when referring to
agents. A rigid constant provides a syntactic name for a semantic agent.
Rigid constants are a special case: a constant $c$ may be
forced rigid by assuming its interpretation $I(c,w)$ to be constant
over all worlds, i.e., by $I(c,w)=I(c,w')$ for all $w,w'\in W$.

\begin{defn}
Let a signature $\Sigma=(\mathtt{V},\mathtt{C},\mathtt{R},\mathtt{F},\mathtt{t})$
and a frame $F=(D,W,R)$ be given. An \textbf{\textit{interpretation}}
of $\Sigma$ over $F$ is a map $I$ satisfying for each $w\in W$:
\begin{enumerate}
\item $I(=,w)$ is the set $\{(d,d)\colon d\in D\}$.
\item For $c\in\mathtt{C}$, $I(c,w)\in D_{\mathtt{t}(c)}$.
\item For $r\in\mathtt{R}$, $I(r,w)\subseteq\prod_{i=1}^{len(\mathtt{t}(r))}D_{\mathtt{t}_{i}(r)}$.
\item For $f\in\mathtt{F}$, $I(f,w)\subseteq\prod_{i=1}^{len(\mathtt{t}(f))}D_{\mathtt{t}_{i}(f)}$
such that $I(f,w)$ is a (possibly partial) map: i.e.,\\
if $(d_{1},...,d_{k},d_{len(\mathtt{t}(f))}),(d_{1},...,d_{k},d'_{len(\mathtt{t}(f))})\in I(f,w)$,
then $d_{len(\mathtt{t}(f))} = d'_{len(\mathtt{t}(f))}$.
\end{enumerate}
With $F=(D,W,R)$ a frame and $I$ an interpretation of $\Sigma$
over $F$, the tuple $M=(D,W,R,I)$ is a \textbf{\textit{model}}. Both $w\in F$
and $w\in M$ states that $w\in W$. When $w\in M$, the pair $(M,w)$
is a \textbf{\textit{pointed model}} with $w$ called the \textit{designated
world}.
\end{defn}
Satisfaction for all formulas without variable occurrences may be
defined over pointed models. To specify satisfaction for the full
language, variables must also be assigned extension. Letting variable
valuations be world independent\textemdash or \textit{rigid}\textemdash trans-world
identification of objects and agents may be made using suitable bound
variables, cf. e.g. the \textit{de re} knowledge in Example \ref{ex:ep.state}.

\begin{defn}
\label{Def. Valuation} Let a signature $\Sigma=(\mathtt{V},\mathtt{C},\mathtt{R},\mathtt{F},\mathtt{t})$
and a frame $F=(D,W,R)$ be given. A \textbf{\textit{valuation }}of
$\Sigma$ over $F$ is a map $v:\mathtt{V}\longrightarrow D$ such
that $v(x)\in D_{\mathtt{t}(x)}$. An \textbf{\textit{x-variant}}\textit{
$v'$ }of $v$ is a valuation such that $v'(y)=v(y)$ for all $y\in\mathtt{V},y\neq x$.
\end{defn}

Jointly, an interpretation and a valuation assigns an extension to
every term $t$ of $\Sigma$ relative to every world of a frame. The
following involved, but uniform, notation will be used throughout
to denote the extension of terms:

\begin{defn}
Let a signature $\Sigma=(\mathtt{V},\mathtt{C},\mathtt{R},\mathtt{F},\mathtt{t})$,
a model $M=(D,W,R,I)$ and a valuation $v$ be given. The \textbf{\textit{extension
}}of the term $t\in\mathtt{T}$ in $M$ under $v$ is
\begin{align*}
\left\llbracket t\right\rrbracket _{w}^{I,v}= & \begin{cases}
v(t) & \text{if }t\in\mathtt{V}\\
I(t,w) & \text{if }t\in\mathtt{C}\\
d\text{ with }(d_{1},...,d_{n},d)\in I(f,w) & \text{if }t=f(t_{1},...,t_{n})
\end{cases}
\end{align*}
\end{defn}

For exactly the terms $t$ with $\mathtt{t}(t) = \mathtt{agt}$, $R_{\left\llbracket t\right\rrbracket _{w}^{I,v}}$ is then an accessibility relation in $M$. The extension of terms thus play a key role in the satisfaction of modal formulas:

\begin{defn}
\label{def:satisfaction}Let $\Sigma=(\mathtt{V},\mathtt{C},\mathtt{R},\mathtt{F},\mathtt{t})$,
$M=(D,W,R,I)$ and $v$ be given. The \textbf{\textit{satisfaction}}
of formulas of $\mathcal{L}$ is given recursively by

$M,w\vDash_{v}r(t_{1},...,t_{n})$ iff $(\left\llbracket t_{1}\right\rrbracket _{w}^{I,v},...,\left\llbracket t_{n}\right\rrbracket _{w}^{I,v})\in I(r,w)$
for all $r\in\mathtt{R}$, including $=$.

$M,w\vDash_{v}\neg\varphi$ iff not $M,w\vDash_{v}\varphi$.

$M,w\vDash_{v}\varphi\wedge\psi$ iff $M,w\vDash_{v}\varphi$ and
$M,w\vDash_{v}\psi$.

$M,w\vDash_{v}\forall x\varphi$ iff $M,w\vDash_{u}\varphi$ for
every $x$-variant $u$ of $v$.

$M,w\vDash_{v}K_{t}\varphi$ iff $M,w'\vDash_{v}\varphi$ for all
$w'$ such that $(w,w')\in R_{\left\llbracket t\right\rrbracket _{w}^{I,v}}$.
\end{defn}
\begin{example}[Epistemic model for $\text{SC}(3,4,1,2)$]\label{ex:ep.state}
Figure \ref{example:initial_model} depicts an epistemic model $M_0=(D,W,R,I)$ for the initial state $s_0$ described in Section \ref{sec:running_example}. 

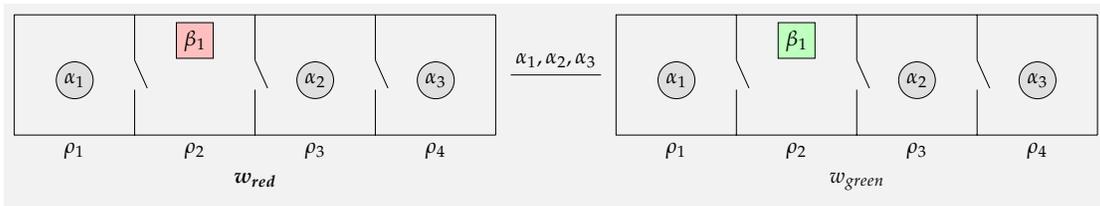
\begin{figure}[H]
  \begin{center}
    \scalebox{0.8}{
            \begin{circuitikz}[framed]
            \draw
            (1,0) node[below] {$\rho_1$}
            (3,0) node[below] {$\rho_2$}
            (4,-0.5) node[below] {$\boldsymbol{w_{red}}$}
            (5,0) node[below] {$\rho_3$}
            (7,0) node[below] {$\rho_4$}
            (2,0) -- (0,0) -- (0,2) -- (2,2)
            (2,2) to[nos] (2,0)
            (1,1.35) node[below] {\circled{$\alpha_1$}}
            (2,2) -- (4,2) 
            (2,0) -- (4,0)
            (4,2) to[nos] (4,0)
            (3,2) node[below] {\squaredr{$\beta_1$}}
            (4,2) -- (6,2) 
            (4,0) -- (6,0)
            (6,2) to[nos] (6,0)
            (5,1.35) node[below] {\circled{$\alpha_2$}}
            (6,2) -- (8,2) 
            (6,0) -- (8,0)
            (8,2) -- (8,0)
            (7,1.35) node[below] {\circled{$\alpha_3$}}
            ;
            \draw
            (8.25,1) -- (9.75,1) node[midway,above]{$\alpha_1,\alpha_2,\alpha_3$}
            ;
            \begin{scope}[xshift=10cm]
            \draw
            (1,0) node[below] {$\rho_1$}
            (3,0) node[below] {$\rho_2$}
            (4,-0.5) node[below] {$w_{green}$}
            (5,0) node[below] {$\rho_3$}
            (7,0) node[below] {$\rho_4$}
            (2,0) -- (0,0) -- (0,2) -- (2,2)
            (2,2) to[nos] (2,0)
            (1,1.35) node[below] {\circled{$\alpha_1$}}
            (2,2) -- (4,2) 
            (2,0) -- (4,0)
            (4,2) to[nos] (4,0)
            (3,2) node[below] {\squaredg{$\beta_1$}}
            (4,2) -- (6,2) 
            (4,0) -- (6,0)
            (6,2) to[nos] (6,0)
            (5,1.35) node[below] {\circled{$\alpha_2$}}
            (6,2) -- (8,2) 
            (6,0) -- (8,0)
            (8,2) -- (8,0)
            (7,1.35) node[below] {\circled{$\alpha_3$}}
            ;
            \end{scope}
            \end{circuitikz}
    }
    \end{center}
\caption{\label{example:initial_model} $(M_0,w_{red})$, a pointed epistemic model for the initial state $s_0$ described in Section \ref{sec:running_example}. The agents are uncertain about the color of $\beta_1$, which may be red $(w_{red})$ or green $(w_{green})$. This is captured by the edge linking $w_{red}$ and $w_{green}$. Reflexive edges are not drawn. The name of the actual world, $w_{red}$, is marked with boldface letters.}
\end{figure}

Formally, the model $M_0 =(D,W,R,I)$ has 
\begin{itemize}
    \item $D = D_{\mathtt{agt}}\dot{\cup}D_{\mathtt{obj}}$, with $D_{\mathtt{agt}}=\{\alpha_1,\alpha_2,\alpha_3\}$ and $D_{\mathtt{obj}} = Rooms_{obj}\cup Boxes_{obj} \cup Colors_{obj}$, where $Rooms_{obj}=\{\rho_1,\rho_2,\rho_3,\rho_4\}$,  $Boxes_{obj}=\{\beta_1\}$ and  $Colors_{obj}=\{Red,Green\}$. 
    \item $W = \{w_{red},w_{green}\}$.
    \item  $R(\alpha_i) = W\times W$, for $i\in\{1,2,3\}$. 
    \item The interpretation of all constants is the same in $w_{red}$ and $w_{green}$, i.e., all constants are rigid: $I(a_i,u)=\alpha_i$, $I(r_i,u)=\rho_i$, $I(b_i,u)=\beta_i$, $I(green,u)=Green$ and $I(red,u)=Red$ for all $u\in W$.
    \item  The interpretation of the predicates is as follows: $I(\mathsf{In},u)=\{(\alpha_1,\rho_1),(\alpha_2,\rho_3), (\alpha_3,\rho_4),(\beta_1,\rho_2)\}$, for all $u\in W$, $I(\mathsf{Color},w_{red})=\{(\beta_1,red)\}$, $I(Color,w_{green})=\{(\beta_1,green)\}$. The interpretation of the $\mathsf{Adj}$ predicate is as expected.
\end{itemize}

Following the semantics from Definition \ref{def:satisfaction}, it can be seen that $M_0,w_{red}\vDash_v \forall x (K_{x} \mathsf{In}(b_1,r_2))$, i.e., every agent knows the location of box $\beta_1$. Similarly, every agent knows that all agents know this, since $M_0,w_{red}\vDash_v \forall y \forall x(K_{y}K_{x} \mathsf{In}(b_1,r_2))$.  
Moreover, the agents know that the box has a color, but not what color it is. They thus have what is called \textit{de dicto} knowledge of the coloring of the box, but not \textit{de re} knowledge. Agent $\alpha_3$'s de dicto knowledge is captured by $M_0,w_{red}\vDash_v K_{a_3} \exists x \mathsf{Color}(b_1,x)$, while its lack of de re knowledge is captured by $M_0,w_{red}\vDash_v \neg \exists x K_{a_3}\mathsf{Color}(b_1,x)$. Finally, agent $\alpha_3$ has de re knowledge of the box, or as Hintikka \cite{TML_Hintikka1962} puts it, $\alpha_3$ knows what the box is, captured by $M_0,w_{red}\vDash_v \exists x K_{a_3}(x = b_1)$.

\end{example}

\section{Action Representation}\label{sec:action.rep}
In automated planning, a distinction is often drawn between  \textit{action schemas}, which describe classes of actions in a general way, and \textit{ground actions}, which represent a specific action with a fixed set of agents and objects  \cite{ghallab2004automated,russell2016artificial}. Action schemas use so-called \textit{action parameters} or \textit{variables}, which are instantiated into constants to define an action. For example, a schema may be used to represent all actions of the form `agent $x$ tells $y$ that object $z$ has color $u$', where $x$, $y$, $z$ and $u$ are variables standing for agents and objects. A corresponding ground action is obtained by replacing all free variables by names referring to specific agents and objects. For example, a schema instance could be `$ann$ tells $bob$ that $box1$ has color $red$'. 

In DEL, the descriptions of concrete actions are called \textit{action models}. That is, DEL action models correspond to a ground action in classical planning. Following the DEL naming conventions, models of concrete actions will be called action models, whereas variabilized models in the spirit of PDDL will be called action schemas. Action models and action schemas, as well as a suitable notion of schema instantiation relating the two, are introduced next.

\subsection{Action Models}
Formally and intuitively, action models are closely related to Kripke models. Where Kripke models contains worlds and relations, action models instead contain \textit{events} and relations. Under the standard epistemic interpretation, the relations again represent indistinguishability and are again assumed to be equivalence relations. Again, this is a special case of the models introduced here.

The DEL-style action models we add to the term-modal logic setting include \textit{preconditions} (\cite{baltagmoss2004,BaltagBMS_1998}), \textit{postconditions} (\cite{benthem2006_com-change,bolanderbirkegaard2011,vandit2005dynamic}) as well as \textit{edge-conditions} similarly to \cite{Bolander2014}. Preconditions specify when an event is executable (e.g., a precondition of opening the door is that it is closed.) Postconditions describe the physical effects of events (e.g., the door is open after the event). Edge-conditions are used to represent how an agent's observation of an action depends on the agent's circumstances. For example, the way in which an agent $\alpha_i$ observes an action performed by agent $\alpha_j$ may depend on $\alpha_i$'s proximity to $\alpha_j$ or to the objects affected by the action (e.g., to $\alpha_i$, the events of opening and closing the door are distinguishable if $\alpha_i$ can see or hear the door, but else not). Edge-conditions provide a general way to describe actions whose observability is context-dependent. The epistemic effects of an action model is encoded by the \textit{product update} operation by which action models are applied to Kripke models (defined below).

In more detail, the components of term-modal action models (Def. \ref{Def. Action model} below) play the following roles. $E$ represents the set of events that might occur as the action is executed. $Q$ is a map that assigns to each edge $(e,e')\in E\times E$ an edge-condition: a formula with a single free variable ${x^\star}$. Given a model $M$ describing the situation in which the action is applied, an agent $\alpha$ cannot distinguish $e$ from $e'$ iff the edge-condition from $e$ to $e'$ is true in $M$ \textit{when the free variable ${x^\star}$ is mapped to $\alpha$}. Intuitively, if the situation described by the edge-condition is true for $\alpha$, the way $\alpha$ is observes the action does not allow her to tell whether $e$ or $e'$ is taking place. The precondition restricts the applicability of an event $e$ to those states satisfying the precondition formula $\pre(e)$. Precondition formulas contain no free variables to ensure that their effects are conditional only on the model, but not the variable valuation. The postcondition $\mathsf{post}(e)$ describes the physical changes induced by the event. If both $\pre(e)$ and $\mathsf{post}(e)(r(t_1,\dots,t_n))$ are true in a state $s$ of a model $M$, then the event $e$ occurs, and after its occurrence, $r(t_1,\dots, t_n)$ is true in the updated version of $s$. That is, $r(t_1,\dots, t_n)$ is a \textit{conditional effect} of event $e$ with condition  $\mathsf{post}(e)(r(t_1,\dots,t_n))$.

The language used to state pre- and postconditions in action models is an extension of $\mathcal{L}$, denoted $\mathcal{L}_{AM}$, to be introduced in Section \ref{sec:act.lang}. This extended language has formulas of the form $[A,e]\varphi$, which are interpreted as: `after event $e$ of action $A$, $\varphi$ holds'. This type of formula makes it possible to mention other actions in the pre- and postconditions of actions, i.e., to express syntactically some dependencies or interactions between actions. However, the action model construction does not require or depend on the use of $\mathcal{L}_{AM}$ rather $\mathcal{L}$, so the reader can safely ignore the details of $\mathcal{L}_{AM}$ for now.

\begin{defn}
\label{Def. Action model} An \textbf{\textit{action model}} $A$ is a tuple $A=(E,Q,\mathsf{pre},\mathsf{post})$ where 
\begin{enumerate}
\item $E$ is a non-empty, finite set of possible \textit{events}.
\item $Q: (E\times E) \to \mathcal{L}_{AM}$, where for each pair $(e,e')$ the formula $Q(e,e')$ has exactly one free variable $x^\star$.
\item $\mathsf{pre}: E \to \mathcal{L}_{AM}$ is a map that assigns to each event $e\in E$ a precondition formula with no free variables.
\item $\mathsf{post}: E \to (\mathtt{GroundAtoms}(\mathcal{L}) \to \mathcal{L}_{AM})$ is a map that assigns to each event $e\in E$ a postcondition for each ground atom.
\end{enumerate}
It is required that $\mathsf{post}(e)(=(t,t))=\top$ for each event $e$, to preserve the meaning of equality. A pair $(A,e)$ consisting of the action and an event from $E$ is called a \textbf{\textit{pointed action}}. 
\end{defn}

\begin{notat} Let $A=(E,Q,\mathsf{pre},\mathsf{post})$ be an action model. We denote by $\dom( \mathsf{post}(e))$ the set of atoms for which $\post(e)(r(t_1,\dots,t_k)) \neq r(t_1,\dots,t_k)$. We denote any $\post(e)$ that maps every atom to itself by $id$ (the identity function). When convenient, we add the superscript ``$A$'' to the components of $A$, so that $A=(E^A,Q^A,\mathsf{pre}^A,\mathsf{post}^A)$. 
\end{notat}

To ensure that postcondition functions are finite objects, each $\mathsf{post}(e)$ is often required to be only finitely different from the identity function. That is, $\text{dom}(\mathsf{post}(e))$ is required to be finite. This allows for a finite encoding of postconditions, as only pairs with $post(e)(\varphi)\neq \varphi$ need to be stored in memory. Especially in planning, it should be possible to write down a sequence of symbols that completely specifies any given action model in the language. For the sake of generality, we do not impose this restriction in the definition of an action model. But for all practical purposes, this standard restriction will be needed.

\begin{notat} Let $A=(E,Q,\mathsf{pre},\mathsf{post})$ be an action model.
When $A$ is illustrated as a labelled graph, for each node $e\in E$, we write the precondition and postconditions for $e$ as a pair $\langle \mathsf{pre}(e); \mathsf{post}(e)(\psi_1)=\varphi_1 \wedge \dots \wedge \mathsf{post}(e)(\psi_n)=\varphi_n\rangle$. We write postconditions such as $\post(e)(\varphi)=\top \wedge \post(e)(\psi)=\bot$ using the notation $\varphi \wedge \neg \psi$ (indicating that the action makes $\varphi$ true and $\psi$ false unconditionally). In graphs, we omit the postconditions for atoms $\varphi$ with $post(e)(\varphi) = \varphi$.

We do not include the edge-conditions for reflexive loops in illustrations, but always assume that for all $e\in E$, $Q(e,e)=(x^\star=x^\star)$, to the effect that all agents retain reflexive relations following updates. When two events $e,e'$ are connected by a line without arrowheads labeled by a single formula $\varphi$, this means that $Q(e,e')=Q(e',e)=\varphi$, retaining symmetry.
\end{notat}

\begin{example}[Action models for $\text{SC}(3,4,1,2)$]\label{example:action_models} 
Figures \ref{fig:model_move}, \ref{fig:model_sense} and \ref{fig:model_say} depict graphically the action models for the three actions in the plan from in Section \ref{sec:running_example}, i.e., the following movement, sensing and announcement actions: \textit{$\alpha_1$ moves to $\rho_2$, $\alpha_1$ senses the color of $\beta_1$, $\alpha_1$ announces the color of $\beta_1$}.

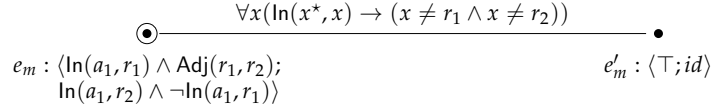
\begin{figure}[H]
    \begin{center}
        \scalebox{0.85}{
            \begin{tikzpicture}
              \tikzstyle{dnode}=[inner sep=1pt,outer sep=1pt,draw,circle,minimum width=9pt]
              \tikzstyle{nnode}=[inner sep=1pt,outer sep=1pt,circle,minimum width=9pt]
              \tikzstyle{label-edge}=[midway,fill=white, inner sep=1pt]
              \node[dnode, label={[align=right]below: $e_m: \langle \mathsf{In}(a_1,r_1) \wedge \mathsf{Adj}(r_1,r_2);$ \\ $\mathsf{In}(a_1,r_2)\wedge \neg \mathsf{In}(a_1,r_1) \rangle$}] (w0) at (0,0) {};
              \fill (w0) circle [radius=2pt];
              \node[nnode,  label={below: $e'_m: \langle \top;id \rangle$}] (w1) at (8,0) {};
              \fill (w1) circle [radius=2pt];
              \draw[-] (w0) -- (w1) node[above,midway] {$\forall x (\mathsf{In}({x^\star},x) \to (x \neq r_1 \wedge x \neq r_2))$};
            \end{tikzpicture}
        }
    \end{center}
\caption{\label{fig:model_move} $\text{Move}(a_1,r_1,r_2)$, the action model for $\alpha_1$ moving from $\rho_1$ to $\rho_2$. Event $e_{m}$ describes what is actually taking place (in the drawing for the model, the actual event is marked with a double circle). The precondition formula says that $\alpha_1$ is in $\rho_1$ and that $\rho_2$ is next to $\rho_1$. The action changes $\alpha_1$'s location to $\rho_2$, as captured by the postcondition. The event $e'_m$ describes the situation in which nothing happens. This is how the action looks to any agent that is neither in the room $\alpha_1$ is currently in, nor in the room the agent is moving to. The edge-condition linking the two events captures this observability constraint.} 
\end{figure}

\begin{figure}[H]
    \begin{center}
        \scalebox{0.85}{
            \begin{tikzpicture}
              \tikzstyle{dnode}=[inner sep=1pt,outer sep=1pt,draw,circle,minimum width=9pt]
              \tikzstyle{nnode}=[inner sep=1pt,outer sep=1pt,circle,minimum width=9pt]
              \tikzstyle{label-edge}=[midway,fill=white, inner sep=1pt]
              \node[dnode, label={[align=right]below: $e_{s}: \langle \mathsf{In}(a_1,r_2) \wedge \mathsf{In}(b_1,r_2)\wedge \mathsf{Color}(b_1,red); id \rangle$}] (w0) at (0,0) {};
              \fill (w0) circle [radius=2pt];
              \node[nnode,  label={below: $e'_s:  \langle \mathsf{In}(a_1,r_2) \wedge \mathsf{In}(b_1,r_2)\wedge \neg \mathsf{Color}(b_1,red); id \rangle$}] (w1) at (8,0) {};
              \fill (w1) circle [radius=2pt];
              \node[nnode,  label={above: $e''_s:  \langle \top; id \rangle$}] (w2) at (4,3) {};
              \fill (w2) circle [radius=2pt];
              \draw[-] (w0) -- (w1) node[above,midway] {$\forall x (\mathsf{In}({x^\star},x) \to  x \neq r_2)$};
              \draw[-] (w0) -- (w2) node[above,midway,rotate=37.5] {$\forall x (\mathsf{In}({x^\star},x) \to x \neq r_2)$};
              \draw[-] (w1) -- (w2) node[above,midway,rotate=-37.5] {$\forall x (\mathsf{In}({x^\star},x) \to x \neq r_2)$};
            \end{tikzpicture}
        }
    \end{center}
\caption{\label{fig:model_sense} $\text{SenseCol}(a_1,red,b_1,r_2)$, the action model for $\alpha_1$ sensing in room $\rho_2$ whether box $\beta_1$ is red or not. Event $e_{s}$ describes what is actually taking place, i.e., $\alpha_1$ seeing that the box is red. The action is a purely epistemic action, i.e., it does not change the physical state of the environment, and therefore the postcondition $\post(e_s)$ is $id$. $e'_s$ represents the event in which $\alpha_1$ sees that the box is not red, while $e''_s$ represents the event in which nothing happens. The agents that are not in $\rho_2$ cannot observe what $\alpha_1$ is doing. More precisely, they cannot distinguish between $\alpha_1$ seeing that the box is red, $\alpha_1$ seeing that it is not red, and $\alpha_1$ doing nothing. This is captured by the edge-conditions.}
\end{figure}
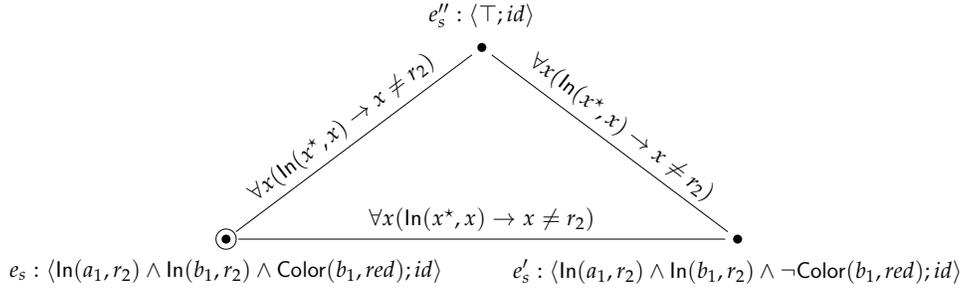

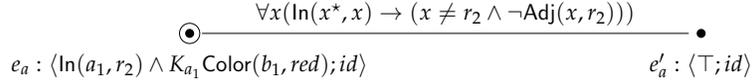
\begin{figure}[H]
    \begin{center}
        \scalebox{0.85}{
            \begin{tikzpicture}
              \tikzstyle{dnode}=[inner sep=1pt,outer sep=1pt,draw,circle,minimum width=9pt]
              \tikzstyle{nnode}=[inner sep=1pt,outer sep=1pt,circle,minimum width=9pt]
              \tikzstyle{label-edge}=[midway,fill=white, inner sep=1pt]
              \node[dnode, label={[align=right]below: $e_a: \langle \mathsf{In}(a_1,r_2) \wedge K_{a_1}\mathsf{Color}(b_1,red); id \rangle$}] (w0) at (0,0) {};
              \fill (w0) circle [radius=2pt];
              \node[nnode,  label={below: $e'_a: \langle \top;id \rangle$}] (w1) at (8,0) {};
              \fill (w1) circle [radius=2pt];
              \draw[-] (w0) -- (w1) node[above,midway] {$\forall x (\mathsf{In}({x^\star},x) \to ( x \neq r_2 \wedge \neg \mathsf{Adj}(x,r_2)))$};
            \end{tikzpicture}
        }
    \end{center}
\caption{\label{fig:model_say} $\text{Announce}(r_1,red,b_1,r_2)$, the action model for $\alpha_1$ announcing that $\beta_1$ is red while in room $\rho_2$. Event $e_a$ describes what is actually taking place. The precondition formula $\pre(e_a)$ says that $\alpha_1$ is in $\rho_2$ and that $a_1$ knows that the color of $\beta_1$ is red. Event $e'_a$ describes the event in which nothing occurs. This is what the announcement looks like to any agent that cannot hear the announcement. An agent $\alpha_i$ cannot hear $\alpha_1$'s announcement if $\alpha_i$ is neither in $\alpha_1$'s room nor in a room that is adjacent to it. This is captured by the (identical) edge-conditions $Q(e_a,e'_a)$ and $Q(e'_a,e_a)$.}
\end{figure}
\end{example}

\subsection{Product Update}\label{subsec:prod.upd.}

Having defined epistemic models and action models, we introduce an operation that computes the epistemic model $M'$ reached by applying action $A$ in model $M$. The operation is a first-order variant of \textit{product update} \cite{BaltagBMS_1998}. Under the indistinguishability interpretation, the core epistemic intuition is that to tell two worlds apart after an update, either the agent could tell them apart beforehand, or it could tell them apart by something happening in one, but not the other. In slightly more detail: Assume that \textit{after} an update, a model contains worlds $(w,e)$ and $(w',e')$, representing that event $e$ occurred in world $w$, and $e'$ occurred in $w'$. Then $(w,e)$ is indistinguishable from $(w',e')$ for agent $\alpha$ iff $\alpha$ found both $w$ and $w'$  indistinguishable and events $e$ and $e'$ indistinguishable. Formally, (term-modal) product update is defined below. An explanatory remark follows the definition.

\begin{defn}
\label{Def. Product update}
Let $M=(D,W,R,I)$ and $A=(E,Q,\mathsf{pre},\mathsf{post})$ be given. The \textbf{\textit{product update}} of $M$ and $A$ yields the epistemic model $M\otimes A = (D', W', R',I')$ where 
\begin{enumerate}
\item $D' = D$
\item $W' = \{ (w,e) \in W\times E \colon (M,w) \vDash_v \mathsf{pre}(e)\}$,
\item For each $i\in D_\texttt{agt}$, $(w,e)R'_i(w',e')$ iff $w R_i w'$ and $M,w\vDash_{v[ {x^\star}\mapsto i ]} Q(e,e')$,
\item $I'(c,(w,e))=I(c,w)$ for all $c\in \mathtt{C}$, $I'(f,(w,e))=I(f,w)$ for all $f\in \mathtt{F}$, and $I'(r,(w,e)) = (I(r,w) \cup r^+(w)) \setminus r^-(w)$, where:
      \begin{align*}
    r^+(w) \coloneqq & \{ (\llbracket t_1\rrbracket^{I,v}_w, \dots,\llbracket t_k\rrbracket^{I,v}_w) \colon (M,w) \vDash_v \mathsf{post}(e)(r(t_1,\dots,t_k))\}\\
    r^-(w) \coloneqq & \{ (\llbracket t_1\rrbracket^{I,v}_w, \dots,\llbracket t_k\rrbracket^{I,v}_w) \colon (M,w) \not\vDash_v \mathsf{post}(e)(r(t_1,\dots,t_k))\}
    \end{align*}
\end{enumerate}
If $(M,w)\vDash_v \pre(e)$, $(A,e)$ is \textbf{\textit{applicable}} to $(M,w)$. If $(A,e)$ is applicable to $(M,w)$, the product update of the two yields the pointed epistemic model $(M\otimes A, (w,e))$. Else it is undefined.
\end{defn}

\begin{remark}
The components of the updated model are as follows. 
The domain of the updated model $D'$ is unchanged, since action models change the state of agents and objects, but do not introduce or remove them. A state $(w,e)$ is in the updated set of states $W'$ if, and only if, $e$ is applicable in $w$, i.e., if $(M,w)$ satisfies the precondition $\pre(e)$. As $\pre(e)$ has no free variables by construction, 
the set of worlds $W'$ is independent of the assignment $v$.
The state $(w,e)$ represents the state reached by taking event $e$ in state $w$. Agent $\alpha$ cannot distinguish $(w,e)$ from $(w',e')$ if (1) $\alpha$ cannot distinguish $w$ from $w'$, which is the case if $wR_\alpha w'$; and (2) $\alpha$ cannot distinguish $e$ from $e'$ given its circumstances in $w$, which is the case if the edge-condition $Q(e,e')$ is true for agent $\alpha$ at $(M,w)$ when ${x^\star}$ is mapped to $\alpha$, i.e., when $M,w\vDash_{v[ {x^\star}\mapsto i ]} Q(e,e')$. Since actions do not change the denotation of ground terms, $I'$ agrees with $I$ in this respect. The extension of relations is changed according to event postconditions. If the condition $\post(e)(r(t_1,\dots,t_k))$ is true at $(M,w)$, then the tuple $(\llbracket t_1\rrbracket^{I,v}_w, \dots,\llbracket t_k\rrbracket^{I,v}_w)$ is added to the extension of $r$ at $(w,e)$, and it is removed otherwise.
\end{remark}

\begin{example}[Product updates for $\text{SC}(3,4,1,2)$] \label{example:prod_update} 
Starting from the initial epistemic model $(M_0,w_{red})$ from Example \ref{example:initial_model}, we model the effects of applying the actions in $\alpha_1$'s plan. First, $\alpha_1$ moves right. This action yields the new pointed model $(M_0\otimes \text{Move}(a_1,r_1,r_2), (w_{red},e_m))$, depicted in Figure \ref{fig:first_update}.
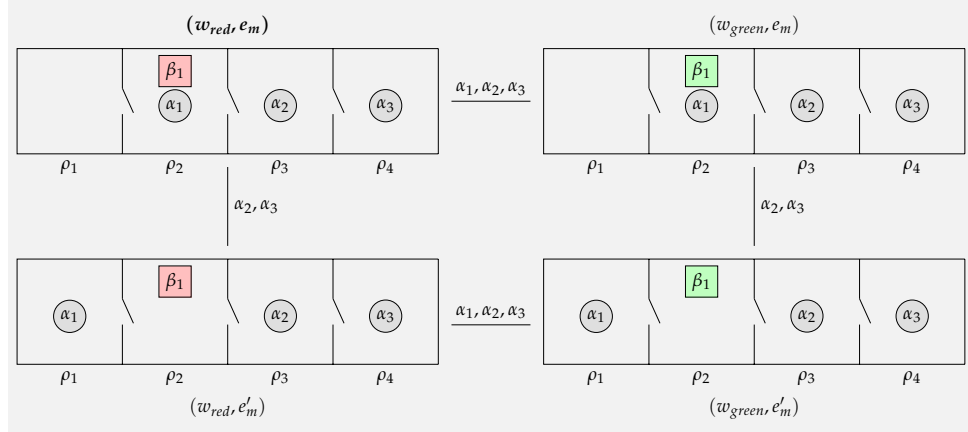
\begin{figure}[H]
  \begin{center}
    \scalebox{0.7}{
            \begin{circuitikz}[framed]
            \draw
            (4,2.75) node[below] {$\boldsymbol{(w_{red},e_m)}$}
            (1,0) node[below] {$\rho_1$}
            (3,0) node[below] {$\rho_2$}
            (5,0) node[below] {$\rho_3$}
            (7,0) node[below] {$\rho_4$}
            (2,0) -- (0,0) -- (0,2) -- (2,2)
            (2,2) to[nos] (2,0)
            (3,1.35) node[below] {\circled{$\alpha_1$}}
            (2,2) -- (4,2) 
            (2,0) -- (4,0)
            (4,2) to[nos] (4,0)
            (3,2) node[below] {\squaredr{$\beta_1$}}
            (4,2) -- (6,2) 
            (4,0) -- (6,0)
            (6,2) to[nos] (6,0)
            (5,1.35) node[below] {\circled{$\alpha_2$}}
            (6,2) -- (8,2) 
            (6,0) -- (8,0)
            (8,2) -- (8,0)
            (7,1.35) node[below] {\circled{$\alpha_3$}}
            ;
            \draw
            (8.25,1) -- (9.75,1) node[midway,
    above]{$\alpha_1,\alpha_2,\alpha_3$}
            (8.25,-3.25) -- (9.75,-3.25) node[midway,
    above]{$\alpha_1,\alpha_2,\alpha_3$}
            (14,-0.25) -- (14,-1.75) node[midway,
    right]{$\alpha_2,\alpha_3$}
            (4,-0.25) -- (4,-1.75) node[midway,
    right]{$\alpha_2,\alpha_3$}
            ;
            \begin{scope}[xshift=10cm]
            \draw
            (1,0) node[below] {$\rho_1$}
            (3,0) node[below] {$\rho_2$}
            (4,2.75) node[below] {$(w_{green},e_m)$}
            (5,0) node[below] {$\rho_3$}
            (7,0) node[below] {$\rho_4$}
            (2,0) -- (0,0) -- (0,2) -- (2,2)
            (2,2) to[nos] (2,0)
            (3,1.35) node[below] {\circled{$\alpha_1$}}
            (2,2) -- (4,2) 
            (2,0) -- (4,0)
            (4,2) to[nos] (4,0)
            (3,2) node[below] {\squaredg{$\beta_1$}}
            (4,2) -- (6,2) 
            (4,0) -- (6,0)
            (6,2) to[nos] (6,0)
            (5,1.35) node[below] {\circled{$\alpha_2$}}
            (6,2) -- (8,2) 
            (6,0) -- (8,0)
            (8,2) -- (8,0)
            (7,1.35) node[below] {\circled{$\alpha_3$}}
            ;
            \begin{scope}[yshift=-4cm]
            \draw
            (1,0) node[below] {$\rho_1$}
            (3,0) node[below] {$\rho_2$}
            (4,-0.5) node[below] {$(w_{green},e'_m)$}
            (5,0) node[below] {$\rho_3$}
            (7,0) node[below] {$\rho_4$}
            (2,0) -- (0,0) -- (0,2) -- (2,2)
            (2,2) to[nos] (2,0)
            (1,1.35) node[below] {\circled{$\alpha_1$}}
            (2,2) -- (4,2) 
            (2,0) -- (4,0)
            (4,2) to[nos] (4,0)
            (3,2) node[below] {\squaredg{$\beta_1$}}
            (4,2) -- (6,2) 
            (4,0) -- (6,0)
            (6,2) to[nos] (6,0)
            (5,1.35) node[below] {\circled{$\alpha_2$}}
            (6,2) -- (8,2) 
            (6,0) -- (8,0)
            (8,2) -- (8,0)
            (7,1.35) node[below] {\circled{$\alpha_3$}}
            ;
            \begin{scope}[xshift=-10cm]
            \draw
            (1,0) node[below] {$\rho_1$}
            (3,0) node[below] {$\rho_2$}
            (4,-0.5) node[below] {$(w_{red},e'_m)$}
            (5,0) node[below] {$\rho_3$}
            (7,0) node[below] {$\rho_4$}
            (2,0) -- (0,0) -- (0,2) -- (2,2)
            (2,2) to[nos] (2,0)
            (1,1.35) node[below] {\circled{$\alpha_1$}}
            (2,2) -- (4,2) 
            (2,0) -- (4,0)
            (4,2) to[nos] (4,0)
            (3,2) node[below] {\squaredr{$\beta_1$}}
            (4,2) -- (6,2) 
            (4,0) -- (6,0)
            (6,2) to[nos] (6,0)
            (5,1.35) node[below] {\circled{$\alpha_2$}}
            (6,2) -- (8,2) 
            (6,0) -- (8,0)
            (8,2) -- (8,0)
            (7,1.35) node[below] {\circled{$\alpha_3$}}
            ;
            \end{scope}
            \end{scope}
            \end{scope}
            \end{circuitikz}
    }
    \end{center}
    \caption{\label{fig:first_update} The pointed model $(M_0\otimes \text{Move}(a_1,r_1,r_2), (w_{red},e_m))$, representing the state after $\alpha_1$ moves into room $\rho_2$. Edges in the reflexive-transitive closure of the indistinguishability relations are omitted. At this point, $\alpha_2$ and $\alpha_3$ are uncertain about the location of $\alpha_1$. More precisely, they cannot tell whether $\alpha_1$ stayed in room $\rho_1$ or moved to $\rho_2$.}
\end{figure}

The second step is sensing the color of $\beta_1$. This action yields the model $(M_0 \otimes \text{Move}(a_1,r_1,r_2)\otimes\text{SenseCol}(a_1,red,b_1,r_2), (w_{red},e_m,e_s))$, depicted in Figure \ref{fig:second_update}.

\begin{figure}[H]
  \begin{center}
    \scalebox{0.7}{
            \begin{circuitikz}[framed]
            \draw
            (4,2.75) node[below] {$\boldsymbol{(w_{red},e_m,e_s)}$}
            (1,0) node[below] {$\rho_1$}
            (3,0) node[below] {$\rho_2$}
            (5,0) node[below] {$\rho_3$}
            (7,0) node[below] {$\rho_4$}
            (2,0) -- (0,0) -- (0,2) -- (2,2)
            (2,2) to[nos] (2,0)
            (3,1.35) node[below] {\circled{$\alpha_1$}}
            (2,2) -- (4,2) 
            (2,0) -- (4,0)
            (4,2) to[nos] (4,0)
            (3,2) node[below] {\squaredr{$\beta_1$}}
            (4,2) -- (6,2) 
            (4,0) -- (6,0)
            (6,2) to[nos] (6,0)
            (5,1.35) node[below] {\circled{$\alpha_2$}}
            (6,2) -- (8,2) 
            (6,0) -- (8,0)
            (8,2) -- (8,0)
            (7,1.35) node[below] {\circled{$\alpha_3$}}
            ;
            \draw
            (8.25,1) -- (9.75,1) node[midway,
    above]{$\alpha_2,\alpha_3$}
            (8.25,-3.25) -- (9.75,-3.25) node[midway,
    above]{$\alpha_1, \alpha_2,\alpha_3$}
            (14,-0.25) -- (14,-1.25) node[midway,
    right]{$\alpha_2,\alpha_3$}
            (8.25,5) -- (9.75,5) node[midway,
    above]{$\alpha_1,\alpha_2,\alpha_3$}
            (4,-0.25) -- (4,-1.25) node[midway,
    right]{$\alpha_2,\alpha_3$}
            (4,2.75) -- (4,3.75) node[midway,
    right]{$\alpha_2,\alpha_3$}
            (14,2.75) -- (14,3.75) node[midway,
    right]{$\alpha_2,\alpha_3$}
            ;
            \begin{scope}[xshift=10cm]
            \draw
            (1,0) node[below] {$\rho_1$}
            (3,0) node[below] {$\rho_2$}
            (4,2.75) node[below] {$(w_{green},e_m,e'_s)$}
            (5,0) node[below] {$\rho_3$}
            (7,0) node[below] {$\rho_4$}
            (2,0) -- (0,0) -- (0,2) -- (2,2)
            (2,2) to[nos] (2,0)
            (3,1.35) node[below] {\circled{$\alpha_1$}}
            (2,2) -- (4,2) 
            (2,0) -- (4,0)
            (4,2) to[nos] (4,0)
            (3,2) node[below] {\squaredg{$\beta_1$}}
            (4,2) -- (6,2) 
            (4,0) -- (6,0)
            (6,2) to[nos] (6,0)
            (5,1.35) node[below] {\circled{$\alpha_2$}}
            (6,2) -- (8,2) 
            (6,0) -- (8,0)
            (8,2) -- (8,0)
            (7,1.35) node[below] {\circled{$\alpha_3$}}
            ;
            \begin{scope}[yshift=-4cm]
            \draw
            (1,0) node[below] {$\rho_1$}
            (3,0) node[below] {$\rho_2$}
            (4,2.2) node[above] {$(w_{green},e'_m,e''_s)$}
            (5,0) node[below] {$\rho_3$}
            (7,0) node[below] {$\rho_4$}
            (2,0) -- (0,0) -- (0,2) -- (2,2)
            (2,2) to[nos] (2,0)
            (1,1.35) node[below] {\circled{$\alpha_1$}}
            (2,2) -- (4,2) 
            (2,0) -- (4,0)
            (4,2) to[nos] (4,0)
            (3,2) node[below] {\squaredg{$\beta_1$}}
            (4,2) -- (6,2) 
            (4,0) -- (6,0)
            (6,2) to[nos] (6,0)
            (5,1.35) node[below] {\circled{$\alpha_2$}}
            (6,2) -- (8,2) 
            (6,0) -- (8,0)
            (8,2) -- (8,0)
            (7,1.35) node[below] {\circled{$\alpha_3$}}
            ;
            \begin{scope}[xshift=-10cm]
            \draw
            (1,0) node[below] {$\rho_1$}
            (3,0) node[below] {$\rho_2$}
            (4,2.2) node[above] {$(w_{red},e'_m,e''_s)$}
            (5,0) node[below] {$\rho_3$}
            (7,0) node[below] {$\rho_4$}
            (2,0) -- (0,0) -- (0,2) -- (2,2)
            (2,2) to[nos] (2,0)
            (1,1.35) node[below] {\circled{$\alpha_1$}}
            (2,2) -- (4,2) 
            (2,0) -- (4,0)
            (4,2) to[nos] (4,0)
            (3,2) node[below] {\squaredr{$\beta_1$}}
            (4,2) -- (6,2) 
            (4,0) -- (6,0)
            (6,2) to[nos] (6,0)
            (5,1.35) node[below] {\circled{$\alpha_2$}}
            (6,2) -- (8,2) 
            (6,0) -- (8,0)
            (8,2) -- (8,0)
            (7,1.35) node[below] {\circled{$\alpha_3$}}
            ;
            \begin{scope}[yshift=8cm]
            \draw
            (4,2.75) node[below] {${(w_{red},e_m,e''_s)}$}
            (1,0) node[below] {$\rho_1$}
            (3,0) node[below] {$\rho_2$}
            (5,0) node[below] {$\rho_3$}
            (7,0) node[below] {$\rho_4$}
            (2,0) -- (0,0) -- (0,2) -- (2,2)
            (2,2) to[nos] (2,0)
            (3,1.35) node[below] {\circled{$\alpha_1$}}
            (2,2) -- (4,2) 
            (2,0) -- (4,0)
            (4,2) to[nos] (4,0)
            (3,2) node[below] {\squaredr{$\beta_1$}}
            (4,2) -- (6,2) 
            (4,0) -- (6,0)
            (6,2) to[nos] (6,0)
            (5,1.35) node[below] {\circled{$\alpha_2$}}
            (6,2) -- (8,2) 
            (6,0) -- (8,0)
            (8,2) -- (8,0)
            (7,1.35) node[below] {\circled{$\alpha_3$}}
            ;
            \begin{scope}[xshift=10cm]
            \draw
            (4,2.75) node[below] {${(w_{green},e_m,e''_s)}$}
            (1,0) node[below] {$\rho_1$}
            (3,0) node[below] {$\rho_2$}
            (5,0) node[below] {$\rho_3$}
            (7,0) node[below] {$\rho_4$}
            (2,0) -- (0,0) -- (0,2) -- (2,2)
            (2,2) to[nos] (2,0)
            (3,1.35) node[below] {\circled{$\alpha_1$}}
            (2,2) -- (4,2) 
            (2,0) -- (4,0)
            (4,2) to[nos] (4,0)
            (3,2) node[below] {\squaredg{$\beta_1$}}
            (4,2) -- (6,2) 
            (4,0) -- (6,0)
            (6,2) to[nos] (6,0)
            (5,1.35) node[below] {\circled{$\alpha_2$}}
            (6,2) -- (8,2) 
            (6,0) -- (8,0)
            (8,2) -- (8,0)
            (7,1.35) node[below] {\circled{$\alpha_3$}}
            ;
            \end{scope}
            \end{scope}
            \end{scope}
            \end{scope}
            \end{scope}
            \end{circuitikz}
    }
    \end{center}
    \caption{\label{fig:second_update} The pointed model $(M_0 \otimes \text{Move}(a_1,r_1,r_2)\otimes\text{SenseCol}(a_1,red,b_1,r_2), (w_{red},e_m,e_s))$, representing the state after $\alpha_1$ senses that $\beta_1$ is red while in room $\rho_2$. Edges in the reflexive-transitive closure of the indistinguishability relations are omitted. Note that, in the actual world, $(w_{red},e_m,e_s)$, $\alpha_1$ does not face any uncertainty. In particular, $K_{a_1}\mathsf{Color}(b_1,red)$ is true at $(w_{red},e_m,e_s)$. On the other hand, $\alpha_2$ and $\alpha_3$ have not observed any of $\alpha_1$'s actions. That is, $\alpha_1$ may or may not have moved, and it may or may not have sensed whether $\beta_1$ is red. As a result, it holds at $(w_{red},e_m,e_s)$ that $\forall x (x\neq a_1 \to \neg K_x \mathsf{In}(a_1,r_2) \wedge \neg K_x \mathsf{Color}(b_1,red))$.}
\end{figure}
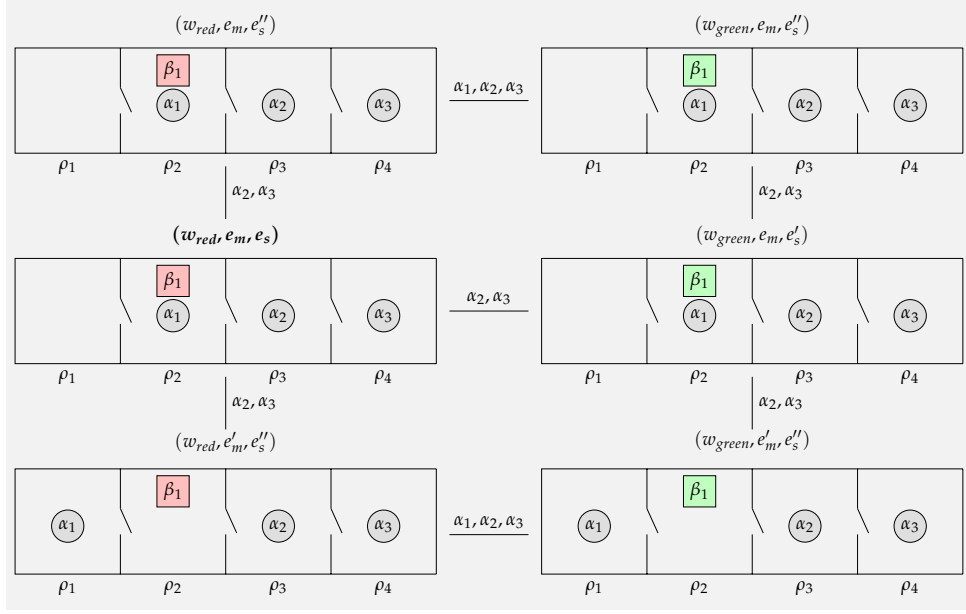

Finally, $\alpha_1$ announces in room $\rho_2$ that the color of $\beta_1$ is red. The result of this action is the model $(M_0 \otimes \text{Move}(a_1,r_1,r_2)\otimes\text{SenseCol}(a_1,red,b_1,r_2) \otimes \text{Announce}(a_1,red,b_1,r_2), (w_{red},e_m,e_s,e_a))$, depicted in Figure \ref{fig:third_update}.

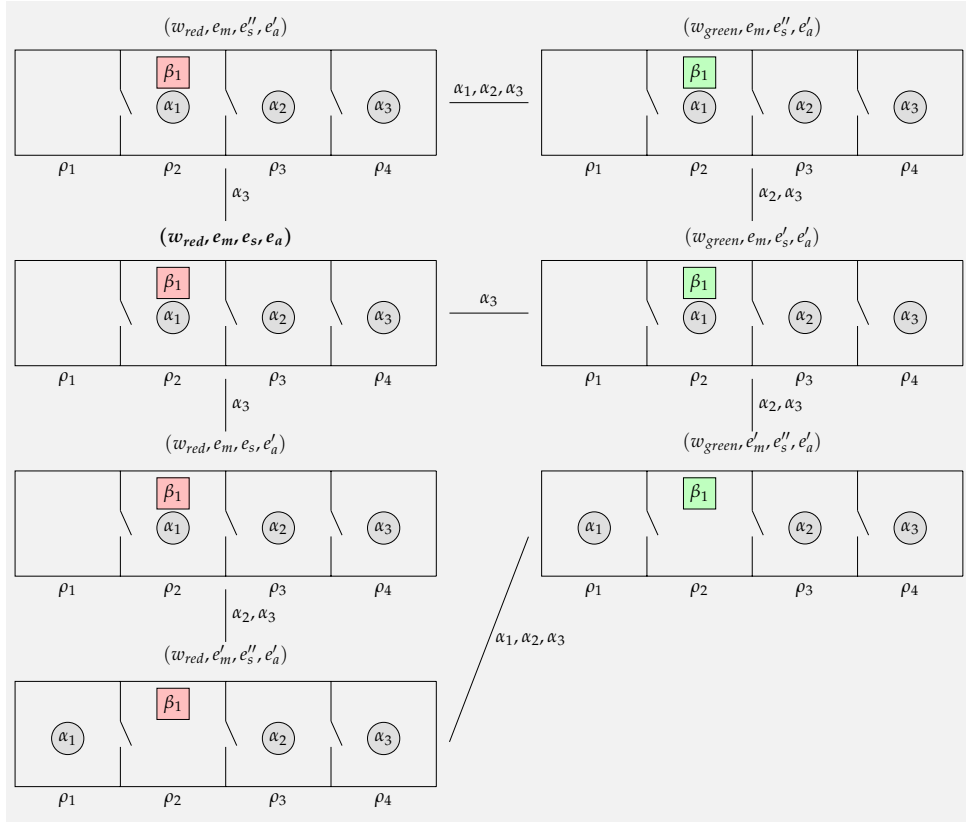
\begin{figure}[H]
  \begin{center}
    \scalebox{0.7}{
            \begin{circuitikz}[framed]
            \draw
            (4,2.75) node[below] {$\boldsymbol{(w_{red},e_m,e_s,e_a)}$}
            (1,0) node[below] {$\rho_1$}
            (3,0) node[below] {$\rho_2$}
            (5,0) node[below] {$\rho_3$}
            (7,0) node[below] {$\rho_4$}
            (2,0) -- (0,0) -- (0,2) -- (2,2)
            (2,2) to[nos] (2,0)
            (3,1.35) node[below] {\circled{$\alpha_1$}}
            (2,2) -- (4,2) 
            (2,0) -- (4,0)
            (4,2) to[nos] (4,0)
            (3,2) node[below] {\squaredr{$\beta_1$}}
            (4,2) -- (6,2) 
            (4,0) -- (6,0)
            (6,2) to[nos] (6,0)
            (5,1.35) node[below] {\circled{$\alpha_2$}}
            (6,2) -- (8,2) 
            (6,0) -- (8,0)
            (8,2) -- (8,0)
            (7,1.35) node[below] {\circled{$\alpha_3$}}
            ;
            \draw
            (8.25,1) -- (9.75,1) node[midway,
    above]{$\alpha_3$}
            (8.25,-7.15) -- (9.75,-3.25) node[midway,
    right]{$\alpha_1,\alpha_2,\alpha_3$}
            (14,-0.25) -- (14,-1.25) node[midway,
    right]{$\alpha_2,\alpha_3$}
            (8.25,5) -- (9.75,5) node[midway,
    above]{$\alpha_1,\alpha_2,\alpha_3$}
            (4,-0.25) -- (4,-1.25) node[midway,
    right]{$\alpha_3$}
            (4,2.75) -- (4,3.75) node[midway,
    right]{$\alpha_3$}
            (14,2.75) -- (14,3.75) node[midway,
    right]{$\alpha_2,\alpha_3$}
            (4,-4.25) -- (4,-5.25) node[midway,
    right]{$\alpha_2,\alpha_3$}
            ;
            \begin{scope}[xshift=10cm]
            \draw
            (1,0) node[below] {$\rho_1$}
            (3,0) node[below] {$\rho_2$}
            (4,2.75) node[below] {$(w_{green},e_m,e'_s,e'_a)$}
            (5,0) node[below] {$\rho_3$}
            (7,0) node[below] {$\rho_4$}
            (2,0) -- (0,0) -- (0,2) -- (2,2)
            (2,2) to[nos] (2,0)
            (3,1.35) node[below] {\circled{$\alpha_1$}}
            (2,2) -- (4,2) 
            (2,0) -- (4,0)
            (4,2) to[nos] (4,0)
            (3,2) node[below] {\squaredg{$\beta_1$}}
            (4,2) -- (6,2) 
            (4,0) -- (6,0)
            (6,2) to[nos] (6,0)
            (5,1.35) node[below] {\circled{$\alpha_2$}}
            (6,2) -- (8,2) 
            (6,0) -- (8,0)
            (8,2) -- (8,0)
            (7,1.35) node[below] {\circled{$\alpha_3$}}
            ;
            \begin{scope}[yshift=-4cm]
            \draw
            (1,0) node[below] {$\rho_1$}
            (3,0) node[below] {$\rho_2$}
            (4,2.2) node[above] {$(w_{green},e'_m,e''_s,e'_a)$}
            (5,0) node[below] {$\rho_3$}
            (7,0) node[below] {$\rho_4$}
            (2,0) -- (0,0) -- (0,2) -- (2,2)
            (2,2) to[nos] (2,0)
            (1,1.35) node[below] {\circled{$\alpha_1$}}
            (2,2) -- (4,2) 
            (2,0) -- (4,0)
            (4,2) to[nos] (4,0)
            (3,2) node[below] {\squaredg{$\beta_1$}}
            (4,2) -- (6,2) 
            (4,0) -- (6,0)
            (6,2) to[nos] (6,0)
            (5,1.35) node[below] {\circled{$\alpha_2$}}
            (6,2) -- (8,2) 
            (6,0) -- (8,0)
            (8,2) -- (8,0)
            (7,1.35) node[below] {\circled{$\alpha_3$}}
            ;
            \begin{scope}[xshift=-10cm]
            \draw
            (1,0) node[below] {$\rho_1$}
            (3,0) node[below] {$\rho_2$}
            (4,2.2) node[above] {$(w_{red},e_m,e_s,e'_a)$}
            (5,0) node[below] {$\rho_3$}
            (7,0) node[below] {$\rho_4$}
            (2,0) -- (0,0) -- (0,2) -- (2,2)
            (2,2) to[nos] (2,0)
            (3,1.35) node[below] {\circled{$\alpha_1$}}
            (2,2) -- (4,2) 
            (2,0) -- (4,0)
            (4,2) to[nos] (4,0)
            (3,2) node[below] {\squaredr{$\beta_1$}}
            (4,2) -- (6,2) 
            (4,0) -- (6,0)
            (6,2) to[nos] (6,0)
            (5,1.35) node[below] {\circled{$\alpha_2$}}
            (6,2) -- (8,2) 
            (6,0) -- (8,0)
            (8,2) -- (8,0)
            (7,1.35) node[below] {\circled{$\alpha_3$}}
            ;
            \begin{scope}[yshift=8cm]
            \draw
            (4,2.75) node[below] {${(w_{red},e_m,e''_s,e'_a)}$}
            (1,0) node[below] {$\rho_1$}
            (3,0) node[below] {$\rho_2$}
            (5,0) node[below] {$\rho_3$}
            (7,0) node[below] {$\rho_4$}
            (2,0) -- (0,0) -- (0,2) -- (2,2)
            (2,2) to[nos] (2,0)
            (3,1.35) node[below] {\circled{$\alpha_1$}}
            (2,2) -- (4,2) 
            (2,0) -- (4,0)
            (4,2) to[nos] (4,0)
            (3,2) node[below] {\squaredr{$\beta_1$}}
            (4,2) -- (6,2) 
            (4,0) -- (6,0)
            (6,2) to[nos] (6,0)
            (5,1.35) node[below] {\circled{$\alpha_2$}}
            (6,2) -- (8,2) 
            (6,0) -- (8,0)
            (8,2) -- (8,0)
            (7,1.35) node[below] {\circled{$\alpha_3$}}
            ;
            \begin{scope}[xshift=10cm]
            \draw
            (4,2.75) node[below] {${(w_{green},e_m,e''_s,e'_a)}$}
            (1,0) node[below] {$\rho_1$}
            (3,0) node[below] {$\rho_2$}
            (5,0) node[below] {$\rho_3$}
            (7,0) node[below] {$\rho_4$}
            (2,0) -- (0,0) -- (0,2) -- (2,2)
            (2,2) to[nos] (2,0)
            (3,1.35) node[below] {\circled{$\alpha_1$}}
            (2,2) -- (4,2) 
            (2,0) -- (4,0)
            (4,2) to[nos] (4,0)
            (3,2) node[below] {\squaredg{$\beta_1$}}
            (4,2) -- (6,2) 
            (4,0) -- (6,0)
            (6,2) to[nos] (6,0)
            (5,1.35) node[below] {\circled{$\alpha_2$}}
            (6,2) -- (8,2) 
            (6,0) -- (8,0)
            (8,2) -- (8,0)
            (7,1.35) node[below] {\circled{$\alpha_3$}}
            ;
            \end{scope}
            \begin{scope}[yshift=-12cm]
            \draw
            (4,2.2) node[above] {${(w_{red},e'_m,e''_s,e'_a)}$}
            (1,0) node[below] {$\rho_1$}
            (3,0) node[below] {$\rho_2$}
            (5,0) node[below] {$\rho_3$}
            (7,0) node[below] {$\rho_4$}
            (2,0) -- (0,0) -- (0,2) -- (2,2)
            (2,2) to[nos] (2,0)
            (1,1.35) node[below] {\circled{$\alpha_1$}}
            (2,2) -- (4,2) 
            (2,0) -- (4,0)
            (4,2) to[nos] (4,0)
            (3,2) node[below] {\squaredr{$\beta_1$}}
            (4,2) -- (6,2) 
            (4,0) -- (6,0)
            (6,2) to[nos] (6,0)
            (5,1.35) node[below] {\circled{$\alpha_2$}}
            (6,2) -- (8,2) 
            (6,0) -- (8,0)
            (8,2) -- (8,0)
            (7,1.35) node[below] {\circled{$\alpha_3$}}
            ;
            \end{scope}
            \end{scope}
            \end{scope}
            \end{scope}
            \end{scope}
            \end{circuitikz}
        }
    \end{center}
    \caption{\label{fig:third_update} The pointed model representing the state after $\alpha_1$ announces that $\beta_1$ is red while at room $\rho_2$, $(M_0 \otimes \text{Move}(a_1,r_1,r_2)\otimes\text{SenseCol}(a_1,red,b_1,r_2) \otimes \text{Announce}(a_1,red,b_1,r_2), (w_{red},e_m,e_s,e_a))$. Edges in the reflexive-transitive closure of the indistinguishability relations are omitted. Note that, in the actual world, $(w_{red},e_m,e_s,e_a)$, $\alpha_1$ and $\alpha_2$ do not face any uncertainty; there are no outgoing edges from the actual world for these agents. The goal $g$ as stated in Section \ref{sec:running_example} holds in the actual world: both $\alpha_1$ and $\alpha_2$ know the color of $\beta_1$, $\alpha_1$ knows that $\alpha_2$ knows this, and $\alpha_1$ knows that $\alpha_3$ does not know this.}
\end{figure}

\end{example}

\subsection{Succinct Representation of Actions via Epistemic Action Schemas \label{subsec:action_schemas}}

We introduce epistemic action schemas, which represent sets of actions in a general way, as done in common planning formalisms such as PDDL. Schemas use \textit{variables} to describe actions, rather than \textit{constant symbols}. These variables denote arbitrary agents and objects and are used to describe their roles with respect to a type of action, such as the roles of \textit{speaker} and \textit{listener} in an action of type `announcement'. 

As anticipated in Section \ref{sec:static_language}, a major reason for introducing schemas is that they result in action representations whose size is independent of the number of agents and objects in a domain. For the \textit{SelectiveCommunication} domain $\text{SC}(n,m,k,\ell)$, there are $n\cdot m \cdot k \cdot \ell \cdot 2^{n-1}$ possible announcement actions, since each of the $n$ agents could, in each of the $m$ rooms, announce about each of the $k$ boxes, that it is of one out of $\ell$ colors, with one out of the $2^{n-1}$ subsets of the other agents hearing the announcement. Representing all actions requires $n\cdot m \cdot k \cdot \ell \cdot 2^{n-1}$ standard DEL action models, i.e., one model per action. Variants of standard DEL models such as edge-conditioned models \cite{Bolander2014} fare substantially better, since the set of hearers is implicitly represented in such models, but $n\cdot m \cdot k \cdot \ell$ models are still required to represent the set of announcements. Other variants of DEL are also more succinct than standard DEL action models, e.g. the symbolic models of \cite{charrier2017succinct, van2018symbolic}. However, all these announcements can be compactly represented with a single \textit{epistemic action schema}, as shown below in Example \ref{example:action_schemas}.

Moreover, epistemic schemas open up the possibility of applying well-known techniques such as \textit{least commitment} or \textit{partial order planning} \cite{weld1994introduction} to epistemic problems. These approaches use the notion of a \textit{partially instantiated action}, such as $Move(B,x,C)$, where $x$ is a variable whose substitution has not yet been chosen. If specifying a binding constraint for $x$ is unnecessary at the current point in the planning process, it is often advantageous to delay this commitment until later, i.e., until other necessary parts of the plan are discovered that further constrain what $x$ should be. 
Other approaches to lifted planning, such as \textit{hierarchical task networks} (HTNs) (\cite{ghallab2004automated}, ch. 11) similarly exploit partial substitutions to optimise the search for solutions. For an epistemic version of lifted HTN planning, the epistemic schemas defined here could play the role of primitive tasks. 

An epistemic action schema $a(x_1,\dots, x_n)$ is defined using an \textit{action name} $a$ and a \textit{parameter list} $(x_1, \dots, x_n)$, as done e.g. in PDDL. The parameter list fixes a finite set of agents and objects involved in the execution of the action. Schemas are required to follow a STRIPS-like scope assumption; all variables referenced in the preconditions or postconditions of an action schema must appear in the action's parameter list. Any agent or object unmentioned in the parameter list is assumed to be unrelated to the action's pre- and postconditions.

\begin{defn}
\label{Def. Action schema} 
An \textbf{\textit{epistemic action schema}} is of the form $a(\vec{x})=(E,Q,\mathsf{pre},\mathsf{post})$ where
\begin{enumerate}
    \item $a$ is the \textit{action name} and $\vec{x}\in \mathtt{V}^n$ is a finite \textit{parameter list}.
    \item $E$ is a non-empty, finite set of \textit{events}.
    \item $Q: (E\times E) \to \mathcal{L}_{AM}$ is an \textit{edge-condition function}, where the formula $Q(e,e')$ has a free variable ${x^\star}$ of type $\mathtt{agt}$, and possibly other free variables all in $\vec{x}$.
    \item $\mathsf{pre}: E \to \mathcal{L}_{AM}$ assigns to each event a \textit{precondition formula} with all free variables in $\vec{x}$.
    \item $\mathsf{post}: E \to (\mathtt{FreeAtoms}(\mathcal{L}) \rightharpoonup\mathcal{L}_{AM})$ assigns to each event a \textit{partial postcondition function} such that if $y_1,\dots,y_m$ all occur in $\vec{x}$, then $\mathsf{post}(e)(r(y_1,\dots,y_m))$ has all free variables from $\vec{x}$; else, $\mathsf{post}(e)(r(y_1,\dots,y_m))$ is undefined.
\end{enumerate}
 $\dom( \mathsf{post}(e))$ denotes the set of atoms for which $\mathsf{post}(e)(r(t_1,\dots,t_k)) \neq r(t_1,\dots,t_k)$.
\end{defn}

The postcondition for each event $e$ is defined as a partial function, with all atoms whose arguments are not a subset of those ocurring in $\vec{x}$ left unaffected. Since the parameter list is required to be finite, this yields a finite encoding of postconditions. 

Note that the parameter list of a schema may include \textit{agent variables}. Just like any other action parameter, these variables can appear in the preconditions or effects of the schema. This is in line with what occurs in multi-agent extensions of PDDL such as MAPL or MA-PDDL \cite{brenner2003multiagent,kovacs2012multi}. Such variables are included to enable the schematization of actions also with respect to agents. In this paper, we adopt the epistemic operators from term-modal logic, \textit{indexed by agent variables}, to achieve schematization with respect to agents when formalizing epistemic planning. In order to be able to express epistemic pre- or postconditions relative to an agent variable $x$ in an action schema, it is necessary to use the variable in the scope of a modal operator.  For example, if an action has e.g. a precondition that requires agent $x$ to know $P(y,z)$, we need the term-modal formula $K_x P(y,z)$ to express this constraint in a schematic way. 

\begin{example}[Action schemas for $\text{SC}(n,m,k,\ell)$] \label{example:action_schemas} Figures  \ref{fig:schema_move}, \ref{fig:schema_sense} and \ref{fig:schema_say} depict graphically the action schemas for the movement, sensing and announcement actions described in the example from Section \ref{sec:running_example}. The schemas have the same structure as the action models from Example \ref{example:action_models}, but the conditions are now expressed in a general way, via free variables for agents, boxes and colors. 

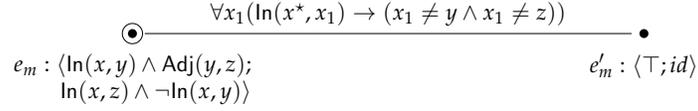
\begin{figure}[H]
    \begin{center}
        \scalebox{0.85}{
            \begin{tikzpicture}
              \tikzstyle{dnode}=[inner sep=1pt,outer sep=1pt,draw,circle,minimum width=9pt]
              \tikzstyle{nnode}=[inner sep=1pt,outer sep=1pt,circle,minimum width=9pt]
              \tikzstyle{label-edge}=[midway,fill=white, inner sep=1pt]
              \node[dnode, label={[align=right]below: $e_m: \langle \mathsf{In}(x,y) \wedge \mathsf{Adj}(y,z);$ \\ $\mathsf{In}(x,z)\wedge \neg \mathsf{In}(x,y) \rangle$}] (w0) at (0,0) {};
              \fill (w0) circle [radius=2pt];
              \node[nnode,  label={below: $e'_m: \langle \top;id \rangle$}] (w1) at (8,0) {};
              \fill (w1) circle [radius=2pt];
              \draw[-] (w0) -- (w1) node[above,midway] {$\forall x_1 (\mathsf{In}({x^\star},x_1) \to (x_1 \neq y \wedge x_1 \neq z))$};
            \end{tikzpicture}
        }
    \end{center}
\caption{\label{fig:schema_move} $\text{Move}(x,y,z)$, the action schema for agent $x$ moving from room $y$ to room $z$.}
\end{figure}

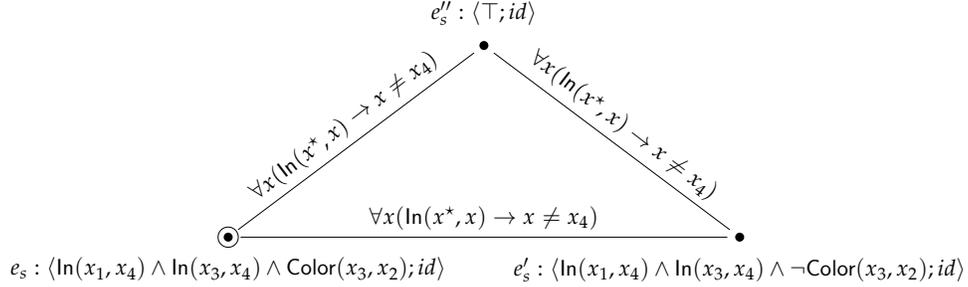
\begin{figure}[H]
    \begin{center}
        \scalebox{0.85}{
            \begin{tikzpicture}
              \tikzstyle{dnode}=[inner sep=1pt,outer sep=1pt,draw,circle,minimum width=9pt]
              \tikzstyle{nnode}=[inner sep=1pt,outer sep=1pt,circle,minimum width=9pt]
              \tikzstyle{label-edge}=[midway,fill=white, inner sep=1pt]
              \node[dnode, label={[align=right]below: $e_{s}: \langle \mathsf{In}(x_1,x_4) \wedge \mathsf{In}(x_3,x_4)\wedge \mathsf{Color}(x_3,x_2); id \rangle$}] (w0) at (0,0) {};
              \fill (w0) circle [radius=2pt];
              \node[nnode,  label={below: $e'_s:   \langle \mathsf{In}(x_1,x_4) \wedge \mathsf{In}(x_3,x_4)\wedge \neg \mathsf{Color}(x_3,x_2); id \rangle$}] (w1) at (8,0) {};
              \fill (w1) circle [radius=2pt];
              \node[nnode,  label={above: $e''_s:  \langle \top; id \rangle$}] (w2) at (4,3) {};
              \fill (w2) circle [radius=2pt];
              \draw[-] (w0) -- (w1) node[above,midway] {$\forall x (\mathsf{In}({x^\star},x) \to  x \neq x_4)$};
              \draw[-] (w0) -- (w2) node[above,midway,rotate=37.5] {$\forall x (\mathsf{In}({x^\star},x) \to x \neq x_4)$};
              \draw[-] (w1) -- (w2) node[above,midway,rotate=-37.5] {$\forall x (\mathsf{In}({x^\star},x) \to x \neq x_4)$};
            \end{tikzpicture}
        }
    \end{center}
\caption{\label{fig:schema_sense} $\text{SenseCol}(x_1,x_2,x_3,x_4)$, the action schema for $x_1$ sensing in room $x_4$ whether box $x_3$ has color $x_2$.}
\end{figure}

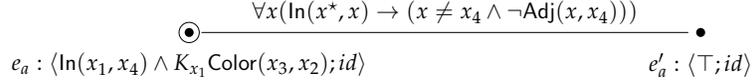
\begin{figure}[H]
    \begin{center}
        \scalebox{0.85}{
            \begin{tikzpicture}
              \tikzstyle{dnode}=[inner sep=1pt,outer sep=1pt,draw,circle,minimum width=9pt]
              \tikzstyle{nnode}=[inner sep=1pt,outer sep=1pt,circle,minimum width=9pt]
              \tikzstyle{label-edge}=[midway,fill=white, inner sep=1pt]
              \node[dnode, label={[align=right]below: $e_a: \langle \mathsf{In}(x_1,x_4) \wedge K_{x_1}\mathsf{Color}(x_3,x_2); id \rangle$}] (w0) at (0,0) {};
              \fill (w0) circle [radius=2pt];
              \node[nnode,  label={below: $e'_a: \langle \top;id \rangle$}] (w1) at (8,0) {};
              \fill (w1) circle [radius=2pt];
              \draw[-] (w0) -- (w1) node[above,midway] {$\forall x (\mathsf{In}({x^\star},x) \to ( x \neq x_4 \wedge \neg \mathsf{Adj}(x,x_4)))$};
            \end{tikzpicture}
        }
    \end{center}
\caption{\label{fig:schema_say} $\text{Announce}(x_1,x_2,x_3,x_4)$, the action schema for $x_1$ announcing that $x_3$ has color $x_2$ while in room $x_4$.}
\end{figure}

\end{example}

As usual in planning, schemas can be instantiated into concrete actions via grounding substitutions. Schema instantiation is defined as follows. Let $a(x_1,\dots,x_n)$ be a schema and $\sigma: \{x_1,\dots,x_n\} \to \mathtt{C}$ be a \textbf{\textit{grounding substitution}}, i.e., a mapping from variables into constants. For a formula $\varphi$, let $\varphi\sigma$ be the result of replacing each occurrence of a free variable $y$ in $\varphi$ by $\sigma(y)$.
\begin{defn}
\label{Def. schema instantiation} Let $a(x_1,\dots,x_n)=(E,Q,\mathsf{pre},\mathsf{post})$  be an action schema and let $\sigma: \{x_1,\dots,x_n\} \to \mathtt{C}$ be a grounding substitution. The \textbf{\textit{action model induced by}} $\sigma$ is  $a(\sigma(x_1),\dots,\sigma(x_n))=(E',Q',\mathsf{pre}',\mathsf{post}')$ where
\begin{enumerate}
\item $E'=E$.
\item for each $e,e'\in E$, $Q'(e,e')=Q(e,e')\sigma$ 
\item for each $e\in E$, $\mathsf{pre}'(e)=\mathsf{pre}(e)\sigma$ .
\item for each $e\in E$, \[ \mathsf{post}'(e)(r(t_1,\dots, t_n)\sigma) = \begin{cases} 
      \mathsf{post}(e)(r(t_1,\dots, t_n))\sigma & \text{if } \mathsf{post}(e)(r(t_1,\dots, t_n)) \text{ is defined }  \\
      r(t_1,\dots, t_n)\sigma & \text{otherwise.}
   \end{cases}
\]
\end{enumerate}
\end{defn}

\begin{example}
The announcement model $A_1$ in Figure \ref{fig:model_say} is the ground action of schema $S_1$ from Figure \ref{fig:schema_say} induced by the substitution $\sigma = \{ x_1\mapsto a_2, x_2\mapsto r_3, x_3\mapsto b_1, x_4\mapsto green\}$.
\end{example}

\section{Problems, Plans and Solutions}\label{sec:planning} 

This section defines \textit{first-order epistemic planning tasks}. An epistemic planning task consists of an initial state, a set of actions, and a goal to be achieved. To solve an epistemic planning task, one may take either an \textit{external} or and \textit{internal} perspective \cite{aucher2010internal}. The external perspective is the view of the system designer, who knows the precise initial state and actual effect of every action. The internal perspective is the view of an in-system agent, who has uncertainty about the state of the world and therefore uncertainty about the effects of executed actions. In the DEL planning framework, an external planning task is defined as a special case of a classical planning task (as in, e.g., \cite{aucher2013undecidability,bolander2015complexity}).  Following \cite{ghallab2004automated}, any \textit{classical planning domain} can be described as a state-transition system $T = (\mathcal{S}, \mathcal{A}, \gamma)$ where $\mathcal{S}$ is a finite or recursively enumerable set of finite states, $\mathcal{A}$ is a finite set of actions and $\gamma : \mathcal{S} \times \mathcal{A} \rightharpoonup \mathcal{S}$ is a partial, computable state-transition function. A \textit{classical planning task} is a triple $(T, s_0, S_G)$, where $T$ is a state-transition system, $s_0\in \mathcal{S}$ is the initial state and $S_G\subseteq \mathcal{S}$ is the set of goal states. A \textit{solution} to a classical planning task $(T, s_0, S_G)$ is a plan consisting of a finite sequence of actions $a_1, a_2,\dots,a_n$ such that (1) For all $i \leq n$, $\gamma(\gamma(\dots \gamma(\gamma(s_0, a_1), a_2), \dots , a_{i-1}), a_i)$ is defined, and (2) $\gamma(\gamma(\dots \gamma(\gamma(s_0, a_1), a_2), \dots , a_{n-1}), a_n) \in S_G$. Epistemic planning tasks can be defined as special cases of classical planning tasks, as follows. 

\begin{defn}
\label{def.ext.planning.task}
Let $\mathsf{A}$ be a finite set of action schemas. A \textit{(first-order) epistemic planning task} based on $\mathsf{A}$ is a triple $P=(s_0, \mathcal{A}, \varphi_G)$ where the initial state $s_0$ is a finite epistemic state with a finite domain, $\mathcal{A}$ is the set of all ground instances of the schemas in $\mathsf{A}$, and the goal formula $\varphi_G$ is a sentence of $\mathcal{L}$. Any epistemic planning task $(s_0, \mathcal{A}, \varphi_G)$ induces a classical planning task $((\mathcal{S}, \mathcal{A}, \gamma), s_0, S_G)$ given by:
\begin{description}
    \item[-] $\mathcal{S} \coloneqq \{s_0 \otimes a_1 \otimes \dots \otimes a_n \mid n\in\mathbb{N}, a_i \in \mathcal{A}\}$
    \item[-] $S_G \coloneqq \{s \in \mathcal{S} \mid s \vDash \varphi_G\}$
    \item[-] $\gamma(s,a) \coloneqq  
    s\otimes a$ if $a$ is applicable in $s$, else undefined.

\end{description}
A \textit{solution} to an epistemic planning task is a solution to the induced classical planning task.
\end{defn}

All the ingredients in this definition of an external planning task can come from the formalism presented here. Note that the planner-modeler in such a task is not one of the agents in the domain $D_{\mathtt{agt}}$. The planner-modeler has access to the actual states $s_i$, i.e., to pointed models $(M,w)$ where $w$ is the actual world. 

Formalisms of internal epistemic planning based on DEL are often defined from the external planning model, either by adding structure to the models or making small modifications. For example, \cite{andersen2012conditional} represents internal perspectives using \textit{information cells}, which are defined from the accessibility relations of an epistemic model. An alternative involves using \textit{multi-pointed models} or adding a set of so-called \textit{designated points} to the epistemic model, with each point describing a world that the agent considers as possible from its internal perspective (see e.g. \cite{bolanderbirkegaard2011}). A third approach uses a belief state representation of the agent's internal view as primitive and then defines an epistemic model from it \cite{kominis2015beliefs}. The approach in \cite{aucher2010internal} offers two different flavors of internal view, both defined  on the basis of a standard epistemic model. These various notions of internal perspective, as well as their associated planning tasks, may be upgraded to our framework without major modifications. We believe that the simplest way to add internal perspectives to the present would be to adopt the approach proposed by Bolander and Andersen \cite{bolanderbirkegaard2011}. Save for the fact that such models are multi-pointed, the core semantics remains the same. We therefore envision no major difficulty in bringing the internal perspective into our formalism; similarly defined multi-pointed structures should suffice. We do not develop the detail here.

\subsection{Decidability of the Plan Existence Problem}

Having defined first-order epistemic planning tasks, a natural first question is whether the corresponding plan existence problem is decidable. We follow Aucher and Bolander \cite{aucher2013undecidability} in defining the plan existence problem:

\begin{defn}
Let $n\in \mathbb{N}$. $\mathsf{PlanEx}(n)$ is the problem: ``Given a (first-order) epistemic planning task $P=(s_0, \mathcal{A}, \varphi_G)$ where $s_0$ is an $n$-agent epistemic state, does $P$ have a solution?''.
\end{defn}

For propositional DEL, the corresponding problem is undecidable in general \cite{bolanderbirkegaard2011}. This entails that the unrestricted first-order problem is undecidable as well, since first-order epistemic planning extends propositional DEL planning. However, decidable and reasonably expressive fragments of propositional DEL planning have been found, such as single-agent planning and multi-agent planning with non-modal preconditions. In \cite{ijcai2020}, we show that \textit{the corresponding first-order fragments are also decidable}. In that paper, bisimulations for the term-modal models presented here are introduced, and shown to have standard model-theoretic properties.  Such bisimulations are key in proving the decidability results, as they allow us to show that the state spaces for certain planning fragments are finitely representable, up to bisimulation. We state the main results here and refer the reader to \cite{ijcai2020} for details.

\begin{thm}[\cite{ijcai2020}] $\mathsf{PlanEx}(1)$ (single-agent planning) is decidable.
\end{thm}{}

\begin{thm}[\cite{ijcai2020}] If all actions have non-modal preconditions, then $\mathsf{PlanEx}(k)$ is decidable, for $k\geq 1$.
\end{thm}{}

\subsection{An Example with a PDDL-like Description}\label{sec:example.non-rigid}

With the above, the dynamic term-modal planning framework of the paper has been introduced. This section contains a second example, using its different components in one place. The example also serves to illustrate how uncertainty about names may play a role in epistemic planning, and how a term-modal planning domain and an associated planning problem may be described using a `PDDL-like syntax'. This description is meant as an indication of how such definitions could be standardized with a PDDL flavor, but no attempt is made at a precise syntax.

\subsubsection{The \textit{MachineMalfunction} (MM) Domain}
 In the $\text{MM}(n,m,k)$ domain, there are $n+m$ agents supervising $k$ machines, the agents' tasked to ensure that machines function correctly at all times, where a choice of $n,m$ and $k$ fixes the universe of the domain. The agents have different roles: $n$ of the agents are monitoring the machines for potential malfunctions, while the remaining $m$ agents are system administrators, that from behind a terminal may solve any malfunction by issuing a \textsf{reboot} command to the affected machine. To reboot a machine, admins need to know its serial number, which a monitor may be uncertain about. Hence, the optimal sequence of actions to fix a malfunction problem will depend on how the knowledge of serial numbers is distributed amongst agents. Finally, company policy dictates epistemic preconditions for the actions: monitoring agents are only allowed to report machine as malfunctioning once they know that it is malfunctioning. To avoid deadlock, the requirement for admins is weaker: they may reboot any machine once they know some machine is malfunctioning.
 
 The remainder of this section concerns an external epistemic planning task $\mathsf{MM\_task}\coloneqq(s_0; \mathbf{A}; \varphi_g)$ in $\text{MM}(1,2,2)$, i.e., with one monitor, two admins, and two machines.

\subsubsection{Initial State and Goal}
Figure \ref{example:initial_model_Machine} depicts the the pointed epistemic model $s_0:=(M_0,w_0):=((D,W,R,I),w_0)$, the initial state of $\mathsf{MM\_task}$.
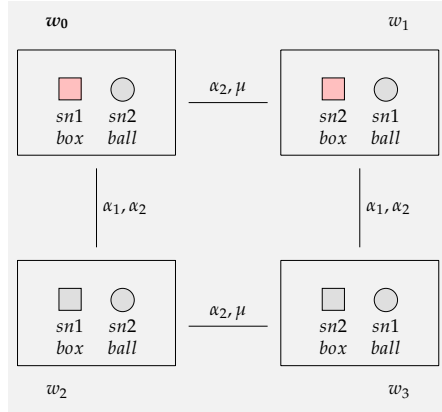
\begin{figure}[H]
  \begin{center}
    \scalebox{0.7}{
            \begin{circuitikz}[framed]
            \draw
            (2.75,2.75) node[below] {$\boldsymbol{w_0}$}
            (9.25,2.75) node[below] {${w_1}$}
            (2.75,-4.25) node[below] {${w_2}$}
            (9.25,-4.25) node[below] {${w_3}$}
            (2,2) -- (2,0)
            (3,1.6) node[below] {\squaredr{\phantom{a}}}
            (3,1) node[below] {$sn1$}
            (3,0.6) node[below] {$box$}
            (2,2) -- (5,2) 
            (2,0) -- (5,0)
            (5,2) -- (5,0)
            (4,1.6) node[below] {\circled{\phantom{a}}}
            (4,1) node[below] {$sn2$}
            (4,0.6) node[below] {$ball$}
            ;
            \draw
            (5.25,1) -- (6.75,1) node[midway,
    above]{$\alpha_2,\mu$}
            (5.25,-3.25) -- (6.75,-3.25) node[midway,
    above]{$\alpha_2,\mu$}
            (8.5,-0.25) -- (8.5,-1.75) node[midway,
    right]{$\alpha_1,\alpha_2$}
            (3.5,-0.25) -- (3.5,-1.75) node[midway,
    right]{$\alpha_1,\alpha_2$}
            ;
            \begin{scope}[xshift=5cm]
            \draw
            (2,2) -- (2,0)
            (3,1.6) node[below] {\squaredr{\phantom{a}}}
            (3,1) node[below] {$sn2$}
            (3,0.6) node[below] {$box$}
            (2,2) -- (5,2) 
            (2,0) -- (5,0)
            (5,2) -- (5,0)
            (4,1.6) node[below] {\circled{\phantom{a}}}
            (4,1) node[below] {$sn1$}
            (4,0.6) node[below] {$ball$}
            ;
            \begin{scope}[yshift=-4cm]
            \draw
            (2,2) -- (2,0)
            (3,1.6) node[below] {\squared{\phantom{a}}}
            (3,1) node[below] {$sn2$}
            (3,0.6) node[below] {$box$}
            (2,2) -- (5,2) 
            (2,0) -- (5,0)
            (5,2) -- (5,0)
            (4,1.6) node[below] {\circled{\phantom{a}}}
            (4,1) node[below] {$sn1$}
            (4,0.6) node[below] {$ball$}
            ;
            \begin{scope}[xshift=-5cm]
            \draw
            (2,2) -- (2,0)
            (3,1.6) node[below] {\squared{\phantom{a}}}
            (3,1) node[below] {$sn1$}
            (3,0.6) node[below] {$box$}
            (2,2) -- (5,2) 
            (2,0) -- (5,0)
            (5,2) -- (5,0)
            (4,1.6) node[below] {\circled{\phantom{a}}}
            (4,1) node[below] {$sn2$}
            (4,0.6) node[below] {$ball$}
            ;
            \end{scope}
            \end{scope}
            \end{scope}
            \end{circuitikz}
    }
    \end{center}
    \caption{\label{example:initial_model_Machine} The pointed epistemic model $s_0=(M_0,w_{0})$, the initial state of $\mathsf{MM\_task}$. The agent domain $D_{\mathtt{agt}}$ comprises three agents, $\alpha_1, \alpha_2$ and $\mu$, for simplicity only depicted as edge labels. Agents $\alpha_1$ and $\alpha_2$ are admins, while $\mu$ is a monitor. There are two machines, depicted as a box and a ball; beneath them are the constants that denote them in that world. In $w_0$, the box machine is denoted in-system by the serial number $sn1$, but colloquially as $box$, etc. In all worlds, the constants $a1, a2$ and $m1$ denote respectively $\alpha_1, \alpha_2$ and $\mu$. The red coloring specifies malfunction. The monitor can observe malfunctions, while the admins cannot. Neither the monitor nor the newly employed admin $\alpha_2$ know the serial numbers, but $\alpha_1$ does. They all know the colloquial descriptions. The monitor knows that both admins face uncertainty about the malfunction, and that $\alpha_1$ knows the serial numbers while $\alpha_2$ does not.
    }
\end{figure}

We aim for a simple example. Hence, the administrator/monitor roles are not formally specified, but could straightforwardly be assigned using predicates.

Before stating the available actions, specify the goal of $\mathsf{MM\_task}$ to be that some agent knows that all machines are not malfunctioning. I.e.,
$$
\varphi_g \coloneqq \exists x K_x \forall y \neg \mathsf{Malfunction}(y)
$$
With this goal achieved, the agent knowing that no malfunction occurs can announce this to the remaining agents to achieve that this becomes known to all, but for simplicity, we have omitted this aspect.

The initial state $s_0$ and the goal $\varphi_g$ may then be described in a `PDDL-like syntax', as it would appear in a problem file, cf. Figure \ref{fig:pddl_problem}.

\begin{figure}[H]
\hrule\vspace{8pt}
\begin{minted}[fontsize=\footnotesize]{text}
(define (problem machine-malfunction-p1)
  (:domain machine-malfunction)
  (:universe
      Alpha1 Alpha2 - admin_agent
      Mu            - monitoring_agent
      Box Ball      - machine)
  (:constants
      sn1 sn2 box ball - machine_id
      m1               - monitoring_agent_id
      a1 a2            - admin_agent_id)
  (:init
    (:actual_world w0
      :constant_map ((sn1 Box) (sn2 Ball) (box Box)  (ball Ball) (m1 Mu) (a1 Alpha1) (a2 Alpha2))
      :atoms     ((malfunction sn1) (malfunction box)))
    (:world w1
      :constant_map ((sn2 Box) (sn1 Ball) (box Box)  (ball Ball) (m1 Mu) (a1 Alpha1) (a2 Alpha2))
      :atoms     ((malfunction sn2) (malfunction box)))
    (:world w2
      :constant_map ((sn1 Box) (sn2 Ball) (box Box) (ball Ball) (m1 Mu) (a1 Alpha1) (a2 Alpha2))
      :atoms     ())
    (:world w3
      :constant_map ((sn2 Box) (sn1 Ball) (box Box) (ball Ball) (m1 Mu) (a1 Alpha1) (a2 Alpha2))
      :atoms     ())
    (:edges
      :Alpha1 ((w0 -- w2) (w1 -- w3))
      :Alpha2 (all)
      :Mu     ((w0 -- w1) (w2 -- w3)))
  (:goal (exists (?a -agent_id) (knows (?a) (forall (?o - object_id) (not (malfunction ?o)))))))
\end{minted}
\hrule
\caption{\label{fig:pddl_problem} A `PDDL-style syntax' description of the initial state $s_0$ and the goal $\varphi_0$. A universe of agents and objects is defined using the keyword $\mathtt{:universe}$, with constants denoting this domain declared the keyword $\mathtt{:constants}$. The $\mathtt{:init}$ keyword precedes the description of the initial state, which comprises worlds and edges between worlds. Each world is declared with a keyword $\mathtt{:world}$ and encompasses a $\mathtt{:constant\_map}$ stating what each constant refers to as a list of pairs (constant, entity), as well as a list of true ground $\mathtt{:atoms}$ (where the closed-world assumption holds). The actual world is defined with the $\mathtt{:actual\_world}$ keyword.  The indistinguishability relation for agents is specified with the $\mathtt{:edges}$ keyword. For each agent, a set of pairs of worlds in the relation is listed, whose reflexive-transitive closure defines the full relation. Finally, the $\mathtt{:goal}$ keyword declares the goal.} 
\end{figure}

\subsubsection{Available Actions as Action Schemas and Domain Definition}
To finalize the specification of the external epistemic planning task $\mathsf{MM\_task}$, the available actions of reporting malfunctions rebooting machines must be defined. To this end, we use the action schemas depicted graphically in Figure \ref{fig:schemas_machine}. 
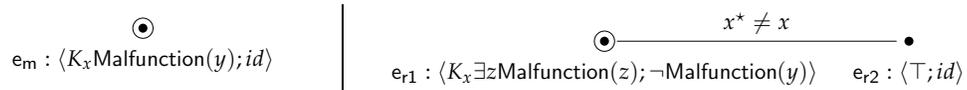
\begin{figure}[H]
\centering
\begin{minipage}{.34\textwidth}
  \centering
\scalebox{0.85}{
            \begin{tikzpicture}
              \tikzstyle{dnode}=[inner sep=1pt,outer sep=1pt,draw,circle,minimum width=9pt]
              \node[dnode, label={[align=right]below: ${\mathsf{e_m}}: \langle K_x \mathsf{Malfunction}(y) ;id\rangle$}] (w0) at (0,0) {};
              \fill (w0) circle [radius=2pt];
            \end{tikzpicture}
        }
\end{minipage}%
\vrule{}
\begin{minipage}{.56\textwidth}
  \centering
        \scalebox{0.85}{
            \begin{tikzpicture}
            \tikzstyle{dnode}=[inner sep=1pt,outer sep=1pt,draw,circle,minimum width=9pt]
            \tikzstyle{nnode}=[inner sep=1pt,outer sep=1pt,circle,minimum width=9pt]
              \node[dnode, label={[align=right]below: ${\mathsf{e_{r1}}}: \langle K_x \exists z \mathsf{Malfunction}(z);\neg\mathsf{Malfunction}(y)\rangle$}] (w0) at (0,0) {};
              \fill (w0) circle [radius=2pt];
              \node[nnode, right=130pt of w0, label={[align=right]below: ${\mathsf{e_{r2}}}: \langle \top ; id\rangle$}] (w1) at (0,0) {};
              \fill (w1) circle [radius=2pt];
              \draw[-] (w0) -- (w1) node[above,midway] {$x^\star \neq x$};
            \end{tikzpicture}
        }
\end{minipage}
\caption{\label{fig:schemas_machine} \textit{Left}: $\text{Malfunction}(x,y)$, the action schema for agent $x$ announcing that they know that $y$ is malfunctioning. \textit{Right}: $\text{Reboot}(x,y)$, the action schema for agent $x$ rebooting machine $y$, an action that $x$ is permitted to do only if $x$ knows that some malfunction is occurring, and which is done privately: other agents than $x$ are uncertain about its execution.}
\end{figure}

Continuing with the `PDDL-like syntax', Figure \ref{fig:pddl_Machine} describes the \textit{MachineMalfunction} domain. The figure is suggestive of a possible approach to standardizing domain definitions, but again, no formal specification of the syntax is attempted.

\begin{figure}[H]
\hrule\vspace{4pt}
\begin{minted}[fontsize=\footnotesize]{text}
;; machine-malfunction domain.
(define (domain machine-malfunction)
  (:types admin_agent_id - agent_id
          monitoring_agent_id - agent_id
          serial_number - machine_id
          machine_id
          agent_id)
  (:predicates (malfunction ?o - machine_id))
  (:action MALFUNCTION
    :agent ?s - monitoring_agent_id
    :parameters (?o - machine_id)
    (:actual_event em
      :precondition  (knows (?s) (malfunction ?o))
      :postcondition        (id))
    (:edge-conditions
      :em -- em (= ?x* ?x*)))
  (:action REBOOT
    :agent ?a - admin_agent_id
    :parameters (?n - serial_number)
    (:actual_event er1
      :precondition  (knows (?a) (exists (?x - object) malfunction (?x)))
      :postcondition        ((malfunction ?n if FALSE))
    (:event er2
      :precondition  (TRUE)
      :posttcondition  (id))
    (:edge-conditions
      :er1 -- er1 (= ?x* ?x*)
      :er2 -- er2 (= ?x* ?x*)
      :er1 -- er2 (not (= ?x* ?a)))
\end{minted}
\vspace{-4pt}\hrule
\caption{\label{fig:pddl_Machine} A domain definition for the MM domain in a `PDDL-style syntax'. Each action schema includes an $\mathtt{:agent}$ executing it as well as the schema's $\mathtt{:parameters}$ list. The possible events comprised in the action and their corresponding edge-conditions are listed next. The actual event is defined with the $\mathtt{:actual\_event}$ keyword, under which pre- and postconditions are listed. Postconditions are given as a list of statements of the form ``$\mathtt{ground\_atom \ if \ condition}$''. A similar $\mathtt{:event}$ keyword is used for non-actual events.  The keyword $\mathtt{edge\text{-}conditions}$ lists, for each pair of events, its edge-condition, via a line of the form $\mathtt{event1 \text{--} event1 (edge\text{-}condition)}$. Both the $\mathtt{MALFUNCTION}$ and the $\mathtt{REBOOT}$ schemas have epistemic preconditions, which can be schematized thanks to the variable-indexed epistemic operators $K_x$.}
\end{figure}

\subsubsection{Plan and Execution}
Given the initial state and actions of the external epistemic planning task $\mathsf{MM\_task}$, the goal $\varphi_g$ may be achieved by the monitor agent reporting that the machine colloquially called $box$ is malfunctioning, after which the administrator $\alpha_1$ may reboot that machine, by knowing its serial number $sn1$. I.e., the plan $\pi$ below achieves the goal, cf. Figure \ref{example:machine_updates}:
$$
\pi \coloneqq \text{Malfunction}(m1,box), \text{Reboot}(a1,sn1)
$$

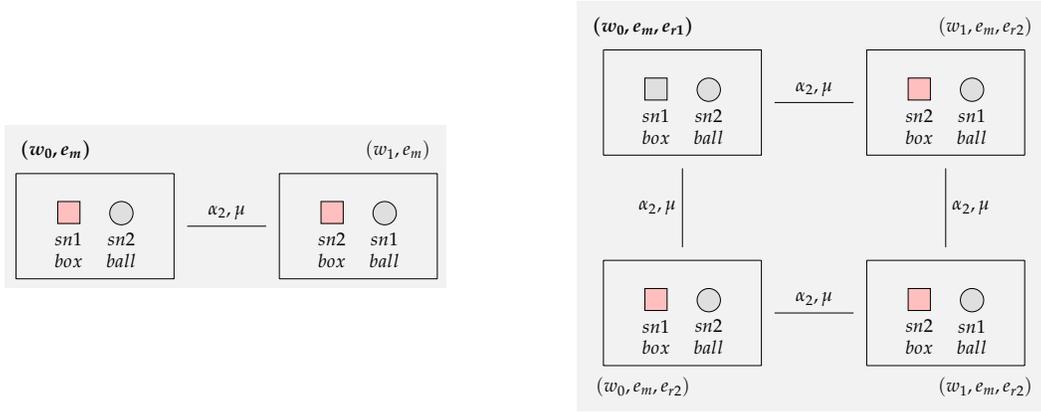
\begin{figure}[H]
\begin{minipage}{.5\textwidth}
  \centering
    \scalebox{0.7}{
            \begin{circuitikz}[framed]
            \draw
            (2.75,2.75) node[below] {$\boldsymbol{(w_0,e_m)}$}
            (9.25,2.75) node[below] {$(w_1,e_m)$}
            (2,2) -- (2,0)
            (3,1.6) node[below] {\squaredr{\phantom{a}}}
            (3,1) node[below] {$sn1$}
            (3,0.6) node[below] {$box$}
            (2,2) -- (5,2) 
            (2,0) -- (5,0)
            (5,2) -- (5,0)
            (4,1.6) node[below] {\circled{\phantom{a}}}
            (4,1) node[below] {$sn2$}
            (4,0.6) node[below] {$ball$}
            ;
            \draw
            (5.25,1) -- (6.75,1) node[midway, above]{$\alpha_2,\mu$}
            ;
            \begin{scope}[xshift=5cm]
            \draw
            (2,2) -- (2,0)
            (3,1.6) node[below] {\squaredr{\phantom{a}}}
            (3,1) node[below] {$sn2$}
            (3,0.6) node[below] {$box$}
            (2,2) -- (5,2) 
            (2,0) -- (5,0)
            (5,2) -- (5,0)
            (4,1.6) node[below] {\circled{\phantom{a}}}
            (4,1) node[below] {$sn1$}
            (4,0.6) node[below] {$ball$}
            ;
            \end{scope}
            \end{circuitikz}
    }
\end{minipage}%
\begin{minipage}{.5\textwidth}
  \centering
  \scalebox{0.7}{
            \begin{circuitikz}[framed]
            \draw
            (2.75,2.75) node[below] {$\boldsymbol{(w_0,e_m,e_{r1})}$}
            (9.25,2.75) node[below] {${(w_1,e_m,e_{r2})}$}
            (2,2) -- (2,0)
            (3,1.6) node[below] {\squared{\phantom{a}}}
            (3,1) node[below] {$sn1$}
            (3,0.6) node[below] {$box$}
            (2,2) -- (5,2) 
            (2,0) -- (5,0)
            (5,2) -- (5,0)
            (4,1.6) node[below] {\circled{\phantom{a}}}
            (4,1) node[below] {$sn2$}
            (4,0.6) node[below] {$ball$}
            ;
            \draw
            (5.25,1) -- (6.75,1) node[midway, above]{$\alpha_2,\mu$}
            ;
            \begin{scope}[xshift=5cm]
            \draw
            (2,2) -- (2,0)
            (3,1.6) node[below] {\squaredr{\phantom{a}}}
            (3,1) node[below] {$sn2$}
            (3,0.6) node[below] {$box$}
            (2,2) -- (5,2) 
            (2,0) -- (5,0)
            (5,2) -- (5,0)
            (4,1.6) node[below] {\circled{\phantom{a}}}
            (4,1) node[below] {$sn1$}
            (4,0.6) node[below] {$ball$}
            ;
            \end{scope}
            \begin{scope}[yshift=-4cm]
            \draw
            (2,2) -- (2,0)
            (3,1.6) node[below] {\squaredr{\phantom{a}}}
            (3,1) node[below] {$sn1$}
            (3,0.6) node[below] {$box$}
            (2,2) -- (5,2) 
            (2,0) -- (5,0)
            (5,2) -- (5,0)
            (4,1.6) node[below] {\circled{\phantom{a}}}
            (4,1) node[below] {$sn2$}
            (4,0.6) node[below] {$ball$}
            ;
            \draw
            (5.25,1) -- (6.75,1) node[midway, above]{$\alpha_2,\mu$}
            ;
            \end{scope}
            \begin{scope}[xshift=5cm, yshift=-4cm]
            \draw
            (2,2) -- (2,0)
            (3,1.6) node[below] {\squaredr{\phantom{a}}}
            (3,1) node[below] {$sn2$}
            (3,0.6) node[below] {$box$}
            (2,2) -- (5,2) 
            (2,0) -- (5,0)
            (5,2) -- (5,0)
            (4,1.6) node[below] {\circled{\phantom{a}}}
            (4,1) node[below] {$sn1$}
            (4,0.6) node[below] {$ball$}
            ;
            \end{scope}
            \draw
            (3.5,-0.25) -- (3.5,-1.75) node[midway, left]{$\alpha_2,\mu$}
            (8.5,-0.25) -- (8.5,-1.75) node[midway, right]{$\alpha_2,\mu$}
            (2.75,-4.1) node[below] {$(w_0,e_m,e_{r2})$}
            (9.25,-4.1) node[below] {${(w_1,e_m,e_{r2})}$}
            ;
            \end{circuitikz}
    }
\end{minipage}
    \caption{
    Executing the plan $\pi$ in $s_0$. \textit{Left}: The epistemic state $s_1 \coloneqq s_0\otimes \text{Malfunction}(m1,box)$ reached after $\mu$ reports that $box$ is malfunctioning. All agents learn that a machine is malfunctioning. However, $\alpha_2$ and $\mu$ are still unsure about the serial number of the malfunctioning robot, but $\alpha_1$ is not. \textit{Right}: The epistemic state $s_2 \coloneqq s_1\otimes \text{Reboot}(a1,box)$ reached after admin agent $\alpha_1$ reboots $sn1$. Since $\alpha_1$ has rebooted $sn1$ privately, $\alpha_2$ and $\mu$ still do not know that all machines are functioning. However, the goal is achieved: there is \textit{some} agent ($\alpha_1$) that knows that all machines are functioning.}
\label{example:machine_updates}
\end{figure}

\subsubsection{Knowing Who and Alternative Plans}
There are other plans than $\pi$ that would also solve the external epistemic planning task $\mathsf{MM\_task}$, but $\pi$ achieves the goal in the fewest number of steps.

The plan $\pi$ is kept short by making admin $\alpha_1$ act and finally witness the existence criteria in the goal $\varphi_g$. Admin $\alpha_1$ is well-suited to this purpose because $\alpha_1$ does not face uncertainty about the two names $box$ and $sn1$. For the both names, $\alpha_1$ \textit{knows who} the names refer to, which in turn entails that in $s_1$, $\alpha_1$ knows what machine to reboot.

Any plan where admin $\alpha_2$, instead of $\alpha_1$, acts, will necessarily be longer, because $\alpha_2$ does not know is uncertain about what serial number belong to which machine, and must therefore reboot \textit{both} machines. More specifically, then given $s_0$, any successful plan must start with the announcement Malfunction$(m1,box)$ to the effect that the admins know of a malfunction, required for them to execute reboots. From $s_0$, this results in $s_1$ of Figure \ref{example:machine_updates}. From there, both Reboot$(a2,sn1)$ and Reboot$(a2,sn2)$ must be performed before $\alpha_2$ knows that no machine is malfunctioning. Hence, the two shortest successful plans in which only $\alpha_2$ reboots are
\begin{align*}
\pi' \coloneqq \text{Malfunction}(m1,box), \text{Reboot}(a2,sn1), \text{Reboot}(a2,sn2)\\
\pi'' \coloneqq \text{Malfunction}(m1,box), \text{Reboot}(a2,sn2), \text{Reboot}(a2,sn1)    
\end{align*}

Hence, the potential uncertainty introduced by the non-rigidity of constants may be consequential for epistemic planning. In this simple example, the non-rigid serial numbers refers to objects, but as also agent terms may be non-rigid, the presented framework allows modeling situations in which e.g. an agent is uncertain about who a message must be delivered to, or planning in situations involving `code names', known only to a strict subset of agents.

\section{\label{sec:act.lang}Languages for Actions}

We define a language for reasoning about actions, denoted $\mathcal{L}_{AM}$. This language extends the basic language $\mathcal{L}$ with \textit{action modalities} with the form $[A,e]$, where $A$ is an action model and $e$ is an event from $A$. The language  $\mathcal{L}_{AM}$ has formulas of the form $[A,e]\varphi$, which are interpreted as: `after event $e$ of action $A$ occurs, $\varphi$ is true'. This language extension allows us to include formulas mentioning other actions in the pre- and postconditions of some actions, as well as in goal formulas. It is thus possible to define, e.g., a goal such as: ``Achieve a state in which it is impossible to perform an action that will result in $\varphi$''. With finitely many actions described by the models $\mathcal{A}=\{A_1,\dots,A_n\}$, such a formula would be $\bigwedge_{A\in\mathcal{A},e\in E^A} [A,e]\neg \varphi$. 

The grammar of $\mathcal{L}_{AM}$ is defined by double recursion, adapting a construction well known in the DEL literature (see, e.g., appendix H in \cite{sep-dynamic-epistemic} or \cite{ditmarsch2007}).

\begin{defn}
\label{Def. dynamic language}
Let $\mathcal{L}_0 = \mathcal{L}$, and let $AM_0$ be the set of pointed action models whose precondition formulas are all from $\mathcal{L}_0$. Define $\mathcal{L}_{k+1}$ and $AM_{k+1}$ as follows:
\begin{gather*}\tag{\(\mathcal{L}_{k+1}\)}
\varphi \Coloneqq r(t_{1},...,t_{\ell})\mid\neg\varphi\mid\varphi\wedge\varphi\mid K_{t}\varphi\mid\forall x\varphi \mid [A,e]\varphi\\
\end{gather*}
where $(A,e)\in AM_{k}$, and let $\mathtt{AM}_{k+1}$ be the set of pointed action models whose precondition formulas are all from $\mathcal{L}_{k+1}$. Lastly, define the \textbf{\textit{language}} $\mathcal{L}_{AM}$ and the set of \textbf{\textit{action models}} $AM$ as \[\mathcal{L}_{AM} \coloneqq \bigcup_{k\in\mathbb{N}}\mathcal{L}_{k}, \; AM \coloneqq \bigcup_{k\in\mathbb{N}}(AM_k)\] 
\end{defn}

As with the formulas from the static language $\mathcal{L}$, the formulas from $\mathcal{L}_{AM}$ are evaluated over epistemic models. 

\begin{defn} The satisfaction relation between epistemic models, assignments and formulas of $\mathcal{L}_{AM}$ is the smallest extension of $\vDash$ that satisfies:
 \[M,w \vDash_v [A,e]\varphi \text{ iff } M,w\vDash_v \mathsf{pre}(e) \text{ implies } M\otimes A, (w,e)\vDash_v \varphi \] 
\end{defn}

This extended satisfaction relation makes it possible to model-check conditions concerning actions. Given a pointed model $(M,w)$, we may want to know whether a formula $\varphi$ would hold after a sequence of pointed action models $(A_1,e_1),\dots, (A_n,e_n)$ has been executed.
This can of course be done by computing a sequence of product updates and checking whether $M\otimes A_1\otimes \dots \otimes A_n,(w,e_1,\dots,e_n)\vDash \varphi$. But, equivalently, we can check whether the corresponding formula holds at $(M,w)$, i.e., whether $M,w\vDash [A_1,e_1]\dots [A_n,e_n]\varphi$. If $\varphi$ is a goal formula and $(A_1,e_1),\dots,(A_n,e_n)$ is a plan, then model-checking such a formula corresponds to \textit{plan verification}. Section \ref{subsec:dynamic_metatheory} gives so-called reduction axioms for $\mathcal{L}_{AM}$ formulas, showing that any formula containing an action modality can be expressed as a formula in the basic epistemic language $\mathcal{L}$. Consequently, plan verification could be treated as a problem of model-checking formulas of $\mathcal{L}$ in an initial state $s_0 = (M,w)$. 

\section{Axiomatic Systems and Metatheory}\label{sec:metatheory}
This section presents axiom systems for both static and dynamic term-modal logic. Metatheoretical results include soundness and completeness, frame characterizations, and decidability results. All proofs may be found in Appendix \ref{A.proof.appendix}.
\subsection{\label{subsec:Normal-Term-Modal-Logic}Normal Term-Modal Logic}

\subsubsection{Axiom System}

\begin{table}
\begin{tabular}{lllll}
First-order principles &  &  & Modal and interaction principles & \tabularnewline
\cline{1-2} \cline{4-5} 
\noalign{\vskip3pt}
all propositional tautologies &  &  & $K_{t}(\varphi\rightarrow\psi)\rightarrow(K_{t}\varphi\rightarrow K_{t}\psi)$ & K\tabularnewline
\noalign{\vskip3pt}
$\forall x\varphi\rightarrow\varphi\left(y/x\right)$, for $y$ free
in $\varphi$ & UE &  & $\forall xK_{t}\varphi\rightarrow K_{t}\forall x\varphi$, for $x$
not occurring in $t$ & BF\tabularnewline
\noalign{\vskip3pt}
$t=t$, for $t\in\mathtt{T}$ & Id &  & $(x\neq y)\rightarrow K_{t}(x\neq y)$ & KNI\tabularnewline
\noalign{\vskip3pt}
$(x=y)\rightarrow\big(\varphi(x)\rightarrow\varphi(y)\big)$ & PS &  &  & \tabularnewline
\noalign{\vskip3pt}
$(c=c)\rightarrow\exists x(x=c)$ & $\exists$Id &  & Inference rules & \tabularnewline
\cline{4-5} 
\noalign{\vskip3pt}
$x\neq y$, if $\mathtt{t}(x)\neq\mathtt{t}(y)$ & DD &  & From $\varphi,\varphi\rightarrow\psi$, infer $\psi$ & MP\tabularnewline
\noalign{\vskip3pt}
 &  &  & From $\varphi$, infer $K_{t}\varphi$ & KG\tabularnewline
\noalign{\vskip3pt}
 &  &  & From $\varphi\rightarrow\psi$, infer $\varphi\rightarrow\forall x\psi$,
for $x$ not free in $\varphi$ & UG\tabularnewline
\noalign{\vskip3pt}
\end{tabular}

\caption{\label{tab:axioms}Axiom schemata for the minimal normal term-modal
logic $\mathsf{K}$.}
\end{table}
Table \ref{tab:axioms} contains the axioms and inference rules for the term-modal logic
$\mathsf{K}$. Some are common first-order axioms, like \textit{Universal
Elimination }(UE), Reflexivity of Identity (Id), and the \textit{Principle
of Substitution} (PS). In a modal logical context, PS also has a modal
feature: it is restricted to variables to allow for non-rigid constants.
If PS is assumed also for constants, $(a=b)\rightarrow(K_{t}\varphi(a)\rightarrow K_{t}\varphi(b))$
becomes a theorem, valid only for rigid constants. \textit{Existence
of Identicals} ($\exists$Id) is included to ensure that all constants
obtains an extension in the canonical models of Section \ref{subsec:Completeness};
\textit{Divided Domain} (DD) is included to enforce type-distinction
between variables logically rather than syntactically. The modal and
interaction principles \textit{Distribution} (K) and \textit{Knowledge
of Non-Identity }(KNI) are formulated as standard while the \textit{Barcan
Formula }(BF) has a restriction in the term-modal case; the Barcan Formula
ensures constant domains: its validity implies non-growing domains,
illustrated in the proof of soundness (Section \ref{A.static.soundness}),
and its converse implies non-shrinking domains (and is provable in
$\mathsf{K}$, cf. e.g. \cite[p. 245]{Hughes_&_Cresswell_New.Intro}).
Knowledge of Non-Identity reflects the rigidity of variables. The
inference rules \textit{Modus Ponens }(MP), \textit{Knowledge Generalization
}(KG) and \textit{Universal Generalization }(UG) contain no surprises.

Notice that nothing in the language or axioms of $\mathsf{K}$ specify
the number of agents in the system. The number of agents emerges as
a definable frame characteristic, cf. Section \ref{subsec:frame.char}.

\subsubsection{Normality}

In Section \ref{subsec:static_metatheory}, we formally state that
$\mathsf{K}$ is complete with respect to the class of all frames.
The axioms and inference rules sufficient for a complete system are
close to standard axiomatizations of first-order modal logic, cf.
e.g. \cite{Brauner&Ghilardi-FOML,Fagin_et_al._1995,Hughes_&_Cresswell_New.Intro}.
We take the close-to-standard format of the $\mathsf{K}$ axioms to
indicate the innocence of the term-modal extensions of the syntax
and semantics. This is further corroborated by the main result of
this section, the \textit{Canonical Class Theorem} \vpageref{thm:Canonical}.
In essence, the theorem shows that any closed extension of $\mathsf{K}$
is complete with respect to the class of its canonical models. The
result thus justifies the following definition:
\begin{defn}
A set of formulas $\Lambda\subseteq\mathcal{L}$ is called a \textbf{\textit{normal
term-modal logic}} if, and only if, $\Lambda$ contains all axioms
of Table \ref{tab:axioms} and is closed under the Table \ref{tab:axioms}'s
inference rules. The smallest normal term-modal logic is denoted $\mathsf{K}$.
\end{defn}

\subsubsection{\label{subsec:static_metatheory}Canonical Class Theorem and Completeness}

In ordinary modal logic, each normal modal logic gives rise to a unique
\textit{canonical model}. In a similar manner, each normal term-modal
logic $\Lambda$ gives rise to a \textit{class }of canonical models,
one for each $\Lambda$-maximal consistent set. Section \ref{A.proof.appendix}
contains the details of the construction, as well as the proof of
the following main theorem:
\begin{restatable}[Canonical Class Theorem]{thm}{Canonical}\label{thm:Canonical}
Any normal term-modal logic $\Lambda$ is strongly complete with respect to its canonical class.
\end{restatable}
Mirroring the role of the Canonical Model Theorem of ordinary modal logic (see e.g. \cite{BlueModalLogic}), we obtain the following corollary to Theorem \ref{thm:Canonical}:
\begin{restatable}[Completeness]{cor}{Completeness}\label{Cor. Complete ALL-1}
The logic $\mathsf{K}$ is strongly complete with respect to the class of all frames $\boldsymbol{F}$.
\end{restatable}
$\mathsf{K}$ is also sound with respect to the class of all frames.
Section \ref{A.static.soundness} contains the formal statement and
a proof sketch, with details given for the axiom K and the Barcan
Formula.

\subsubsection{\label{subsec:frame.char}Characterizing Frame Properties}

The completeness result of Corollary \ref{Cor. Complete ALL-1} may be extended to more specific frame classes. Table \ref{tab:char} contains an overview of axiom schemata and the frame conditions they characterize. For illustration, proofs for 4 and N are given in Section \ref{A.frame.char.}. From the Canonical Class Theorem and Table \ref{tab:char}, completeness results for standard logics like $\mathsf{KD45}$, $\mathsf{S4}$ and $\mathsf{S5}$ follow as corollaries.
\begin{table}[H]
\begin{centering}
\begin{tabular}{lllll}
Axiom &  &  &  & Frame condition\tabularnewline
\cline{1-3} \cline{5-5} 
$\forall x\big(K_{x}\varphi\rightarrow\varphi\big)$ &  & T &  & Reflexive\tabularnewline
$\forall x(\neg K_{x}\bot)$ &  & D &  & Serial\tabularnewline
$\forall x\big(K_{x}\varphi\rightarrow K_{x}K_{x}\varphi\big)$ &  & 4 &  & Transitive\tabularnewline
$\forall x\big(\neg K_{x}\varphi\rightarrow K_{x}\neg K_{x}\varphi\big)$ &  & 5 &  & Euclidean\tabularnewline
$\exists x_{1},...,x_{n}\left(\left(\bigwedge_{i\leq n}K_{x_{i}}\top\right)\wedge\left(\bigwedge_{i,j\leq n,i\neq j}x_{i}\neq x_{j}\right)\wedge\forall y\left(K_{y}\top\rightarrow\bigvee_{i\leq n}y=x_{i}\right)\right)$ &  & N &  & $|D_{\mathtt{agt}}|=n$\tabularnewline
$\exists x_{1},...,x_{m}\left(\left(\bigwedge_{i,j\leq m,i\neq j}x_{i}\neq x_{j}\right)\wedge\forall y\left(\bigvee_{i\leq m}y=x_{i}\right)\right)$ &  & M &  & $|D|=m$\tabularnewline
\end{tabular}
\par\end{centering}
\centering{}\caption{\label{tab:char}Term-modal axiom schemata and the frame conditions they characterize.}
\end{table}
The principles N and M are special to our term-modal treatment. N and M define domain sizes. Nothing in the language or axioms of $\mathsf{K}$ specify the number of agents in the system: as in first-order logic, the domain size is by default left unspecified. In ordinary epistemic logic, it is common to assume a fixed, finite index set of agents. The domain size principles N and M similarly fixes domain sizes: N fixes the agent domain to size $n$. It uses the $K_x \top$-expressions to ensure that the bound variables are all of type $\mathtt{agt}$. With the quantifications thus ranging only over agents, N specifies specifically the size of $D_{\mathtt{agt}}$. This may be compared to M, which does not put constraints on the type of the bound variables, thereby fixing only the size of the joint agent--object domain $D$. For details concerning N, see Proposition \ref{prop.N} on \pageref{prop.N}.

The principles T, D, 4 and 5 deviate from their ordinary forms by being quantified. In standard modal logic, the formula
\begin{equation}
K_{c}\varphi\rightarrow K_{c}K_{c}\varphi\label{eq:4_classic}
\end{equation}
characterizes the class of transitive frames. This is not true here, as the constant $c$ may be  \textit{non-rigid}:\footnote{The non-rigidity of constants is reflected in $\mathsf{K}$: \textit{Knowledge of Identity }is provable for variables, but not for constants. I.e., $\mathsf{K}$ proves $(x=y)\rightarrow K_{t}(x=y)$, but not $(a=b)\rightarrow K_{t}(a=b)$.} see Figure \ref{fig:4invalid} for a transitive model invalidating (\ref{eq:4_classic}). 
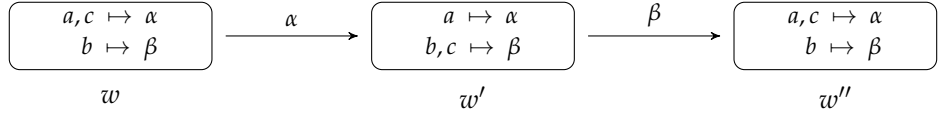
\begin{figure}

\begin{centering}
\scalebox{1.0}{
\begin{tikzpicture}[->,>=stealth',auto,shorten >=5pt,shorten <=5pt]

\node[world] (w) {{\small $a,c \mapsto \alpha$\\ \small$\phantom{a,}b \mapsto \beta$}};
\node[world,right=60pt of w] (ww) {{\small $\phantom{c,}a \mapsto \alpha$\\ \small$b,c \mapsto \beta$}};
\node[world,right=60pt of ww] (www) {{\small $a,c \mapsto \alpha$\\ \small$\phantom{c,}b \mapsto \beta$}};

\node[empty,below=0pt of w] (.) {$w$};
\node[empty,below=0pt of ww] (.) {$w'$};
\node[empty,below=0pt of www] (.) {$w''$};

\draw[->] (w) -- (ww) node[above,midway,xshift=0pt,yshift=0pt] {{\small $\alpha$}};
\draw[->] (ww) -- (www) node[above,midway,yshift=0pt,xshift=0pt] {{\small $\beta$}};
\end{tikzpicture}
}
\par\end{centering}
\caption{\label{fig:4invalid} A transitive model invalidating $K_{c}(b=c)\rightarrow K_{c}K_{c}(b=c)$ at $w$. With $D_{\mathtt{agt}}=\{\alpha,\beta\}$, all relations are transitive. The notation $a\protect\mapsto\alpha$ specifies that $\alpha$ is the extension of the constant $a$ in the given world. In $w$, it holds that $K_c (b = c)$ as $c\protect\mapsto\alpha$ in $w$ and $b,c\protect\mapsto\beta$ in $w'$ (in $w$, the knowledge that $(b = c)$ is held by agent $\alpha$, as $c\protect\mapsto\alpha$ in $w$). World $w'$ satisfies $K_c (b \neq c)$ as $c\protect\mapsto\beta$ in $w'$ and $b\protect\mapsto\beta, c\protect\mapsto\alpha$ in $w''$ (in $w'$, the knowledge that $(b \neq c)$ is held by $\beta$, as $c\protect\mapsto\beta$ in $w'$). As a consequence, $w'$ also satisfies $\neg K_c (b = c)$. Hence, $w$ does not satisfy $K_c K_c (b = c)$---that the agent named $c$ (i.e., $\alpha$) knows that the agent named $c$ (i.e., $\beta$) knows that $(b = c)$.}

\end{figure}
The invalidity arises as the extension of the $c$ is not fixed under scope of operators: in the consequent, the accessibility relation which the inner occurrence of $K_{c}$ quantifies over need not be the same as the accessibility relation of the outer. This makes the appeal to transitivity void.\footnote{In his 1962 \cite{TML_Hintikka1962}, Hintikka argues that $K_{c}\varphi\rightarrow K_{c}K_{c}\varphi$ intuitively is valid only if $c$ knows that she is $c$; i.e., that she knows who $c$ is. Hintikka argues that this is captured by $\exists xK_{c}(x=c)$, which makes $c$ locally rigid for the agent: $I(c,w')=I(c,w)$ for all $w'$ in $R_{I(c,w)}(w)$. Indeed, $\exists xK_{c}(x=c)\wedge K_{c}\varphi\rightarrow K_{c}K_{c}\varphi$ is valid on transitive frames.} The formulation in Table \ref{tab:char} avoids the non-rigidity problem, but does impose the criteria for all agents uniformly.

\subsubsection{Heterogeneous Agents}

Though treating all agents uniformly is common in epistemic logic, one may desire heterogeneous agents. With the given setup, we do not believe this can be done at the level of frames. On the level of models, one option to this end is to attribute epistemic criteria to subgroups using predicates; a second is to introduce individual names. In either case, one may desire the defining criterion to be rigid. However, full rigidity is not definable in general as models may be disconnected. \textit{Local rigidity}\textemdash invariance of interpretation over connected components\textemdash is definable by formulas of the forms 
\begin{align}
\forall x(r(x)\leftrightarrow\forall yK_{y}r(x))\label{eq:pseudo.rigid.P}\\
\exists x((x=a)\wedge\forall yK_{y}(x=a))\label{eq:pseudo.rigid.a}
\end{align}
The validity of (\ref{eq:pseudo.rigid.P}) and (\ref{eq:pseudo.rigid.a}) characterize features of interpretations: (\ref{eq:pseudo.rigid.P}) ((\ref{eq:pseudo.rigid.a}), resp.) is valid in a model $M=(D,W,R,I)$
iff for all $w,w'\in W$, $(w,w')\in R_{\alpha}$ for some $\alpha\in D_{\mathtt{agt}}$
implies $I(r,w)=I(r,w')$ ($I(a,w)=I(a,w')$, resp.). In conjunction
with formulas of the forms 
\begin{align}
\forall x(r(x)\rightarrow(K_{x}\varphi\rightarrow K_{x}K_{x}\varphi))\label{eq:P->4}\\
K_{a}\varphi\rightarrow K_{a}K_{a}\varphi
\end{align}
one may obtain some individuated control over relation properties.

\subsubsection{Decidability}

Let $\mathsf{K}_{n}$ and $\mathsf{K}_{n/m}$ be the smallest normal extensions of $\mathsf{K}$ with, respectively, the domain size axiom N, and both domain size axioms N and M under the proviso that $m>n$. $\mathsf{K}_{n}$ and $\mathsf{K}_{n/m}$ are then sound and complete with respect to, respectively, the class of all frames with exactly $n$ agents, and the class of all frames with exactly $n$ agents and exactly $m-n$ objects. These finite domain properties are used in the proof of items 1. and 2. of the below proposition, shown in Section \ref{A.decidability}. Decidability results from the literature are discussed in Section \ref{subsec:TML.litt.}.
\begin{prop}
Let $\mathsf{K}_{n/m}$, $\mathsf{K}_{n}$ and $\mathsf{K}$ be given
in $\mathcal{L}$, based on the signature $\Sigma$. Let $\mathcal{L}_{\mathtt{agt}}\subseteq\mathcal{L}$
contain all formulas containing only agent-terms, $t\in\mathtt{t}^{-1}(\mathtt{agt})$.
\begin{enumerate}
\item For all $\varphi\in\mathcal{L}$, it is decidable whether $\vdash_{\mathsf{K}_{n/m}}\varphi$ or not.
\item a) For all $\varphi\in\mathcal{L}_{\mathtt{agt}}$, it is decidable whether $\vdash_{\mathsf{K}_{n}}\varphi$ or not. b) In general, it is undecidable.
\item In general, it is undecidable whether $\vdash_{\mathsf{K}}\varphi$
or not.
\end{enumerate}
\end{prop}

\subsection{Dynamic Term-Modal Logic\label{sec:dyn.TML}}
\subsubsection{Axiom System}\label{subsec:red.ax}
Table \ref{tab:axioms_dyn} contains the axioms and inference rules for the dynamic term-modal logic $\mathsf{AM}$. In Section \ref{subsec:dynamic_metatheory}, we formally state that $\mathsf{K} + \mathsf{AM}$ is sound and complete with respect to the class of all frames. This completeness result may be extended to more specific frame classes, as was the case with $\mathsf{K}$ (see Section \ref{subsec:frame.char}). The completeness proof for $\mathsf{K} + \mathsf{AM}$ is by translation, a well-known approach in DEL \cite{BaltagBMS_1998,sep-dynamic-epistemic,ditmarsch2007,plaza1989}. The axioms in $\mathsf{AM}$ are so-called \textit{reduction axioms}, which enable the translation of formulas with action modalities into provably equivalent ones without any action modalities. Then completeness follows from the known completeness of the static logic $\mathsf{K}$. For a detailed description of the reduction strategy to completeness, see e.g. \cite{ditmarsch2007}.

\begin{table}[H]
\centering
\begin{tabular}{llllll}
\noalign{\vskip3pt}
\multicolumn{4}{l}{Reduction axioms} & 
 & \tabularnewline
\hline 
\noalign{\vskip3pt}
\multicolumn{4}{l}{$[A,e] r(t_{1},...,t_{n}) \leftrightarrow (\mathsf{pre}(e)\to \mathsf{post}(e)(r(t_{1},...,t_{n})))$} & 
Action and atom & \tabularnewline
\multicolumn{4}{l}{$[A,e]\neg \varphi \leftrightarrow (\mathsf{pre}(e) \to \neg [A,e]\varphi)$} & 
Action and negation & \tabularnewline
\multicolumn{4}{l}{$[A,e](\varphi \wedge \psi) \leftrightarrow (([A,e]\varphi) \wedge ([A,e]\psi))$} & 
Action and conjunction & \tabularnewline
\multicolumn{4}{l}{$[A,e]K_{t}\varphi \leftrightarrow \bigwedge_{e'\in E} (Q(e,e')[{x^\star}\mapsto t] \to K_t [A,e']\varphi)$} & 
Action and knowledge & \tabularnewline
\multicolumn{4}{l}{$[A,e]\forall x \varphi \leftrightarrow (\mathsf{pre}(e) \to \forall x [A,e]\varphi)$} & 
Action and quantification (Dynamic Barcan) & \tabularnewline
\multicolumn{4}{l}{$[A,e][A',e']\varphi \leftrightarrow [(A,e)\circ (A',e')]\varphi$} & 
Action composition & \tabularnewline
\noalign{\vskip3pt}
\multicolumn{4}{l}{Inference rules} & 
 & \tabularnewline
\hline 
\noalign{\vskip3pt}
\multicolumn{4}{l}{$\text{From } \varphi, \text{infer } [A,e]\varphi$}  & 
Action necessitation
\end{tabular}
\caption{\label{tab:axioms_dyn} Axiom and rule schemata for the system $\mathsf{AM}$.
} \end{table}

The reduction axioms in $\mathsf{AM}$ are similar to those used in logics for epistemic actions, introduced by \cite{BaltagBMS_1998}. Naturally, as dynamic term-modal logic is first-order, there are reduction axioms for formulas involving quantifiers. Moreover, the axiom for formulas with the knowledge operator is non-standard. Unlike standard action models, the ones presented here are edge-conditioned and use variable substitutions, which require some modifications. A more detailed comparison of these axioms and standard ones is provided in Section \ref{subsec:DEL.litt.}. The \textit{Action composition} axiom appeals to action models of the form $(A,e)\circ (A',e')$. This notation refers to the \textit{composition} of $(A,e)$ and $(A',e')$, defined following \cite{ditmarsch_kooi_ontic}, but adapted to accommodate edge-conditions and first-order atoms:

\begin{defn} \label{Def. Action composition} Let $A_1=(E_1,Q_1,\mathsf{pre}_1,\mathsf{post}_1)$ and $A=(E_2,Q_2,\mathsf{pre}_2,\mathsf{post}_2)$  be given. The \textbf{\textit{composition}} of $A_1$ and $A_2$ is the action model $A_1\circ A_2 = (E,Q,\mathsf{pre},\mathsf{post})$ where 
\begin{enumerate}
\item $E= E_1\times E_2$
\item $Q((e_1,f_1),(e_2,f_2)) =  Q_1(e_1,f_1) \wedge [A_1,e_1]Q_2(e_2,f_2)$.
\item $\mathsf{pre}(e_1,e_2) = \pre(e_1) \wedge [A_1, e_1]\pre(e_2)$,
\item $\dom(\mathsf{post}((e_1,e_2)))=\dom(\mathsf{post}_1(e_1))\cup \dom(\mathsf{post}_2(e_2))$ and if $r(t_1,\dots, t_k)\in \dom(\mathsf{post}((e_1,e_2)))$, then
\[\mathsf{post}((e_1,e_2)(r(t_1,\dots, t_k)) = \begin{cases} 
\mathsf{post}_1(e_1)(r(t_1,\dots, t_k)) & \text{if } r(t_1,\dots, t_k)\not\in \dom(\mathsf{post}_2(e_2))\\
[A,e]\mathsf{post}_2(e_2)(r(t_1,\dots, t_k)) & \text{otherwise}
\end{cases}
\]
\end{enumerate}

\end{defn}

\subsubsection{\label{subsec:dynamic_metatheory} Soundness and Completeness via Reduction Axioms} 

As anticipated in Section \ref{subsec:red.ax},  $\mathsf{K+AM}$ is sound and complete with respect to the class of all models. 

\begin{prop}[Soundness of $\mathsf{K+AM}$] \label{prop:sound_dyn} $\mathsf{K+AM}$ is sound with respect to the class of all models. 
\end{prop}

The soundness of $\mathsf{K+AM}$ (Proposition  \ref{prop:sound_dyn}) is established in the standard way, by showing the semantic validity of the reduction axioms and inference rules. The proof is straightforward and therefore omitted.

Completeness follows as a corollary from a number of lemmas that are presented in Section \ref{subsec:comp_dyn}. 

\begin{cor}[Completeness of $\mathsf{K+AM}$] \label{cor:comp_dyn} $\mathsf{K+AM}$ is complete with respect to the class of all models. Moreover, any extension of $\mathsf{K+AM}$ obtained by adding axioms characterizing frame conditions is complete with respect to the corresponding class of models. 
\end{cor}

\section{Related Work}\label{sec:related.work}

\subsection{Epistemic Planning}\label{subsec:rel.work.ep.plan}
Several articles on multi-agent epistemic planning have appeared recently. The existing work can be organised along the following categories: modeling of epistemic planning, tractability and complexity, and implementation and applications.

On the modeling side, multiple articles have presented formalisms for multi-agent epistemic planning based on DEL \cite{andersen2012conditional, bolanderbirkegaard2011, lowe2011planning, yu2013multi}. These models are very expressive, capturing several key aspects of multi-agent epistemic planning. These aspects include: epistemic actions and goals, higher-order knowledge and belief, partial observability, etc. A thorough comparison of the present framework with existing DEL formalisms is found in Section \ref{subsec:DEL.litt.}. 

The rich expressivity of DEL comes at a cost, as planning problems specified in DEL are in general computationally difficult to solve (more on this below). This has partly motivated the introduction of simpler formalisms for epistemic planning. Some of these formalisms build on classical planning. The model in  \cite{petrick2004extending} extends STRIPS to allow knowledge declarations in preconditions and postconditions. The framework is however restricted to single-agent planning, does not support higher-order reasoning, and allows only a restricted form of quantification. The multi-agent planning frameworks in \cite{kominis2015beliefs,muise2015planning} follow a compilation approach, translating restricted fragments of epistemic planning into classical planning languages. 

The approaches in \cite{van2002tractable, jamroga2007constructive} describe planning domains via a type of state-transition system extended with epistemic information, called a \textit{concurrent epistemic game structure} (CEGS). This representation makes it easy to define multi-agent notions such as `joint action' or `multi-agent plan'. However, the representation inherits some of the well-known problems of transition-system models, including the lack of compact descriptions of actions and efficient heuristics that can avoid building the full state-transition system when planning (see  \cite{bolander2017gentle} for a discussion of these and other limitations). 

The non-DEL formalism that most closely resembles the approach of this paper is the epistemic game description language GDL-III \cite{thielscher2017gdl}. The language is epistemic and first-order.  A key feature of this language is that only what agents can see and do has to be defined. This is done via declarations that use the keywords $\mathbf{Sees}$ and $\mathbf{Does}$, which loosely correspond to modalities. GDL-III has a simple syntax and allows compact specifications of actions. For instance, the following GDL-III rules \cite{baral_et_al:DR:2017:8285} describe schematically communication actions which are similar to the ones from Example \ref{example:action_schemas}:
\[\mathbf{Sees}(x, i_a) \Leftrightarrow \mathbf{Does}(i, announce(z)) \wedge Obs(i, x)) \]
\[\mathbf{Sees}(y, \varphi) \Leftrightarrow \mathbf{Does}(i, announce(\varphi)) \wedge Obs(i, y) \wedge Listen(i, y) \wedge \varphi\]
These two transition rules are interpreted as follows: if agent $i$ announces $z$ then any agent $x$ observing $i$ will receive the information $i_a$. Any agent $y$ that observes $i$ and listens to $i$ will learn the content of the announcement, $\varphi$. Given the semantics of GDL-III, it follows that agents who only see $i_a$ will know that $i$ made an announcement but will not learn the content of the announcement. Agents who observe $i$, however, will know that $\varphi$ must be true. Moreover, if an agent $x$ observes agent $i$, does not listen to $i$, but knows that another agent $y$ listens to $i$, then the semantics entails that $x$ will know that $y$ will know the content of the announcement after it has been made. The model is therefore schematic and context-sensitive, like the epistemic action schemas presented here. The syntax of GDL-III is simpler than that of DEL when it comes to representing actions. However, as noted in \cite{Engesser2018}, specifying nested and higher-order knowledge is more difficult in GDL-III than in DEL, and the formalism requires more involved semantics. The work in \cite{Engesser2018} provides a detailed comparison of DEL and GDL-III, concluding that GDL-III offers a simpler syntax, while DEL provides simpler semantics. In \cite{Engesser2018}, it is shown that large fragments of GDL-III and DEL are equally expressive by giving compilations between the two.

Concerning decidability and complexity, it was first shown in \cite{bolanderbirkegaard2011} that the general plan existence problem in propositional DEL planning (i.e., deciding whether a plan exists given a multi-agent planning task) is undecidable. In fact, the problem is undecidable with two agents only, no common knowledge, and no postconditions. In \cite{le2018small} it is shown that public actions are enough for undecidability when the initial state meets certain technical conditions. That paper also identifies an undecidable subclass of small epistemic planning problems comprising two agents, one action, six propositions and a fixed goal. The undecidability results straighforwardly apply also to the present framework.

Although the general problem is undecidable, several papers have identified decidable fragments of epistemic planning that are still reasonably expressive. Single-agent epistemic planning is decidable \cite{bolanderbirkegaard2011}. The multi-agent problem becomes decidable if actions are only allowed to have propositional preconditions (i.e., no epistemic formulas appear in the preconditions) \cite{yu2013multi}. The computational complexity of this fragment belongs to $(d + 1)\textsc{-ExpTime}$ for a goal formula whose modal depth is $d$. If actions are restricted to have propositional preconditions and no postconditions, the plan existence problem becomes $\textsc{PSpace-}$complete \cite{charrier2016impact}. Stronger restrictions, such as allowing only private and public announcements, bring the complexity down to $\textsc{NP-}$complete \cite{bolander2015complexity}.
As mentioned in Section 6, in \cite{ijcai2020} we show that single-agent epistemic planning and multi-agent planning with non-modal preconditions are also decidable in the term-modal, first-order case described in this paper, echoing the results for propositional planning in \cite{bolanderbirkegaard2011,yu2013multi}.

As for implementation and applications, a number of techniques and planners have been developed over the last decade. An approach that has gained popularity is the compilation approach. The idea involves choosing a suitably restricted fragment of DEL that can be encoded in a classical planning language. Epistemic problems are then translated into classical ones so that state-of-the-art planners can be used to solve them efficiently. The compilations rely on different restrictions.  The system in \cite{kominis2015beliefs} assumes that actions are public, physical actions are deterministic, and that all agents start with a common initial belief on the set of worlds that are possible. The paper adopts a centralised perspective, with planning done off-line from the viewpoint of a single agent. In \cite{kominis2017multiagent}, the authors extend this framework to cover on-line planning from the perspective of the agents themselves. The planner in Muise \textit{et al.} \cite{muise2015planning} requires a finite depth of nesting of modalities and no disjunctions. Cooper \textit{et al.} \cite{cooper2016simple} use an encoding based on special variables describing what agents can see. The epistemic problems expressible with this restricted language are then encoded in PDDL and solved using the Fast Downward planner \cite{helmert2006fast}. As mentioned before, the PKS system in \cite{petrick2004extending} encodes epistemic planning using a STRIPS-like language. This language can describe single-agent, epistemic planning problems with conditional effects. The PKS system tries to solve these problems using an efficient but incomplete algorithm. 

A small number of epistemic planners do not rely on compilation into classical planning. The system MEPK \cite{huanggeneral} performs multi-agent epistemic planning from the third-person viewpoint. The system can handle private actions and beliefs, as formalized with the modal logic  KD45n. The systems does not support arbitrary common knowledge but can deal with a weaker form of common knowledge. Finally, Le \textit{et al.} \cite{le2018efp} present two forward planners, called EFP and PG-EFP, for multi-agent epistemic planning. These planners can deal with unlimited nested beliefs, common knowledge, and epistemic goals when the number of worlds in the initial state is not too large.

\subsection{\label{subsec:DEL.litt.}Dynamic Epistemic Logic} There is a vast and excellent literature on both epistemic logic and dynamic epistemic logic to which the reader is referred for both technical and conceptual introductions\textemdash see e.g. \cite{sep-dynamic-epistemic,benthem2011b,ditmarsch2007,Fagin_et_al._1995,TML_Hintikka1962,mossintro,Lee2018}.

The approach to modeling actions taken in this paper is based on the idea of \textit{action models} applied using \textit{product update} as introduced first by Baltag, Moss and Solecki \cite{BaltagBMS_1998}. The reduction axiom approach to proving completeness for logics with actions was first suggested by Plaza \cite{plaza1989} for the case of \textit{truthful public announcements}. Our approach is the same, but for general actions models. It is based on  \cite{baltagmoss2004,sep-dynamic-epistemic,ditmarsch2007}.

The format of the action models presented here differs mainly in four aspects from those introduced in \cite{BaltagBMS_1998}: our action models have \textit{postconditions}; are \textit{first-order} rather than propositional; accommodate \textit{term-modal relations}; and have \textit{conditioned edges}. 

Our approach to postconditions is inspired by \cite{benthem2006_com-change,ditmarsch_kooi_ontic}. From there, it is a straightforward generalization to alter pre- and postconditions to allow updates of first-order Kripke models.

A substantial departure from the standard is the accommodation of term-modal relations and the edge-conditioning. The definition avoids two problems for term-modal action models\textemdash one pointed out by Kooi \cite{TML_KooiTermModalLogic} and one concerning reduction axioms\textemdash by an adjustment of the propositional \textit{edge-conditioned} action models of Bolander \cite{Bolander2014}.\footnote{Both approaches result in context-sensitive actions: the distinguishability of two events depends on model to be updated. See \cite{Bolander2014,Lee2018,Rendsvig2014a,Rendsvig-DS-DEL-2015,Rendsvig2018_Thesis} for arguments to the effect that more context-sensitivity than what is present in standard action models is desirable.}

In the standard definition, an action model $A$ for index set of agents $\mathcal{I}$ consists of a finite set of events $E=\{e,...,e'\}$ and a map $R:\mathcal{I}\rightarrow\mathcal{P}(E\times E)$, plus assignments of pre- and postconditions. In the term-modal treatment, the set $\mathcal{I}$ is a proper part of the semantics of state representations. Adding an operator $[A,e]$ to the language thus conflates syntax and semantics, Kooi points out.

In considering reduction actions, we found that this problem runs deep. Consider the standard reduction axiom for the modal operator:
\begin{equation}
[A,e]K_{i}\varphi\leftrightarrow\left(pre(e)\rightarrow\bigwedge_{f:(e,f)\in R_{i}}K_{i}[A,f]\varphi\right)\label{eq:red.ax.classic}
\end{equation}
In (\ref{eq:red.ax.classic}), the agent index $i$ links the occurrences of the modal operator $K_{i}$ with the relation $R_{i}$ used in the quantifying conjunction. This link is broken in the term-modal treatment: the ``$i$'' indexing the operators is a syntactic term, while the ``$i$'' indexing the relation is an element of a domain of quantification. Without consulting an interpretation (or variable assignment), these two occurrences are unlinked: there is no guarantee that $R_{i}$ is the relation used in evaluating $K_{i}\varphi$.

To resolve the conflation problem, Kooi defines action models where accessibility relations over events $E$ are assigned to \textit{groups} of agents on a per-application basis: With $\Phi$ a finite set of mutually inconsistent and jointly exhaustive formulas with free variable $x$, each pointed model $(M,w)$ and variable valuation $v$ defines a partition on the agent domain with cells $\{d\in D_{agt}\colon M,w\vDash_{v[x/d]}\varphi(x)\}$ for each $\varphi(x)\in\Phi$; each agent in such a cell (group) is assigned the same accessibility relation using a map $S:\Phi\rightarrow\mathcal{P}(E\times E)$. In effect, the action model makes no direct reference to the agent domain, thus avoiding the conflation problem.

Additionally, Kooi's definition yields a solution to the problem of unlinked indices as the relations of the action model may now be referred to using syntactical constructs. With $\Phi=\{P_{1}(x),...,P_{n}(x)\}$, a suggestion for a reduction axiom could be

\[
[A,e]K_{t}\varphi\leftrightarrow\left(pre(e)\rightarrow\bigwedge_{k\leq n}\left(P_{k}(t)\rightarrow\bigwedge_{f:(e,f)\in R(P_{k})}K_{t}[A,f]\varphi\right)\right).
\]

We obtain a similar solution by adjusting the edge-conditioned action model of Bolander \cite{Bolander2014}. In an edge-conditioned action model, whether two events are related for some agent $i\in\mathcal{I}$ is conditional on whether a given formula is satisfied in the pointed model on which the action is executed. Formally, each agent-edge pair is assigned a condition by a map $Q:\mathcal{I}\rightarrow(E\times E\rightarrow\mathcal{L})$. 

Inspired by both Bolander and Kooi, we use a map $Q:E\times E\rightarrow\mathcal{L}$ where $Q(e,e')(x^\star)$ has exactly one free variable, $x^\star$. When the resulting action model is executed on a pointed model $(M,w)$, an edge is present for an agent $\alpha\in D_{\mathtt{agt}}$ if $M,w\vDash_{v[x^\star\mapsto\alpha]}Q(e,e')(x^\star)$. As the condition $Q(e,e')(x^\star)$ is a formula, this approach allows the formulation of reduction axioms, cf. Section \ref{subsec:red.ax}. 

Our version of the $Q$ function and Kooi's approach \textsc{$S$} are equally general. Given an action model $(E,S,\mathsf{pre},\mathsf{post})$ with $S:\Phi\rightarrow\mathcal{P}(E\times E)$, let $Q:E\times E\rightarrow\mathcal{L}$ be given by $Q(e,e')=\varphi$ such that $(e,e')\in S(\varphi)$. Then $Q$ emulates $S$: for all models $M$, $M\otimes(E,S,\mathsf{pre},\mathsf{post})=M\otimes(E,Q,\mathsf{pre},\mathsf{post})$. \textit{Vice versa}, to emulate a map $Q$, for each $A\subseteq E\times E$, let 
\begin{equation}
\varphi_{A}:=\bigwedge_{\varphi\in Q(A)}\varphi\wedge\bigwedge_{\psi\in \Psi}\psi\label{eq:kooi}
\end{equation}
with $\Psi$ the largest subset of $\{\neg\varphi\colon\varphi\in Q(E\times E)\backslash Q(A)\}$ such that (\ref{eq:kooi}) is consistent. Then $S:\varphi_{A}\mapsto A$ for each $A\subseteq E\times E$ is a Kooi map that emulates $Q$. We opt for the edge-conditioned formulation due to its correspondence with the standard precondition maps $\mathsf{pre}:E\rightarrow\mathcal{L}$.

Finally, note that both may emulate standard action models over classes of models where each agent $\alpha$ is designated by a rigid constant $a_{\alpha}$ (as is conceptually implied by identifying agents with indices). The standard map $R:D_{\mathtt{agt}}\rightarrow\mathcal{P}(E\times E)$ may be emulated by the map $Q:E\times E\rightarrow\mathcal{L}$ with $Q(e,e')=\bigvee_{a_{\alpha}\colon(e,e')\in R_{\alpha}}(x=a_{\alpha})$.

\subsection{\label{subsec:TML.litt.}Term-Modal Logic}

The term-modal treatment of epistemic operators as behaving both as modal operators and as first-order predicates was suggested already by von Wright in his 1951 \cite{vonWright1951}, though the direction was not formally explored. Formally, Hintikka allowed the constructions in his 1962 \cite{TML_Hintikka1962}, and the term-modal aspects are used in discussions concerning the validity of $K_{a}\varphi\rightarrow K_{a}K_{a}\varphi$, where Hintikka notes that the schema is only valid if $a$ \textit{knows who }$a$\textit{ is}, captured by $\exists xK_{a}(x=a)$ (see also Section \ref{subsec:frame.char} on frame characterizations). Semantically, Hintikka linked individuals and operators in \cite{TML_Hintikka1969} using world-relative first-order interpretations extended to assign alternatives to individuals in the domain of quantification, $D$. Work in philosophical logic followed Hintikka's term-modal syntax\textemdash even called ``standard'' by Carlson in 1988 \cite{TML_Carlson_1988}\textemdash but the semantic link did not pertain: \cite{TML_Sleigh_1972} exemplifies a pseudo-use. Carlson enforced the semantic link, using a partial map $R:D\longrightarrow\mathcal{P}(W\times W)$ to assign accessibility relations to individuals. He further presents a Hintikka-style model set proof theory for a three-valued Kripke-style semantics with non-rigid terms, varying domains and reflexive relations, and shows completeness. 

In computer science, a format similar to Carlson's is frequently used when giving the semantics for propositional epistemic logic, with the set of agents $D$ treated as an \textit{index set} instead of a domain of quantification, even in the first-order case: e.g., in Fagin \textit{et al.}'s first-order treatment \cite{Fagin_et_al._1995}, in a formula like $K_{Alice}Governor(California,Pete)$, both $California$ and $Pete$ are first-order terms, but $Alice$ is not\textemdash $Alice$ is an agent. Here, then, agents and their names are \textit{equated}. 

The issue of equating agents and their names, and why this is unsatisfactory in many computer science applications, is discussed at length by Grove \& Halpern \cite{GroveHalpern1993} and Grove \cite{Grove1995}. They identify the following inadequacies: systems that equate agents with their names cannot represent agent sets of non-fixed size, do not allow for reference to agent groups, for non-rigid names, nor for indexical and relative reasoning (using terms like ``me'' to express e.g. ``the agent to the left of me''). In response, \cite{GroveHalpern1993} develops a propositional epistemic logic with indexical reference obtained by evaluating formulas at agent-world pairs, on which \cite{Grove1995} builds a first-order variant to additionally handle issues of \textit{de dicto/de re}-like reference scope, as well as multi-naming of agents. The latter is in effect a variant of non-rigid constant, varying domain, term-modal logic with formulas evaluated at agent-world pairs. Further, the language contains two sorts, one for \textit{agents} (like our \texttt{agt} terms), and one for \textit{names}.\footnote{Such two sorts are also used by Rendsvig in a quantified, but not term-modal, epistemic logic analysis of semantic competence in relation to Frege's puzzle about identity \cite{Rendsvig2011Thesis, Rendsvig2012a}.} This allows explicit reasoning about \textit{naming}. Adding a third sort to the present framework would be unproblematic, but the indexical semantics would require in-depth re-working. Similarly would varying agent domains require work, unless emulated by an existence predicate, cf. \cite{Fitting_&_Mendelsohn}. Beyond this, the present framework tackles the issues raised in \cite{Grove1995, GroveHalpern1993}: agents and names are not equated by the use of (non-rigid) constants of sort \texttt{agt}, that additionally allow for multi-naming; agents groups may be denoted by predicates and relative properties by relations; finally, \textit{de dicto/de re} distinctions are expressible using quantification. However, beyond the formal difference and similarities, we would find an in-depth philosophical comparison of the interpretation of the two frameworks interesting. 
In \cite{LibermanRendsvig2019}, we illustrate the system presented here with examples that touch on several of the involved issues.

One reason for sticking with ordinary modal operators even in a first-order setting is that term-modal operators adds design choices and possible complications, as discussed by Lomuscio \& Colombetti in their early contribution to the term-modal literature \cite{TML_Lomuscio1997}. In constructing a term-modal extension of multi-agent KD45 with non-rigid terms, they discuss how to evaluate formulas $B_{a}\varphi$ when $a$ is not an agent denoting term. Intuitively, $B_{a}\varphi$ should be false, as only agents can truly hold beliefs, but\textemdash they remark\textemdash this would imply the invalidity of $B_{a}(\varphi\vee\neg\varphi)$. They conclude against a two-sorted approach, as a similar problem surfaces for formulas $B_{a}B_{b}\varphi$ when agent $a$ believes that the term $b$ denotes a non-agent.\footnote{This obstacle is avoided in the present paper by syntactically forcing all operator-subscripts to be of the agent-sort.} Ultimately, Lomuscio \& Columbetti opt for a \textit{partial logic} with \textit{truth-value gaps}, letting the truth-value of $B_{a}\varphi$ be \textit{undefined} when $a$ denotes a non-agent; they take a \textit{valid} formula to be \textit{sometimes satisfied}, but \textit{never false}. Their semantics are constant domain, and each element is, at each world, assigned a set of \textit{doxastic alternatives}; an element is an \textit{agent in world $w$} if it is assigned a non-empty set. Hence, agenthood is \textit{world-relative}. They present an axiom system\textemdash which includes a term-modal Barcan formula $\forall y(B_{x}\varphi(y))\rightarrow B_{x}\forall y(\varphi(y))$ and quantified frame-characterizing formulas like $\forall x(B_{x}\varphi\rightarrow B_{x}B_{x}\varphi)$ like the present paper\textemdash and show soundness, citing \cite{Lomuscio95} for details.

Bivalent systems are presented by Thalmann \cite{TML_Thalmann2000a} and Fitting, Thalmann \& Voronkov \cite{TML_Fitting2001}, with these two works coining the label `term-modal logic'. In their setting, each world $w$ is associated with an \textit{inner domain} $D(w)$ of objects existing at $w$, with $D(w)$ a subset of the \textit{outer domain $D$,} for all $w$. The inner domains are assumed \textit{increasing}: if $wR_{d}w'$ for some $d\in D$, then $D(w)\subseteq D(w')$. Further, terms are assumed rigid and with an interpretation defined at every world ($I(c)\in D(w)$ for all $w\in W$). This combination seemingly\footnote{Seemingly, as we are confused about the satisfaction clause for atomic formulas \cite[Def. 7, It. 1]{TML_Fitting2001}, stating that $w,V\Vdash R(t_{1},...,t_{n})$ iff $w\Vdash R(V(t_{1}),...,V(t_{n}))$ with $V(t_{i})\in D$, but no specification of the conditions for the right-hand condition, nor any specification of how the relation symbol $R$ is assigned extension. However, if this is assumed settled as ordinarily (as in the present paper), the increasing domain assumption seems sufficient to obtain a well-behaved semantics, as is the case in ordinary first-order modal logic. See e.g. \cite{Gamut_II} for an introduction and \cite{Hughes_&_Cresswell_New.Intro} for details.} eliminates the need for truth-value gaps, but the problems raised by non-agents are not discussed. For several classic frame-conditions, \cite{TML_Fitting2001,TML_Thalmann2000a} presents both sequent and tableau proof systems (K, D, T, K4, D4, S4).

Orlandelli \& Corsi \cite{TML_Decidable2018} also investigate sequent calculi for term-modal logics. Their semantics is more general as they omit the increasing domain requirement, and---as they also consider Euclidean frames---they also obtain completeness for more frame classes. The syntax is without constants, so the rigidity/non-rigidity dichotomy is non-applicable. The semantics are bivalent. The combination of varying domains and bivalent semantics is facilitated by the atomic formula satisfaction clause 
\[
M,w\vDash_{v}r(x_{1},...,x_{n})\text{ iff }(v(x_{1}),...,v(x_{n}))\in I(r,w),
\]
with $I(r,w)\subseteq D^{n}$ again with $D$ the outer domain. E.g., with $I(=,w)=\{(d,d)\in D^{2}\colon d\in D\}$, the formula $(x=x)$ is satisfied in $(M,w)$ even if $v(x)\notin w$. However, as the quantifiers only range over the inner domain of worlds, the semantics oddly make $p(x)\wedge\forall y\neg p(y)$ satisfiable.

In \cite{TML_KooiTermModalLogic}, Kooi introduces a \textit{dynamic} term-modal logic, including a first-ever first-order version of DEL action models. The language of \cite{TML_KooiTermModalLogic} is first-order dynamic logic with wildcard assignment, but where the set of first-order terms is also the set of atomic programs, the models for which are constant agents-only domain with non-rigid terms (and very similar to our general case, but restricted to agents-only). This language is more expressive than ordinary term-modal logic. The first-order dynamic logic aspect implies that the validity problem is $\Pi_{1}^{1}$ complete, eliminating hope for a finitary proof system. However, the expressivity of the language allows the definition of a \textit{non-rigid common knowledge}. If not for our two-sorted domain, our language and semantics could be seen as a special case of Kooi's. Kooi's action models are discussed in the next section. 

Seligman \& Wang \cite{TML_SeligmanWang2018} investigate a fragment Kooi's system. The fragment allows only basic assignment modalities to form a quantifer-free term-modal logic (without function symbols), a fragment rich enough to express \textit{de dicto}/\textit{de re} distinctions and \textit{knowing who} constructions in a setting where names are not common knowledge. The main result is a complete axiomatization for the class of S5 models. As Barcan-like formulas are not included in the investigated language fragment but are the common characterizers of constant domain semantics, this result is quite non-standard. The authors also discuss decidability: providing no hard results, they conclude ``We are not that far from the decidability boundary, if not on the wrong side.'' 

Corsi \& Orlandelli \cite{Add_Corsi2013} introduce a generalization of term-modal syntax to be able to express the difference between \textit{de dicto} and \textit{de re} statements without invoking quantifiers. They introduce complex term-modal operators $|t\colon_{x}^{c}|p(x)$ with the reading that $t$ knows of $c$ that (s)he is $p(x)$. These are interpreted over so-called \textit{epistemic transition structures} with \textit{double-domains}. The resulting \textit{indexed epistemic logics} are further investigated in \cite{Add_Corsi2014,Add_Corsi2016}. It would be interesting to know what the relationship is to the also expressive language of Kooi \cite{TML_KooiTermModalLogic}.

Where the domain of Kooi \cite{TML_KooiTermModalLogic} consists only of agents, Rendsvig \cite{TML_Rendsvig2010} introduces a model with a single-sorted language with non-rigid terms that denote elements in a constant domain containing both agents and objects. As in \cite{TML_Lomuscio1997}, this requires an \textit{ad hoc} solution to the semantics of formulas $K_{a}\varphi$ when $a$ denotes a non-agent. The solution used is to then interpret $K_{a}\varphi$ as a \textit{global }modality. This preserves the bivalence of the systems while making all operators normal. As a result, \cite{TML_Rendsvig2010} presents a canonical model theorem, facilitating completeness proofs for classic frame classes.

The semantics of this paper are based on Achen's \cite{Achen2017}, which in turn is a two-sorted refinement of \cite{TML_Rendsvig2010}. What we consider an improvement of \cite{Achen2017} over \cite{TML_Rendsvig2010} is exactly the two-sorted approach: distinguishing between agent and object terms removes the need to define \textit{ad hoc} semantics for knowledge operators indexed by non-agents. Taking a two-sorted approach eliminates the possibility of modeling agents that are uncertain about whether a given term refers to an agent or an object, but results in a system which we consider well-behaved.

Term-modal like, Naumov \& Tao \cite{Naumov2018} present a propositional term-modal logic, but where operators may be indexed by sets of terms, making $\exists xK_{\{x,a\}}\varphi$ a formula. Such operators are given a \textit{distributed knowledge} semantics in S5 models with constant agents-only domain and rigid terms for which a complete axiom systems is presented.

Sawasaki, Sano and Yamada \cite{Sawasaki2019} consider a term-modal syntax where operators are indexed by a sequence of terms making e.g. $\forall x \forall y K_{[x,y]} R(x,y)$ well-formed, with the intended deontic reading that $x$ is obliged by $y$ to ensure $R(x,y)$. They present complete axiom system and sequent calculi.

Sedlar \cite{TML_IgorSedlar2014--INCOMPLETE-REF} also uses a rigid terms, agents-only constant domain semantics to represent an epistemic logic of evidence using a term-modal language as that presented here. Sedlar shows that his term-modal framework is able to emulate monotonic modal logics and epistemic logics with awareness, obtaining a decidability result for the fragment with no constants nor functions, but $0$-ary predicates and single unary predicate.

Several other authors have also looked at decidability issues for varieties of term-modal logics. Kooi \cite{TML_KooiTermModalLogic} points out that the monadic fragment of his system is undecidable by a result of Kripke \cite{Kripke1962}. As Kripke's result concerns first-order modal logic in general (see e.g. \cite[p. 271 ff.]{Hughes_&_Cresswell_New.Intro}), it applies to broadly to term-modal logics, too. For term-modal logics, Padmanabha \& Ramanujam \cite{TML_Padmanabha2018} even show that the propositional fragment is undecidable. As decidable, they identify the monodic fragment (formulas using only one free variable in the scope of a modality).
\cite{Padmanabha2017} considers model checking for the fragment over a restricted model class and \cite{Padmanabha2019} presents a translation of the monodic fragment (without identity) into \textsf{FOML}.

In \cite{Padmanabha2019b}, Padmanabha \& Ramanujam further investigate a variable-free propositional bi-modal logic with implicit quantification, with formulas $[\forall]\varphi$ and $[\exists]\varphi$ asserting that along all (resp. some) accessibility relations $\varphi$ is necessary. These variable-free formulas thus correspond to the propositional term-modal formulas $\forall x K_x \varphi$ and $\exists x K_x \varphi$. The relevant logic is shown decidable, to be bisimulation-invariant fragment of an appropriate two-sorted first-order logic, related to the `bundled fragment' of term-modal logic.
Model checking for the system is investigated in \cite{Padmanabha2020}.
In \cite{Padmanabha2019a} Padmanabha \& Ramanujam, turn to the two variable fragment of term-modal logic, which they show decidable. 
The thesis \cite{Padmanabha2019} collects a selection of the mentioned results, and additionally presents a translation of \textsf{TML} without identity into propositional \textsf{TML}.

For their own system, Orlandelli \& Corsi \cite{TML_Decidable2018} show two fragments decidable, the first propositional with quantifiers and operators occurring only in pairs of the forms $\exists x[x]$ or $\forall x\langle x\rangle$. This fragment simulates non-normal monotone epistemic logics. The second fragment allows expressing $1$-ary groups' higher-order knowledge about proposition symbols, e.g. with $\forall x(p(x)\rightarrow K_{x}(K_{y}q))$ an allowed formula. Also  Pliu\v{s}kevi\v{c}ius  \&  Pliu\v{s}kevi\v{c}ien\.{e} \cite{TML_Pliuskevicius2006a} treats a fragment of propositional term-modal logic, but with, term-modal operators for belief and mutual belief, allowing only pair-wise quantifier-operator nestings (e.g., for $p$ a propositional atom, $\forall xB_{x}\exists yB_{y}p$ is well-formed, while $\forall x\exists yB_{x}B_{y}p$ is not). For their agents-only constant domain KD45 semantics, they present a terminating sequent calculus decision procedure. For further decidability results, it may be relevant to consult Shtakser \cite{Add_Shtakser2018}, who investigates propositional modal languages includes quantification over modal operators and predicate symbols that take modal operators as arguments.

Beyond its main decidability result, Padmanabha \& Ramanujam \cite{Padmanabha2019a} also discusses translation of term-modal logic into first-order modal logic. In a setting with no constants or function symbols, the authors suggest a translation of \textsf{TML} into \textsf{FOML} with a single modality $K$ and a new unary predicate $P$, inductively translating $K_{x}\varphi$ to $K(P(x)\rightarrow\varphi)$ and $\hat{K}_{x}\varphi$ to $\hat{K}(P(x)\wedge\varphi)$.
\cite{Padmanabha2019a} omits the details, but claims this translation produces \textsf{FOML }formulas equi-satisfiable with their \textsf{TML }originals. This suggests that completeness results for term-modal logics may also be shown indirectly via translation and application of well-known results for \textsf{FOML} (see e.g. \cite{Hughes_&_Cresswell_New.Intro}), instead of by the direct constructions found in the Appendix.\footnote{We thank a reviewer for pointing this out.} Whether a translation approach would work for the present framework is an open question, but we have reservations concerning the general applicability of the suggested translation.
\footnote{We hold a reservation as satisfiability is not generally preserved by the translation. In the class of \textsf{TML}  models with exactly 2 agents (characterized by axiom N for $n=2$) wlog called $\alpha$ and $\beta$, with constants $a$ and $b$ locally rigid, but non-identical (characterized by $\exists x\exists y((x\neq y)\wedge(x=a)\wedge(y=b)\wedge\forall zK_{z}((x=a)\wedge(y=b)))$), and satisfying for $i,j\in\{\alpha,\beta\},i\neq j,$ $\forall x,y,z\in W,\text{ if }xR_{i}y\text{ and }xR_{j}z,\text{ then }yR_{i}z$ (characterized by $\forall x\forall y(((x\neq y)\wedge\hat{K}_{x}\top\wedge\hat{K}_{y}\varphi)\rightarrow K_{x}\hat{K}_{y}\varphi)$), the formula $\exists x\exists y((x\neq y)\wedge\hat{K}_{y}\top\wedge\hat{K}_{x}\hat{K}_{x}\top\wedge K_{x}K_{x}\neg\hat{K}_{x}\top)$ is satisfiable. However, the translation of the latter is not satisfiable in the class of \textsf{FOML} models characterized by the translation of the three former.}

\section{Final Remarks}\label{sec:final.remarks}
We conclude with open questions we see in relation to epistemic planning, and a summary of the main contributions of the paper. The following are some possible avenues for future research:
\begin{enumerate}
    \item \textit{Decidability and complexity}. As presented in the literature review on epistemic planning with propositional DEL (Section \ref{subsec:rel.work.ep.plan}), results exist concerning the undecidability of several classes of epistemic planning problems, but decidability and complexity results also exist. It is clear that the negative results apply in the richer setting of this paper. In \cite{ijcai2020}, we show that some of the positive decidability results can be established in the decidable finite-agent setting of dynamic term-modal logic (i.e., decidability for single-agent planning and multi-agent planning with non-modal preconditions). It is an open question whether any other decidability results can be extended as well, and the complexity of first-order epistemic planning has not been studied. 
    \item \textit{Reasoning about schematic actions}. In extension to defining first-order variants of action models, it was natural to define action schemas to obtain succinct action representations. These action schemas are however not described by the dynamic languages and logics introduced. We find it an interesting question how the languages and logics should be altered to obtain a logic of action schemas. Constructing such a logic could possibly draw connections to recent work on \textit{Arbitrary Public Announcemnet Logic} and its generalizations, cf. e.g. \cite{APAL_2007,APAL_2016}.
    \item \textit{Supporting other planning features}. A possibly fruitful avenue for future research is to devise a first-order \textit{probabilistic} DEL framework for probabilistic epistemic planning. In the standard planning literature, probabilistic PDDL is often used to support probabilistic effects, allowing the specification of Markov decision processes \cite{younes2004ppddl1}. There is a rich literature on probabilistic propositional DEL on which a first-order setting for probabilistic epistemic planning could be based (for an overview, see  \cite[Appendix L]{sep-dynamic-epistemic}). Other well-known planning features, such as numeric fluents, temporal aspects, etc., could also be integrated.
\end{enumerate}{}

Finally, we briefly recall what we see as the main contributions of the paper:

\begin{enumerate}
    \item \textit{A first-order dynamic epistemic logic}. The paper develops novel dynamics for a variant of term-modal logic with the addition of first-order action models. It thereby generalizes propositional DEL to a setting allowing full first-order epistemic reasoning about both objects and agents.
    \item \textit{A compact epistemic domain definition language}. As the epistemic planning formalism developed builds on first-order logic, it allows for a compact specification of domain dynamics via \textit{epistemic action schemas}. Such schematization is inspired by that used in PDDL, and to the best of our knowledge, it provides the most compact representation of actions available in the DEL framework. The setting conservatively extends propositional DEL, in the sense that it contains it as a special case, inheriting the ingredients of the DEL planning framework. 
    \item \textit{Expressive, yet decidable axiom systems for reasoning about epistemic actions}. On the reasoning side, the paper develops static and dynamic axiom systems that are well-behaved. Although the logical languages proposed are fairly expressive, it is shown that sound, complete and decidable systems exist for several natural classes of models.
\end{enumerate}{}

\section*{Acknowledgements}
We sincerely thank the three anonymous reviewers for their insightful questions, comments and criticisms: We appreciate your time and efforts.

The Center for Information and Bubble Studies is funded by the Carlsberg Foundation. RKR was partially supported by the DFG-ANR  joint project Collective Attitude Formation [RO 4548/8-1].

\appendix

\section{\label{A.proof.appendix}Proof Appendix}

\subsection{Term-Modal Logic}

This section establishes the results stated in Section \ref{subsec:Normal-Term-Modal-Logic}.
The logic $\mathsf{K}$ is well-behaved, with standard techniques
for establishing strong completeness carrying over from the propositional
and quantified modal logic cases. Therefore, the section presents
only proof strategy, with non-standard elements given special attention.
Full details may be found in \cite{Achen2017}.

The involved notions are standard (see e.g. \cite{BlueModalLogic,Brauner&Ghilardi-FOML,Hughes_&_Cresswell_New.Intro}),
but we remark that a formula $\varphi$ is \textit{valid} over a class
of frames $\boldsymbol{X}$ iff for every frame $F=(D,W,R)\in\boldsymbol{X}$,
every interpretation $I$ over $F$, every world $w\in W$ and every
valuation $v$, it is the case that $M,w\vDash_{v}\varphi$. That
$\varphi$ is a \textit{semantic consequence} of the formula-set $\Gamma$
over a class $\boldsymbol{X}$ is written $\Gamma\vDash_{\boldsymbol{X}}\varphi$.
For $\varphi$ provable from the assumptions $\Gamma$ in the logic $\Lambda$,
write $\Gamma\vdash_\Lambda\varphi$. In both cases, when $\Gamma=\emptyset$, it is omitted.%

\subsubsection{\label{A.static.soundness}Soundness}

\begin{prop}
The system $\mathsf{K}$ is sound with respect to the class $\boldsymbol{F}$
of all frames: for all $\varphi\in\mathcal{L}$, if $\vdash_{\mathsf{K}}\varphi$,
then $\vDash_{\boldsymbol{F}}\varphi$.
\end{prop}
\begin{proof}
The proof is standard: the axioms of $\mathsf{K}$ are shown valid
over $\boldsymbol{F}$ and the rules of inference are shown to preserve
validity. To give a feel, arguments follow for the K axiom and the
Barcan~Formula.\medskip{}

\noindent K: Let $M$ be a model based on an arbitrary frame $F\in\boldsymbol{F}$,
let $w\in M$ and let $v$ be a valuation; let $K_{t}\top,\varphi,\psi\in\mathcal{L}$.
To show that $M,w\vDash_{v}K_{t}\big(\varphi\rightarrow\psi\big)\rightarrow\big(K_{t}\varphi\rightarrow K_{t}\psi\big)$,
assume $M,w\vDash_{v}K_{t}\big(\varphi\rightarrow\psi\big)$. As
$K_{t}\top\in\mathcal{L}$, $\left\llbracket t\right\rrbracket _{w}^{I,v}\in D_{\mathtt{agt}}$
by assumption. Hence $F$ contains an accessibility relation $R_{\left\llbracket t\right\rrbracket _{w}^{I,v}}$.
Having fixed the accessibility relation going though the term $t$
to the agent domain, the argument is standard: By the semantics of
$K_{t}$, $M,w'\vDash_{v}\varphi\rightarrow\psi$ for every $w'\in M$
with $w'\in R_{\left\llbracket t\right\rrbracket _{w}^{I,v}}(w)$.
Hence $M,w'\vDash_{v}\neg\varphi$ or $M,w'\vDash_{v}\psi$. If
all such $w'$ satisfies $\varphi$, $M,w\vDash_{v}K_{t}\varphi$;
but then each $w'$ must also satisfy $\psi$, so $M,w\vDash_{v}K_{t}\psi$,
and hence $M,w\vDash_{v}K_{t}\varphi\rightarrow K_{t}\psi$. Else,
some such $w'$ satisfies $\neg\varphi$; then $M,w\vDash_{v}\neg K_{t}\varphi$,
so $M,w\vDash_{v}K_{t}\varphi\rightarrow K_{t}\psi$.

\medskip{}

\noindent \textbf{}%
BF: Let $M,w,v,\varphi$ and $t$ be as above. Pick a variable $x\neq t$
and assume that $M,w\vDash_{v}\forall xK_{t}\varphi$. Then for all
$x$-variants $v'$ of $v$, $M,w\vDash_{v}K_{t}\varphi$ (i.e., intuitively,
if $x$ is free in $\varphi$ so that $K_{t}\varphi(x)$ defines a
predicate, all elements in the $\mathtt{t}(x)$-domain of $w$ fall
in this predicate's extension). From $M,w\vDash_{v}K_{t}\varphi$,
it follows that for all $w'\in R_{\left\llbracket t\right\rrbracket _{w}^{I,v}}(w)$,
$M,w'\vDash_{v'}\varphi$ (intuitively, as $v'$ is an arbitrary $x$-variant
$v$, all $\mathtt{t}(x)$-elements \textit{existing in }$w'$ fall
in the extension of $\varphi(x)$. This would not hold if elements
could exist in $w'$ that do not exist in $w$). As $v'$ is an arbitrary
$x$-variant of $v$, it follows that $M,w'\vDash_{v}\forall x\varphi$
(again, illegitimate \textit{if} new elements could spring to existence).
As $w'$ was arbitrary from $R_{\left\llbracket t\right\rrbracket _{w}^{I,v}}(w)$,
finally $M,w\vDash_{v}K_{t}\forall x\varphi$.
\end{proof}

\subsubsection{Completeness \label{subsec:Completeness}}

This section establishes that the system $\mathsf{K}$ is \textit{strongly
complete} with respect to the class $\boldsymbol{F}$ of all frames.
I.e.,
\[
\text{for all }\Gamma\subseteq\mathcal{L},\text{for all }\varphi\in\mathcal{L},\text{ if }\Gamma\vDash_{\boldsymbol{F}}\varphi,\text{ then }\Gamma\vdash_{\mathsf{K}}\varphi.
\]

This follows as a corollary of the section's main result, the \textit{Canonical
Class Theorem} (Theorem \ref{thm:Canonical}) which states that any
normal term-modal logic is strongly complete with respect to its \textit{canonical
class}.

The theorem is establish by appeal to the following well-known\footnote{See e.g. \cite[p. 194]{BlueModalLogic}.}
proposition linking satisfaction and completeness:
\begin{prop}
\label{prop.iff} A logic $\Lambda$ is strongly complete with respect
to a class of structure $\boldsymbol{S}$ iff every $\Lambda$-consistent
set of formulas is satisfiable on some $s\in\boldsymbol{S}$.
\end{prop}
By this proposition, a completeness proof can be undertaken as an
existence proof: For a consistent set of formulas $\Gamma$, a satisfying
model from the appropriate class must be found. In the propositional
case, one model is constructed for all consistent sets simultaneously,
giving rise to the propositional \textit{Canonical Model Theorem} (see
e.g. \cite{BlueModalLogic}): any normal propositional modal logic
is strongly complete with respect to its \textit{canonical model}. 

The present proof cannot rely on single canonical model. As variables
are semantically rigid and any signature $\Sigma$ includes identity, the
same identity statements between variables are true across all worlds
of any model-valuation pair. A canonical model defined as usual would
not satisfy this: with consistent sets forming the basis of worlds,
if two worlds are disconnected by all accessibility relations, then
they need not satisfy the same identity statements between variables.
Hence, a rigid variable valuation cannot be defined. Further, different
$\mathsf{K}$-consistent sets may give rise to different domains.
Hence, non-constant domains result, and the construction is thus not
of the appropriate class. Therefore, our construction is of a canonical
model \textit{per }consistent set, resulting in a \textit{canonical class}.

The construction contains first-order aspects irrelevant in the propositional
case and term-modal logical aspects irrelevant to the standard quantified
case, but the approach is familiar: worlds are maximally consistent
sets that bear witnesses, ensured constructable by Lindenbaum-like
lemmas; domains are equivalence classes of variables induced by identity
statements; and canonical accessibility relations, interpretation
and valuation are defined as expected. That the canonical accessibility
relations are well-defined requires an additional lemma, but a familiar
Existence Lemma facilitates a familiar Truth Lemma, which in combination
with the above Proposition \ref{prop.iff} yields the main result.

\paragraph{Canonical Worlds}

Fix a signature $\Sigma=(\mathtt{V},\mathtt{C},\mathtt{R},\mathtt{F},\mathtt{t})$,
its language $\mathcal{L}$ and a normal term-modal logic $\Lambda\subseteq\mathcal{L}$.
When a set $\Gamma\subseteq\mathcal{L}$ is \textit{maximal $\Lambda$-consistent}
(defined as usual \cite{BlueModalLogic}), call $\Gamma$ a \textit{$\Lambda$-mcs}.

Maximal consistency does not suffice for a set to be a canonical world
in the first-order case. It must also be ensured that whenever a formula
of the form $\neg\forall x\varphi$ is included in $\Gamma$, then
$\Gamma$ must bear witness to this ``falsity'' of $\forall x\varphi$:\footnote{Witnesses bearing is called the \textit{$\forall$-property} in \cite[p. 257]{Hughes_&_Cresswell_New.Intro};
that the set is \textit{saturated} is also used in the literature.}
\begin{defn}
A set $\Gamma\subseteq\mathcal{L}$ \textit{bears witnesses }if for
every $\varphi\in\mathcal{L}$, for every variable $x$, there is
some variable $y$ such that $\big(\varphi(y/x)\rightarrow\forall x\varphi\big)\in\Gamma$.
\end{defn}
If a set $\Gamma$ bear witnesses, then so does every super-set of
$\Gamma$. If $\Gamma$ is a $\Lambda$-mcs that bears witnesses and contains
$\neg\forall x\varphi$, then for some $y\in\mathtt{V}$, $\neg\varphi(y/x)\in\Gamma$.

To ensure that every $\Lambda$-mcs can be extended to one bearing
witnesses, countably infinite sets of both agent and object variables
beyond those in $\mathtt{V}$ are needed. Define the extended signature
$\Sigma^{+}$ as $(\mathtt{V^{+}},\mathtt{C},\mathtt{R},\mathtt{F},\mathtt{t}^{+})$
where $\mathtt{V}\subseteq\mathtt{V}^{+}$, $\mathtt{t}^{+}(x)=\mathtt{t}(x)$
for all $x\in\mathtt{V}\cup\mathtt{C}\cup\mathtt{R}\cup\mathtt{F}$
and both $(\mathtt{t}^{+})^{-1}(\mathtt{agt})\cap\mathtt{V}^{+}\backslash\mathtt{V}$
and $(\mathtt{t}^{+})^{-1}(\mathtt{obj})\cap\mathtt{V}^{+}\backslash\mathtt{V}$
are countably infinite. Let $\mathcal{L}$$^{+}$ be the term-modal
language based on $\Sigma^{+}$. Then $\mathcal{L}\subseteq\mathcal{L}^{+}$.
The following two lemmas then ensure that the worlds of the canonical
models are constructable:
\begin{lem}
[Lindenbaum]\label{lem:Lindenbaum} If $\Gamma\subseteq\mathcal{L}$
is $\Lambda$-consistent, then there is a $\Lambda$-mcs $\Gamma^{'}$
such that $\Gamma\subseteq\Gamma^{'}$.
\end{lem}
\begin{lem}
[Witnessed]\label{lem:Witnessed} If $\Gamma\subseteq\mathcal{L}$
is $\Lambda$-consistent, then there is a set $\Gamma^{+}\subseteq\mathcal{L}^{+}$
such that $\Gamma\subseteq\Gamma^{+}$ and $\Gamma^{+}$ bears witnesses.
\end{lem}

\paragraph{\label{subsec:Canonical-Models}Canonical Models}

To avoid the issue remarked in this section's introduction, a canonical model is defined per $\Lambda$-mcs, ensuring that all worlds share its \textit{identity theory:}
\begin{defn}
The sets $\Gamma,\Gamma'\subseteq\mathcal{L}^{+}$ have the same \textbf{\textit{identity
theory}} if for all $x,y\in\mathtt{V}^{+},(x=y)\in\Gamma$ iff $(x=y)\in\Gamma'$. 
\end{defn}
\begin{defn}\label{def.Canonical.Model}
Let $\Lambda\subseteq\mathcal{L}$ be a normal term-modal logic. Let $\Gamma\subseteq\mathcal{L}$ be $\Lambda$-consistent and let $\Gamma^{*}\subseteq\mathcal{L}^{+}$ be maximal $\Lambda$-consistent, witness bearing and such that $\Gamma\subseteq\Gamma^{*}$ (existing by Lemmas \ref{lem:Lindenbaum} and \ref{lem:Witnessed}).
The \textbf{\textit{canonical model }}for \textbf{\textit{$(\Lambda,\Gamma^{*})$
}}is $M_{(\Lambda,\Gamma^{*})}=(D,W,R,I)$ such that 
\begin{enumerate}
\item $D\coloneqq D_{\mathtt{agt}}\dot{\cup}D_{\mathtt{obj}}\coloneqq\left\{ \left[x\right]\colon x\in(\mathtt{t}^{+})^{-1}(\mathtt{agt})\cap\mathtt{V}^{+}\right\} \dot{\bigcup}\left\{ \left[y\right]\colon y\in(\mathtt{t}^{+})^{-1}(\mathtt{agt})\cap\mathtt{V}^{+}\right\} $
where $\left[z\right]\coloneqq\left\{ z'\in\mathtt{V}^{+}\colon\left(z=z'\right)\in\Gamma^{*}\right\} $.
\item $W$ is the set of all maximal $\Lambda$-consistent, witness bearing
sets of formulas from $\mathcal{L}^{+}$ that share identity theory
with $\Gamma^{*}$.
\item $R:D_{\mathtt{agt}}\rightarrow\mathcal{P}(W\times W)$ such that for
all $\alpha\in D_{\mathtt{agt}}$, $(w,w')\in R(\alpha)$ iff for
every formula $K_{x}\varphi\in\mathcal{L}^{+}$ with $x\in\alpha$,
if $K_{x}\varphi\in w$, then $\varphi\in w'$,
\item and 
\begin{enumerate}
\item $I(r,w)=\left\{ \big([x_{1}],...,[x_{n}]\big)\in\prod_{i=1}^{len(\mathtt{t}(r))}D_{\mathtt{t}_{i}(r)}\colon r(x_{1},...,x_{n})\in w\right\} $,
for all $r\in\mathtt{R}$;
\item $I(f,w)=\left\{ \big([x_{1}],...,[x_{n}]\big)\in\prod_{i=1}^{len(\mathtt{t}(f))}D_{\mathtt{t}_{i}(f)}\colon\left(f(x_{1},...,x_{n-1})=x_{n}\right)\in w\right\} $,
for all $f\in\mathtt{F}$;
\item $I(c,w)=\left\{ \big([x]\big)\in D_{\mathtt{t}(c)}\colon\left(c=x\right)\in w\right\} $,
for all $c\in\mathtt{C}$.
\end{enumerate}
\end{enumerate}
The \textbf{\textit{canonical valuation}} $v$ for $(\Lambda,\Gamma^{*})$
is given by $v(x)=[x]$ for all $x\in\mathtt{V}^{+}.$
\end{defn}

\paragraph{Lemmas: Uniformity, Existence and Truth}

The canonical model for $(\Lambda,\Gamma^{*})$ is a model for $\mathcal{L}$.
Notably, the domain is well-defined by the identity theory sharing
requirement and a two-partition by the inclusion of the DD axiom.
Further, $I(c,w)$ is well-defined as for every world $w$, there
exists some $x\in\mathtt{V}^{+}$ for which $(c=x)\in w$. See \cite{Achen2017}
for details. Foremost, the map $R$ is well-defined, as is ensured
by the following lemma:
\begin{lem}
[Uniformity]\label{lem.uniformity} Let $K_{x}\varphi\in w\in W$
with $v(x)=\alpha$. Then for all $y\in\mathtt{V}^{+}$ for which
$v(x)=v(y)$, also $K_{y}\varphi\in w$.
\end{lem}
\begin{proof}
Assume \textit{$K_{x}\varphi\in w\in W$ }with $v(x)=\alpha$, and let
$v(x)=v(y)$. Then $[x]=[y]$, so by identity theory sharing assumption,
$(x=y)\in w'$ for every \textit{$w'\in W$; }in particular, $(x=y)\in w$.
By PS, $(x=y)\rightarrow\big(K_{x}\varphi\rightarrow K_{y}\varphi\big)\in w$.
By MP, $\big(K_{x}\varphi\rightarrow K_{y}\varphi\big)\in w$ and
by MP again, $K_{y}\varphi\in w$.
\end{proof}
As in the propositional case, the proof of the Truth Lemma below relies
on the below Existence Lemma. A proof for standard first-order modal
logic may be found in \cite{Hughes_&_Cresswell_New.Intro}; details
for term-modal logic may be found in \cite{Achen2017}.
\begin{lem}
[Existence]\label{lem.existence}If $w\in W$ and $\neg K_{x}\varphi\in w$,
then there exists a $w'\in W$ such that $(w,w')\in R_{\left\llbracket x\right\rrbracket _{w}^{I,v}}$
and $\varphi\in w'$.
\end{lem}
\begin{lem}
[Truth] For all $\varphi\in\mathcal{L}^{+}$, for all $w\in W$,
and for the canonical $v$, $M_{(\Lambda,\Gamma^{*})},w\vDash_{v}\varphi$
iff $\varphi\in w$.
\end{lem}
\begin{proof}
The proof proceeds by induction on the complexity of $\varphi$. For
the quantified formulas, appeal is made to $w$ bearing witnesses.
The negated modal case relies on the Existence Lemma. See \cite{Achen2017}
for full details.
\end{proof}

\paragraph{Canonical Class Theorem}

The canonical models defined facilitate the application of Proposition
\ref{prop.iff} to conclude strong completeness of $\Lambda$ with
respect to its \textit{canonical class}:
\begin{defn}
\label{Def. Canonical class} The \textbf{\textit{canonical class}}
of models for the normal term-modal logic $\Lambda$ is the set $\boldsymbol{C}_{\Lambda}$
of canonical models $M_{(\Lambda,\Gamma^{*})}$ for $\Lambda$-consistent
$\Gamma\subseteq\mathcal{L}$.
\end{defn}
\Canonical* 
\begin{proof}
By Proposition \ref{prop.iff}, it suffices to find for each $\Lambda$-consistent
set $\Gamma$ some $s\in\boldsymbol{C}_{\Lambda}$ that satisfies
$\Gamma$. One such is $(M_{(\Lambda,\Gamma^{*})},\Gamma^{*})$, which
exists by the Lindenbaum and Witnessed Lemmas. As $\Gamma\subseteq\Gamma^{*}$,
the Truth Lemma ensure that $(M_{(\Lambda,\Gamma^{*})},\Gamma^{*})\vDash_{v}\Gamma$
for $v$ the canonical valuation.
\end{proof}
\Completeness* 
\begin{proof}
A frame $F\in\boldsymbol{F}$ that satisfies the $\mathsf{K}$-consistent
set $\Gamma$ is the frame of the canonical model $M_{(\mathsf{K},\Gamma^{*})}$:
$\Gamma$ is satisfied at $\Gamma^{*}$ under the canonical valuation.
\end{proof}

\subsubsection{\label{A.frame.char.}Frame Characterization Proofs}

For illustrative purposes, we show two of the claims made in Table
\ref{tab:char}, Section \ref{subsec:frame.char}.
\begin{prop}
For $\varphi\in\mathcal{L}$, $\forall x(K_{x}\varphi\rightarrow K_{x}K_{x}\varphi)$
is valid on the frame $F=(D,W,R)$ if, and only if, $R(\alpha)$ is
transitive for every $\alpha\in D_{\mathtt{agt}}$.
\end{prop}
\begin{proof}
$\Leftarrow:$ Let $M$ be build on the frame $F$ in which $R_{\alpha}$ is transitive for all $\alpha\in D_{\mathtt{agt}}$. Let $v$ be an arbitrary valuation and assume $M,w\vDash_{v}K_{x}\varphi$. Then $M,w'\vDash_{v}\varphi$ for all $w'\in R_{v(x)}(w)$. For a contradiction, assume $M,w\vDash_{v}\neg K_{x}K_{x}\varphi$. Then there exists a $w^{*}\in R_{v(x)}(w)$ such that $M,w^{*}\vDash_{v}\neg K_{x}\varphi$, and hence there exists a $w^{**}\in R_{v(x)}(w^{*})$ such that $M,w^{**}\vDash_{v}\neg\varphi$. But $R_{v(x)}$ is transitive, so $w^{**}\in R_{v(x)}(w)$. Hence $w^{**}$ satisfies both $\varphi$ and $\neg\varphi$. On pain of contradiction, $M,w\vDash_{v}K_{x}K_{x}\varphi$. As $v$ was arbitrary, $M,w\vDash_{v}\forall x(K_{x}\varphi\rightarrow K_{x}K_{x}\varphi)$. $\Rightarrow:$ By contraposition.
\end{proof}

\begin{prop}\label{prop.N}
The formula $\exists x_{1},...,x_{n}\left(\left(\bigwedge_{i\leq n}K_{x_{i}}\top\right)\wedge\left(\bigwedge_{i,j\leq n,i\neq j}x_{i}\neq x_{j}\right)\wedge\forall y\left(K_{y}\top\rightarrow\bigvee_{i\leq n}y=x_{i}\right)\right)$ is valid on the frame $F=(W,D,R)$ if, and only if, $|D_{\mathtt{agt}}|=n$.
\end{prop}
\begin{proof}
Notice first that the formula, call it $N$, is only well-formed iff the variables $x_1,...,x_n ,y$ are all of type $\mathtt{agt}$, ensured by them appearing as modal operator subscripts. This ensures that the quantifications range only over $D_{\mathtt{agt}}$.

$\Leftarrow:$ Assume given a pointed model $(M,w)$ build on a frame $F=(W,D,R)$ with $|D_{\mathtt{agt}}|=n$. Assume $D_{\mathtt{agt}}$ enumerated such that $D_{\mathtt{agt}}=\{\alpha_1,...,\alpha_n\}$. Let $v$ be an arbitrary valuation. We argue that $M,w\vDash_v N$. Let $v'$ be the valuation identical to $v$ on all points, except for each $i\leq n$, $v'(x_i)=\alpha_i$. Then $M,w\vDash_{v'} \left(\left(\bigwedge_{i\leq n}K_{x_{i}}\top\right)\wedge\left(\bigwedge_{i,j\leq n,i\neq j}x_{i}\neq x_{j}\right)\wedge\forall y\left(K_{y}\top\rightarrow\bigvee_{i\leq n}y=x_{i}\right)\right)$, as it satisfies each conjunct: First, $M,w\vDash_{v'} \bigwedge_{i\leq n}K_{x_{i}}\top$, trivially. Second, $M,w\vDash_{v'} \bigwedge_{i,j\leq n,i\neq j}x_{i}\neq x_{j}$ as $v'(x_i)\neq v'(x_j)$ for all $i,j\leq n$, $i\neq j$, by construction of $v'$. Third and finally, $M,w\vDash_{v'} \forall y\left(K_{y}\top\rightarrow\bigvee_{i\leq n}y=x_{i}\right)$: as $y$ is of type $\mathtt{agt}$, for any $y$-variant $v''$ of $v'$, $v''\in D_{\mathtt{agt}}$, but then $v''(y)=v''(x_i)$ for some $i\leq n$, by construction of $v'$, satisfying the antecedent.

$\Rightarrow:$ Assume given a pointed model $(M,w)$ build on a frame $F=(W,D,R)$ with $|D_{\mathtt{agt}}|\neq n$. Let $v$ be an arbitrary valuation. We argue that not $M,w\vDash_v N$, as $(M,w)$ will falsify either the second or the third conjunct (the first conjunct is satisfied: as each variable $x_i$ is of type $\mathtt{agt}$ for all $i\leq n$, each $K_{x_i}\top$ is satisfied at $w$ under any valuation). If $|D_{\mathtt{agt}}|<n$, then under any valuation $v$, $(M,w)$ will falsify the second conjunct: as each variable $x_i$ is of type $\mathtt{agt}$, $v(x_i)\in D_{\mathtt{agt}}$ for all $i\in {1,...,n}$. But then $v(x_i)=v(x_j)$ for at least two $i,j\leq n, i\neq j$. But then $M,w\vDash_v x_i = x_j$, contrary to the second conjunct. If $|D_{\mathtt{agt}}|>n$, then under any valuation $v$, $(M,w)$ will falsify the third conjunct, as there exists a $y$-variant $v'$ of $v$ such that $v'(y)\neq v(x_i)$ for any $i\leq n$. The existence of this $y$-variant $v'$ is ensured by $|D_{\mathtt{agt}}|>n$, which implies that $D_{\mathtt{agt}}/{v(x_i)\colon i\leq n}\neq\emptyset$, so that we can assume $v'{y}\in D_{\mathtt{agt}}/{v(x_i)\colon i\leq n}$. 
\end{proof}

\subsubsection{\label{A.decidability}Decidability}

\begin{prop}
Let $\mathsf{K}_{n/m}$, $\mathsf{K}_{n}$ and $\mathsf{K}$ be given
in $\mathcal{L}$, based on the signature $\Sigma$. Let $\mathcal{L}_{\mathtt{agt}}\subseteq\mathcal{L}$
contain all formulas containing only agent-terms, $t\in\mathtt{t}^{-1}(\mathtt{agt})$.
\begin{enumerate}
\item For all $\varphi\in\mathcal{L}$, it is decidable whether $\vdash_{\mathsf{K}_{n/m}}\varphi$ or not.
\item 
    \textit{a)} For all $\varphi\in\mathcal{L}_{\mathtt{agt}}$, it is decidable whether $\vdash_{\mathsf{K}_{n}}\varphi$ or not. 
    \textit{b)} In general, $\vdash_{\mathsf{K}_{n}}\varphi$ is undecidable.
\item In general, $\vdash_{\mathsf{K}}\varphi$ is undecidable.
\end{enumerate}
\end{prop}
\begin{proof}
\textbf{1.} $\mathsf{K}_{n/m}$ is sound and complete w.r.t. $\boldsymbol{F}_{n/m}$. To check the validity of any $\varphi\in\mathcal{L}$ over $\boldsymbol{F}_{n/m}$ is a finite procedure: Up to isomorphism, all $F\in\boldsymbol{F}_{n/m}$ share domain $D=D_{\mathtt{agt}}\dot{\cup}D_{\mathtt{obj}}$, $|D_{\mathtt{agt}}|=n$, $|D_{\mathtt{obj}}|=(m-n)$. There are finitely many non-logical symbols in $\varphi$; symbols not in $\varphi$ are irrelevant to its satisfaction. With $D$ fixed, any $w\in F$ will be assigned one of finitely many extensions of $\varphi$'s non-logical symbols: thus, the maximal set of distinct $\varphi$-relevant worlds $W_{\varphi}$ is finite. As $\varphi$ has modal depth $k$, whether $M,w\vDash_{v}\varphi$ depends on at most all worlds within $k$ steps from $w$. Checking whether $M,w\vDash_{v}\varphi$ is thus a finite procedure for all formulas given the finiteness of $D$. Finally, up to bisimulation, the set of graphs over $W_{\varphi}$ and $\{R(\alpha),\alpha\in\text{\ensuremath{D_{\mathtt{agt}}}}\}$ with maximal path length $k$ is finite: hence, the set of needed to be checked pointed models is finite. 
\textbf{2a.} For any $\varphi\in\mathcal{L}_{\mathtt{agt}}$, $\varphi$ is a theorem of $\mathsf{K}_{n}$ iff it is a theorem of $\mathsf{K}_{n/m}$, for any $m>n$. For such $\varphi$, to determine whether $\vdash_{\mathsf{K}_{n}}\varphi$, we can thus check whether $\vdash_{\mathsf{K}_{n/n+1}}\varphi$, which is decidable by 1.
\textbf{2b} and \textbf{3.} General undecidability for $\mathsf{K}_{n}$ and $\mathsf{K}$ follows as both contain unrestricted first-order logic for the arbitrary object domain.
\end{proof}

\subsection{\label{subsec:comp_dyn} Dynamic Term-Modal Logic: Completeness through Translation}

The completeness proof for the dynamic logic $\mathsf{K+AM}$ is based on a reduction argument. The argument relies on the existence of so-called reduction axioms for the dynamic language $\mathcal{L}_{AM}$. The axioms used for this specific proof are listed in Table \ref{tab:axioms_dyn} and can be used to translate every formula from the dynamic language $\mathcal{L}_{AM}$ into a provably equivalent $\mathcal{L}$-formula. Given this translation, the completeness of the dynamic logic follows from the known completeness of the static logic $\mathsf{K}$, established in Corollary \ref{Cor. Complete ALL-1}. The building blocks of the specific reduction argument required to prove completeness for $\mathsf{K+AM}$ are provided below.

First, we provide a translation that by finite iterative application to any formula in the dynamic language $\mathcal{L}_{AM}$ results in a formula from the static language $\mathcal{L}$. The translation is left-to-right: a formula occurring on the left is translated to the formula on the right.

\begin{defn} The \textbf{\textit{translation}} $\tau : \mathcal{L}_{AM} \to \mathcal{L}_{AM}$
is defined as follows:
\begin{flalign*}
      \tau((t_1=t_2)) & = (t_1 = t_2) & \\
      \tau(r(t_{1},...,t_{n})) & =  r(t_{1},...,t_{n}) & \\
      \tau(\neg \varphi) & =  \neg \tau(\varphi)  & \\ 
      \tau(\varphi \wedge \psi) & = \tau(\varphi) \wedge \tau(\psi)  & \\
      \tau(K_{t}\varphi) & = K_{t} \tau(\varphi)  & \\
      \tau(\forall x \varphi) & =  \forall x \tau(\varphi) & \\ 
     \tau([A,e]r(t_{1},...,t_{n})) & =  \tau(\mathsf{pre}(e) \to \post^A(e)(r(t_{1},...,t_{n}))) & \\
      \tau([A,e]\neg \varphi) & =  \tau(\mathsf{pre}(e) \to \neg [A,e]\varphi)  & \\ 
      \tau([A,e](\varphi \wedge \psi)) & =  \tau([A,e]\varphi \wedge [A,e]\psi) & \\
      \tau([A,e]K_{t}\varphi) & =  \tau(\mathsf{pre}(e) \to \bigwedge_{e'\in E^A}(Q(e,e')(x^\star \mapsto t) K_{t} [A,e'] \varphi)) & \\ 
     \tau([A,e]\forall x \varphi) & =  \tau(\mathsf{pre}(e) \to \forall x [A,e]\varphi) & \\
      \tau([A,e][A',e'] \varphi) & =  \tau([A,e \circ A',e'] \varphi) 
\end{flalign*}

\end{defn}

Next, we adapt the formula complexity function introduced by \cite{ditmarsch2007}. 

\begin{defn}
\label{Def. Complexity ordering} The \textbf{\textit{complexity}} $c: \mathcal{L}_{AM} \to \mathbb{N}$ is defined as follows, where $\mathtt{GA}(\mathcal{L})$ abbreviates $\mathtt{GroundAtoms}(\mathcal{L})$:
\begin{align*}
    c(r(t_{1},...,t_{n})) & = 1 &\\
    c(\neg \varphi) & = 1 + c(\varphi) & \\
    c(\varphi \wedge \varphi') & = 1 + \max(c(\varphi),c(\varphi')) & \\
    c(K_t \varphi) & = 1 + c(\varphi) & \\
    c(\forall x\varphi) & = 1 + c(\varphi) & \\
    c([A,e]\varphi) & = (4 + c(A))\cdot c(\varphi) & \\
    c(A) & = \max\left(\bigcup_{e,e'\in E, r(t_1,\dots,t_n)\in \mathtt{GA}(\mathcal{L}) }\{c(\mathsf{pre}^A(e)) \}\cup \{c(\mathsf{post}^A(e)(r(t_1,\dots,t_n))\} \cup \{c(Q(e,e'))\} \right)
\end{align*} 
\end{defn}

A standard ordering lemma ensures that the right side of a given reduction axiom is indeed less complex than the left side. 

\begin{lem} For all $\varphi$, $\psi$ and $\chi$:
\label{lemma:complexity}
\begin{enumerate}
    \item $c(\psi)\geq c(\varphi)$ if $\varphi\in Sub(\psi)$ (where $Sub(\psi)$ is the set of subformulas of $\psi$)
    \item $c([A,e] r(t_{1},...,t_{n})) > c(\mathsf{pre}(e)\to \mathsf{post}(e)(r(t_{1},...,t_{n})))$
    \item $c([A,e]\neg \varphi) > c(\mathsf{pre}(e) \to \neg [A,e]\varphi))$ 
    \item $c([A,e](\varphi \wedge \psi)) > c(([A,e]\varphi) \wedge ([A,e]\psi))$
    \item $c([A,e]K_{t}\varphi) > c(\mathsf{pre}(e) \to \bigwedge_{e'\in E}(Q(e,e')(x^\star \mapsto t) K_{t} [A,e'] \varphi))$ 
    \item $c([A,e]\forall x \varphi) > c(\mathsf{pre}(e) \to \forall x [A,e]\varphi)$ 
    \item $c([A,e][A',e']\varphi) > c([A,e\circ A',e']\varphi)$ 
\end{enumerate}

\end{lem}

\begin{proof}
The proofs are straightforward, along the lines of those provided in \cite[Chapter 7]{ditmarsch2007}. 
\end{proof}

The complexity function $c$ induces an ordering of $\mathcal{L}_{AM}$ formulas which is used to prove the following Lemma, stating that the two sides of a reduction axiom are indeed provably equivalent.

\begin{lem} For all $\varphi \in \mathcal{L}_{AM}$: $\vdash_{\mathsf{K+AM}}\varphi \leftrightarrow \tau(\varphi)$.
\label{lemma:induction syntactic equivalence}
\end{lem}

\begin{proof}
The proof is by induction on the complexity $c(\varphi)$. It is similar to the one provided in  \cite[Chapter 7]{ditmarsch2007}. 
\end{proof}

The completeness of  $\mathsf{K+AM}$ (Corollary \ref{cor:comp_dyn}) follows from the soundness of the dynamic proof system, Lemma \ref{lemma:induction syntactic equivalence} and the completeness of the static sub-system (Corollary \ref{Cor. Complete ALL-1}). The argument, which is standard, is as follows.

\begin{prop} $\vDash \varphi$ implies $\vdash_{\mathsf{K+AM}}\varphi$, for all $\varphi \in \mathcal{L}_{AM}$.
\end{prop}

\begin{proof}
Suppose $\vDash \varphi$. Since  $\vdash_{\mathsf{K+AM}}\varphi \leftrightarrow \tau(\varphi)$ (Lemma \ref{lemma:induction syntactic equivalence}), we have $\vDash\varphi \leftrightarrow \tau(\varphi)$ by the soundness of the proof system $\mathsf{K+AM}$. Thus $\vDash \tau(\varphi)$. The formula $\tau(\varphi)$ does not contain any action model modalities. Given $\vDash \tau(\varphi)$, by the completeness of $\mathsf{K}$ (Corollary \ref{Cor. Complete ALL-1}), it follows that $\vdash_{\mathsf{K}}\tau(\varphi)$. As $\mathsf{K}$ is a subsystem of $\mathsf{K+AM}$, we thus have $\vdash_{\mathsf{K+AM}}\tau(\varphi)$. Since  $\vdash_{\mathsf{K+AM}}\varphi \leftrightarrow \tau(\varphi)$ and $\vdash_{\mathsf{K+AM}}\tau(\varphi)$, it follows that  $\vdash_{\mathsf{K+AM}}\varphi$.
\end{proof}

The completeness result for any system extending $\mathsf{K+AM}$ with frame-characterizing axioms follows from the same type of argument.

\DeclareRobustCommand{\VAN}[2]{#2}

\end{document}